\documentclass[11pt,oneside,fleqn]{article}

\usepackage[ansinew]{inputenc}
\usepackage[mathscr]{eucal}
\usepackage{amsmath,amssymb,amsthm}

\allowdisplaybreaks
\newcommand{\dd}[2]{\frac{\mathrm{d} #1}{\mathrm{d} #2}}

\newcommand{\p}{\partial}
\newcommand{\const}{\mathop{\rm const}\nolimits}

\newcommand{\todo}[1][\null]{\ensuremath{\clubsuit}}

\newcommand{\noprint}[1]{}
\newcommand{\checked}[1][\null]{\ensuremath{\boldsymbol{\surd}}}

\newtheorem{theorem}{Theorem}
\newtheorem{lemma}{Lemma}
\newtheorem{corollary}{Corollary}
\newtheorem{proposition}{Proposition}
{\theoremstyle{definition}
\newtheorem{definition}{Definition}

\newtheorem{remark}{Remark}
\newtheorem*{remark*}{Remark}
}

\newcommand{\DD}{\mathrm{D}}

\newcommand{\vv}{\mathbf{v}}
\newcommand{\ww}{\mathbf{w}}

\newcommand{\ve}{\varepsilon}

\newcommand{\ad}{\mathrm{ad}}
\newcommand{\FF}{\mathcal{F}}
\newcommand{\GG}{\mathcal{G}}
\newcommand{\PP}{\mathcal{P}}

\newcommand{\DDD}{\mathcal{D}}

\begin{document}

\par\noindent {\LARGE\bf
Complete group classification of a class\\ of nonlinear wave equations
\par}

{\vspace{4mm}\par\noindent {\bf Alexander Bihlo$^\dag$, Elsa Dos Santos Cardoso--Bihlo$^\dag$ and Roman O. Popovych$^\dag\, ^\ddag$
} \par\vspace{2mm}\par}

{\vspace{2mm}\par\noindent {\it
$^{\dag}$~Faculty of Mathematics, University of Vienna, Nordbergstra{\ss}e 15, A-1090 Vienna, Austria\\
}}
{\noindent \vspace{2mm}{\it
$\phantom{^\dag}$~\textup{E-mail}: alexander.bihlo@univie.ac.at, elsa.cardoso@univie.ac.at
}\par}

{\vspace{2mm}\par\noindent {\it
$^\ddag$~Institute of Mathematics of NAS of Ukraine, 3 Tereshchenkivska Str., 01601 Kyiv, Ukraine\\
}}
{\noindent \vspace{2mm}{\it
$\phantom{^\dag}$~\textup{E-mail}: rop@imath.kiev.ua
}\par}

\vspace{2mm}\par\noindent\hspace*{8mm}\parbox{140mm}{\small
Preliminary group classification became prominent as an approach to symmetry analysis of differential equations due to the paper by Ibragimov, Torrisi and Valenti [J.\ Math.\ Phys.\ {\bf 32}(11), 2988--2995] 
in which partial preliminary group classification of a class of nonlinear wave equations was carried out 
via the classification of one-dimensional Lie symmetry extensions related to a fixed finite-dimensional subalgebra of the infinite-dimensional equivalence algebra of the class under consideration.
In the present paper we implement, up to both usual and general point equivalence, the complete group classification of the same class using the algebraic method of group classification. 
This includes the complete preliminary group classification of the class and finding Lie symmetry extensions 
which are not associated with subalgebras of the equivalence algebra.
The complete preliminary group classification is based on listing all inequivalent subalgebras of the whole infinite-dimensional equivalence algebra whose projections are qualified as maximal extensions of the kernel algebra.
The set of admissible point transformations of the class is exhaustively described in terms of the partition of the class into normalized subclasses. 
A version of the algebraic method for finding the complete equivalence groups of a general class of differential equations is proposed. 
}\par\vspace{2mm}

\section{Introduction}

The method of preliminary group classification was first introduced in Ref.~\cite{akha91Ay} and became well-known due to Ref.r~\cite{ibra91Ay}. 
In the latter paper, partial preliminary group classification was carried out for the class of equations of the form
\begin{equation}\label{eq:IbragimovClass}
u_{tt}=f(x,u_x)u_{xx}+g(x,u_x),
\end{equation}
where $f\ne0$.
The essence of the approach applied in~\cite{ibra91Ay} is given by the classification of one-dimensional extensions of the kernel algebra with respect to a fixed finite-dimensional subalgebra of the infinite-dimensional equivalence algebra of the class studied. 

The symmetry analysis of the same class was continued in a number of papers. 
The interest to such studies is stimulated because equations of the form~\eqref{eq:IbragimovClass} are used as mathematical models of different continuous media.
They arise, e.g., in the theory of elasticity, in particular, in the course of modeling of hyperelastic homogeneous rods~\cite{jeff82Ay}. 

Thus, in~\cite{hari93Ay} the partial preliminary group classification of the class~\eqref{eq:IbragimovClass} with respect to one-dimensional subalgebras of an infinite-dimensional subalgebra of the equivalence algebra was considered. Second order differential invariants of the equivalence algebra were computed in~\cite{ibra04Ay}. Another direction of investigation for the class~\eqref{eq:IbragimovClass} was initiated in~\cite{ibra00Ay}. Instead of equations of the form~\eqref{eq:IbragimovClass}, related systems of two equations, where the first derivatives of $u$ play the role of the dependent variables, were considered and mapped to the form $v_t = a(x,v)w_x$ and $w_t = b(x,v)v_x$. For the class of systems of this form, certain properties were investigated within the framework of symmetry analysis, including the computation of the equivalence and kernel algebras and the compatibility analysis of the determining equations for Lie symmetries. Upper bounds for the dimension of Lie symmetry extensions were established for the two cases which arose. This study was completed in~\cite{khab09Ay} via exhaustive group classification of such systems by the algebraic method.

A comprehensive review of the literature on group analysis of different classes of $(1+1)$-dimensional wave equations was presented in~\cite{lahn06Ay}. 
Some of these classes are contained in the class~\eqref{eq:IbragimovClass} or nontrivially intersect it. 
In particular, the simple subclasses of~\eqref{eq:IbragimovClass} singled out by the constraints $f_x=g=0$ and $f_x=g_x=0$ 
were considered in~\cite{oron86Ay} and~\cite{gand04Ay}, respectively.   
The class~\eqref{eq:IbragimovClass} has also a subclass common with 
the class of nonlinear wave equations of the general form $u_{tt}=u_{xx}+F(t,x,u,u_x)$, 
whose Lie symmetries were exhaustively investigated by the algebraic method in~\cite{lahn06Ay}.
The intersection obviously consists of equations of the form~\eqref{eq:IbragimovClass} with $f=1$.
Any equation of the form~\eqref{eq:IbragimovClass} is a potential equation for the wave equation of another form, 
$v_{tt}=(f(x,v)v_x+g(x,v))_x$, also called the nonlinear telegraph equation~\cite{blum10Ay,huan07Ay}.

Following the paper~\cite{ibra91Ay}, several classes of differential equations were investigated within the framework of preliminary group classification. Given a class of differential equations, this approach in its essence rests on computing optimal lists of inequivalent subalgebras of the associated equivalence algebra and studying the Lie symmetry extensions induced by these subalgebras. While in the majority of papers on this subject, including Ref.~\cite{ibra91Ay}, only symmetry extensions by means of inequivalent subalgebras of a fixed finite-dimensional subalgebra of a possibly infinite-dimensional equivalence algebra are considered, we have shown in~\cite{card11Ay} that this restriction is in fact not necessary. Stated in another way, there is no obstacle in studying extensions induced by subalgebras of the whole (infinite-dimensional) equivalence algebra. This is, what we have called the complete preliminary group classification in opposite to the various partial preliminary group classifications, which were carried out e.g.\ in~\cite{ibra91Ay,song09Ay}.
As an example, in~\cite{card11Ay} we have solved the complete preliminary group classification problem for the class of nonlinear diffusion equations of the general form $ u_t = f(x,u)u_x^2+g(x,u)u_{xx}$. 

Moreover, in case when the class is normalized (at least in the weak sense~\cite{popo10Cy,popo10Ay}) the same approach gives at once the complete group classification, cf.\ Section~\ref{sec:AlgMethodOfGroupClassification}. This fact was implicitly used in various instances. The most classical examples for this finding are Lie's classifications of second order ordinary differential equations~\cite{lie91Ay} and of second order two-dimensional linear partial differential equations~\cite{lie81Ay}. For numerous modern examples see, e.g.~\cite{basa01Ay,lagn02Ay,lahn06Ay,lahn07Ay,popo04By,wint92Ay,zhda99Ay} and references therein. The technique of group classification explicitly based on the notion of normalized classes of differential equations was developed in~\cite{popo06Ay,popo10Cy,popo10Ay} and then applied to different classes of Schr\"odinger equations, generalized vorticity equations, generalized Korteweg--de Vries equations, etc. 
All the above techniques can be interpreted as particular versions of the algebraic method. 

The purpose of the present paper is to systematically carry out the preliminary group classification of the class of differential equations~\eqref{eq:IbragimovClass} in a similar fashion as in~\cite{card11Ay} and thereby to exhaustively solve the complete group classification problem for this class of nonlinear wave equations using the partition into normalized subclasses. 
The version of the algebraic method applied in the present paper differs from the Lahno--Zhdanov approach~\cite{basa01Ay,lagn02Ay,lahn06Ay,lahn07Ay,zhda99Ay} 
as it does not involve the classification of low-dimensional Lie algebras but is rather based on classifications of all appropriate subalgebras of the corresponding equivalence algebra. 
Note that the consideration is local throughout the paper.

In order to guarantee nonlinearity of equations of the form~\eqref{eq:IbragimovClass}, we explicitly include the nonvanishing condition 
\[(f_{u_x},g_{u_xu_x})\ne(0,0)\] 
into the definition of the class to be studied. 
The reason why we are only concerned with the nonlinear case here is that nonlinear and linear equations of the form~\eqref{eq:IbragimovClass} 
are not mixed by point transformations (cf.\ Remark~\ref{rem:OnInequivOfLinAndNonlinCasesOfIbragimovClass}) 
and have quite different Lie symmetry properties. 
Moreover, linear wave equations of the form~\eqref{eq:IbragimovClass} were already well investigated within the framework of classical symmetry analysis in~\cite{blum89Ay,ovsj62Ay}.

The further organization of this paper is the following: 
Theoretical background of point transformations in classes of differential equations is reviewed in Section~\ref{sec:PointTransInClassesOfDiffEqs}.
This includes the definitions and properties of a class of differential equations, its subclasses, the set of admissible transformations, the usual equivalence group and algebra, 
different notions of normalized classes of differential equations, etc. 
Section~\ref{sec:AlgMethodOfGroupClassification} contains a concise description of the group classification problem together with a discussion on the theory of preliminary group classification 
and complete group classification with the algebraic method. 
Group analysis of the class~\eqref{eq:IbragimovClass} is started in Section~\ref{sec:EquivalenceAlgebraIbragimovClass} by studying the structure of the equivalence algebra of~\eqref{eq:IbragimovClass}. 
The computation of the determining equations for admissible transformations and the equivalence group of the class~\eqref{eq:IbragimovClass} 
by the direct method is given in Section~\ref{sec:PreliminaryStudyOfAdmTrans} and Section~\ref{sec:PointTransInClassesOfDiffEqs}, respectively. 
The algebraic method for the calculation of equivalence groups is first presented and then applied 
to the class~\eqref{eq:IbragimovClass} in Section~\ref{sec:CalculationOfEquivAlgebraByAlgebraicMethod}.
Analysis of the determining equations for Lie symmetries of equations from the class~\eqref{eq:IbragimovClass} is presented in Section~\ref{sec:DetEqsForLieSymsIbragrimovClass}. 
It gives the kernel algebra of this class and allows us to prove that the major subclass of the class~\eqref{eq:IbragimovClass} is weakly normalized with respect to the equivalence algebra of~\eqref{eq:IbragimovClass}. 
The group classification of the complement of the subclass is also carried out. 
Completing the study of admissible transformations by the direct method, in Section~\ref{sec:SetOfAdmTrans} we partition the class~\eqref{eq:IbragimovClass} into 
two subclasses, which are respectively normalized and semi-normalized with respect to the equivalence group of the entire class~\eqref{eq:IbragimovClass}.
In this way we prove that the class~\eqref{eq:IbragimovClass} is semi-normalized. 
Both the equivalence algebra and the equivalence group are used in Section~\ref{sec:ClassificationSubalgebrasIbragimovClass} to classify subalgebras of the equivalence algebra that may be used for preliminary group classification. The adjoint action of the equivalence group on the associated algebra is computed using push-forwards of vector fields, as it was recently proposed in~\cite{card11Ay}. 
The final calculations related to the complete group classification in the class~\eqref{eq:IbragimovClass} and the corresponding list of inequivalent Lie symmetry extensions can be found in Section~\ref{sec:GroupClassificationIbragimovClass}. 
In Section~\ref{sec:ConclusionIbragimovClass} we briefly sum up the results of the present paper and 
make comparative analysis of partial preliminary group classification, complete preliminary group classification and complete group classification
within the framework of the general algebraic method.

\section{Point transformations in classes of differential equations}\label{sec:PointTransInClassesOfDiffEqs}

To make this paper self-contained, in this and in the next sections we restate some important notions from the theory of group classification. More information on this subject can be found, e.g.\ in~\cite{lisl92Ay,ovsi82Ay,popo10Cy,popo10Ay}.

The central notion underlying the theory of group classification is an appropriate definition of a class of (systems of) differential equations. In practice, the structure and properties of a class of differential equations determines which methods of group classification (e.g.\ complete vs.\ preliminary, direct vs.\ algebraic) are the most adapted for it. In short, the definition of a class of differential equations comprises two ingredients. 

The first ingredient is a system of differential equations $\mathcal L_\theta$: $L(x,u_{(p)},\theta_{(q)}(x,u_{(p)}))=0$, parameterized by the tuple of arbitrary elements 
$\theta(x,u_{(p)}) = (\theta^1(x,u_{(p)}),\dots,\theta^k(x,u_{(p)}))$, 
where $x=(x^1,\dots, x^n)$ is the tuple of independent variables 
and~$u_{(p)}$ is the set of all dependent variables~$u=(u^1,\dots,u^m)$ together with all derivatives of~$u$ with respect to~$x$ up to the order~$p$. 
The symbol~$\theta_{(q)}$ stands for the set of partial derivatives of $\theta$ of order not greater than $q$ with respect to the variables $x$ and $u_{(p)}$. 

\looseness=-1
The second ingredient 
concerns possible values of the tuple of arbitrary elements~$\theta$. This tuple is required to run through the solution set~$\mathcal S$ of a joint system (also denoted by~$\mathcal S$) of auxiliary differential equations~$S(x,u_{(p)},\theta_{(q')}(x,u_{(p)}))=0$ and inequalities $\Sigma(x,u_{(p)},\theta_{(q')}(x,u_{(p)}))\ne0$, in which both~$x$ and~$u_{(p)}$ play the role of independent variables and $S$ and~$\Sigma$ are tuples of smooth functions depending on~$x$, $u_{(p)}$ and $\theta_{(q')}$. 
The nonvanishing conditions~$\Sigma\ne0$ might be essential to guarantee that each element of the class has some common properties with all other elements of the same class, such as the same order $p$ or the same linearity or nonlinearity properties.  Thereby, these inequalities can be the crucial factor in order to solve the given group classification problem up to a certain stage. In spite of this, they are often omitted without any reason.

\begin{definition}
The set $\{\mathcal L_\theta\mid\theta\in\mathcal S\}$ denoted by~$\mathcal L|_{\mathcal S}$ is called a \emph{class of differential equations}
defined by the parameterized form of systems~$\mathcal L_\theta$ and the set~$\mathcal S$ of the arbitrary elements~$\theta$.
\end{definition}

An additional problem in defining a class of differential equations 
is that the correspondence $\theta\to\mathcal L_\theta$ between arbitrary elements and systems 
(treated not as formal algebraic expressions but as real systems of differential equations or manifolds in the jet space~$J^{(p)}$ 
which is the space of the variables $(x,u_{(p)}))$ may not be injective.
The values $\theta$ and $\tilde\theta$ of arbitrary elements are called \emph{gauge-equivalent}
if $\mathcal L_\theta$ and $\mathcal L_{\tilde\theta}$ are the same system of differential equations, i.e., 
their sets of solutions coincide.
We formally consider $\mathcal L_\theta$ and $\mathcal L_{\tilde\theta}$ as different representations
of the same system from $\mathcal L|_{\mathcal S}$.
For the correspondence $\theta\to\mathcal L_\theta$ to be one-to-one in the presence of a nontrivial gauge equivalence,
the set $\mathcal S$ of arbitrary elements should be factorized with respect to the gauge equivalence relation 
via changing the representation for the class~$\mathcal L|_{\mathcal S}$.
If this is not convenient, the gauge equivalence should be carefully taken into account 
when carrying out symmetry analysis of the class~$\mathcal L|_{\mathcal S}$~\cite{ivan10Ay}. 

In the course of group classification of a complicated class of differential equations, it is often helpful to consider subclasses of this class.
A \emph{subclass} is singled out from the class~$\mathcal L|_{\mathcal S}$ by attaching additional equations or nonvanishing conditions to the auxiliary system~$\mathcal S$. 

Thus, for the class of equations of the general form~\eqref{eq:IbragimovClass} we have the single dependent variable~$u$ of two independent variables~$t$ and~$x$. The associated tuple of arbitrary elements consists of two functions~$f$ and~$g$ whose domains are contained in the related second-order jet space, i.e., in the space of~$t$, $x$ and~$u$ together with all derivatives of~$u$ up to the second order. The indicated dependence of~$f$ and~$g$ only on~$x$ and~$u_x$ means that the arbitrary elements of this class are solutions of the auxiliary system of differential equations
\begin{gather*}
f_t=f_u=f_{u_t}=f_{u_{tt}}=f_{u_{tx}}=f_{u_{xx}}=0,\\
g_t=g_u=g_{u_t}=g_{u_{tt}}=g_{u_{tx}}=g_{u_{xx}}=0.
\end{gather*}
As we talk about wave equations, we should also impose the inequality~$f\ne0$.
In the present paper we study the subclass of equations of the form~\eqref{eq:IbragimovClass} that consists of only really nonlinear equations and, therefore, is singled out from the entire class of equations of the general form~\eqref{eq:IbragimovClass} by the additional nonvanishing condition $(f_{u_x},g_{u_xu_x})\ne(0,0)$. It is the set of equations which is called the class~\eqref{eq:IbragimovClass} throughout the paper.

Several properties hold for subclasses of a class~$\mathcal L|_{\mathcal S}$. 
The intersection of a finite number of subclasses of~$\mathcal L|_{\mathcal S}$ is also a subclass in~$\mathcal L|_{\mathcal S}$, 
which is defined by the union of the additional auxiliary systems associated with the intersecting sets. 
At the same time, the complement $\overline{\mathcal L|_{\mathcal S'\!}}=\mathcal L|_{\overline{\mathcal S'\!}\,}$
of the subclass $\mathcal L|_{\mathcal S'\!}$ in the class $\mathcal L|_{\mathcal S}$ 
is also a subclass of $\mathcal L|_{\mathcal S}$ only in special cases, e.g., if
the additional system of equations or the additional system of nonvanishing conditions is empty (cf.\ Remark~\ref{rem:SubclassAndItsComplementOfIbragimovClass}).
Namely, if the subset $\mathcal S'$ of arbitrary elements is singled out from $\mathcal S$ by the system $S'_1=0$, \dots, $S'_{s'}=0$
then the additional auxiliary condition for $\overline{\mathcal S'}$ is $|S'_1|^2+\dots+|S'_{s'}|^2\ne0$.
If $\mathcal S'$ is defined by the inequalities $\Sigma'_1\ne0$, \dots, $\Sigma'_{\sigma'\!}\ne0$
then the additional auxiliary condition for $\overline{\mathcal S'}$ is $\Sigma'_1\cdots\Sigma'_{\sigma'\!}=0$.

A point transformation in a space is an invertible smooth mapping of an open domain in this space into the same domain. 
Given a class~$\mathcal L|_{\mathcal S}$ of differential equations, point transformations related to~$\mathcal L|_{\mathcal S}$ form different structures. 

Let~$\mathcal L_\theta$ and~$\mathcal L_{\tilde \theta}$ be elements of the class~$\mathcal L|_{\mathcal S}$. 
By~$\mathrm T(\theta,\tilde\theta)$ we denote the set of point transformations in the space of the variables~$(x,u)$ mapping the system~$\mathcal L_\theta$ to the system~$\mathcal L_{\tilde\theta}$. 
In this notation, the maximal point symmetry (pseudo)group $G_\theta$ of the system $\mathcal L_\theta$ coincides with~$T(\theta,\theta)$. 
If $\mathrm{T}(\theta,\tilde\theta)\ne\varnothing$, i.e.\ 
the systems~$\mathcal L_\theta$ and $\mathcal L_{\tilde\theta}$ are similar with respect to point transformations, 
then $\mathrm{T}(\theta,\tilde\theta)=\varphi^0\circ G_\theta=G_{\tilde\theta}\circ\varphi^0$,
where $\varphi^0$ is a fixed transformation from~$\mathrm{T}(\theta,\tilde\theta)$.

\begin{definition}\label{def:SetOfAdmTrans}
The \textit{set of admissible transformations} of the class~$\mathcal L|_{\mathcal S}$ is given by 
$\mathrm T=\mathrm T(\mathcal L|_{\mathcal S}) = \{(\theta,\tilde \theta,\varphi)|\ \theta,\tilde \theta\in\mathcal S, \varphi \in \mathrm T(\theta, \tilde \theta)\}$.
\end{definition}

If the number~$m$ of dependent variables is equal to one, instead of point transformations one can consider more general contact transformations~\cite{ovsi82Ay} in the same way. 

The notion of admissible transformations~\cite{popo06Ay,popo10Ay} is a formalization of the notion of  
form-preserving \cite{king98Ay,king01Ay} or allowed~\cite{wint92Ay} transformations.
First descriptions of the sets of admissible transformations for nontrivial classes of differential equations 
were given by Kingston and Sophocleous~\cite{king91Ay} for a class of generalized Burgers equations 
and by Winternitz and Gazeau~\cite{wint92Ay} for a class of variable-coefficient Korteweg--de Vries equations.
An infinitesimal analogue of the notion of admissible transformations was proposed and studied by Borovskikh in~\cite{boro04Ay}.
In terms of equivalence groups and normalization properties of subclasses 
(see Definitions~\ref{def:NormalizationIbragimov} and~\ref{def:SemiNormalizationIbragimov}), 
the sets of admissible transformations were exhaustively described 
for a number of different classes of differential equations which are important for application, such as
nonlinear Schr\"odinger equations \cite{popo04By,popo10Ay},
variable-coefficient diffusion--reaction equations \cite{vane07Ay,vane09Ay},
generalized Korteweg--de Vries equations including variable-coefficient Korteweg--de Vries 
and modified Korteweg--de Vries equations~\cite{popo10By}, 
systems of $(1+1)$-dimensional second-order evolution equations~\cite{popo08Ay}, 
generalized vorticity equations~\cite{popo10Cy}, 
etc.

\begin{definition}\label{def:UsualEquivGroup}
The \emph{(usual) equivalence group} $G^{\sim}=G^{\sim}(\mathcal L|_{\mathcal S})$ of the class~$\mathcal L|_{\mathcal S}$ 
is the (pseudo)group of point transformations in the space of $(x,u_{(p)},\theta)$ which are projectable to the space of $(x,u_{(p')})$ for any $0\le p'\le p$,
are consistent with the contact structure on the space of $(x,u_{(p)})$ and preserve the set~$\mathcal S$ of arbitrary elements.
\end{definition}

Recall that a point transformation~$\varphi$: $\tilde z=\varphi(z)$ in the space of the variables $z=(z_1,\ldots,z_k)$
is called projectable on the space of the variables $z'=(z_{i_1},\ldots,z_{i_{k'}})$,
where $1\le i_1<\cdots<i_{k'}\le k$,
if the expressions for~$\tilde z'$ depend only on~$z'$.
We denote the restriction of~$\varphi$ to the $z'$-space as $\varphi|_{z'}$: $\tilde z'=\varphi|_{z'}(z')$. 
A point transformation~$\Phi$ in the space of $(x,u_{(p)},\theta)$, which is projectable to the space of $(x,u_{(p')})$ for any $0\le p'\le p$, 
is consistent with the contact structure on the space of $(x,u_{(p)})$ if $\Phi|_{(x,u_{(p)})}$ is the $p$-th order prolongation of $\Phi|_{(x,u)}$.

Each transformation~$\Phi$ from the equivalence group~$G^{\sim}$ (i.e., an equivalence transformation of the class~$\mathcal L|_{\mathcal S}$)
induces the family of admissible transformations of the form $(\theta,\Phi\theta,\Phi|_{(x,u)})$ parameterized by the arbitrary elements~$\theta$ 
running through the entire set~$\mathcal S$. 
Roughly speaking, $G^{\sim}$ is the set of admissible transformations which can be applied to any~$\theta\in\mathcal S$. 

There exist several generalizations of the notion of equivalence group in the literature on symmetry analysis of differential equations, 
in which restrictions for equivalence transformations (projectability or locality with respect to arbitrary elements) are weakened
\cite{ivan10Ay,mele94Ay,popo10Ay,vane07Ay,vane09Ay}.

The common part $G^\cap=G^\cap(\mathcal L|_{\mathcal S})=\bigcap_{\theta\in{\mathcal S}}G_\theta$ of
all $G_\theta$, $\theta\in\mathcal S$, is called
the \emph{kernel of the maximal point symmetry groups} of systems from the class $\mathcal L|_{\mathcal S}$ \cite{ovsi82Ay}. 
The following folklore assertion is true (see, e.g., \cite{card11Ay,popo10Ay}). 

\begin{proposition}\label{pro:OnKernelGroupAsANormalSubgroup}
The kernel group $G^\cap$ of the class~$\mathcal L|_{\mathcal S}$
is naturally embedded into the (usual) equivalence group~$G^{\sim}$ of this class via trivial (identical) prolongation of
the kernel transformations to the arbitrary elements.
The associated subgroup $\hat G^\cap$ of $G^{\sim}$ is normal.
\end{proposition}

Properties of~$G^\cap$ described in Proposition~\ref{pro:OnKernelGroupAsANormalSubgroup} were first noted in different works by Ovsiannikov (see, e.g., \cite{ovsi75Ay} and \cite[Section~II.6.5]{ovsi82Ay}).
Another formulation of this proposition was given in \cite[p.~52]{lisl92Ay}, Proposition~3.3.9. 

As the study of point transformations of differential equations usually involves cumbersome and sophisticated calculations, 
instead of finite point transformations one may consider their infinitesimal counterparts.
This leads to a certain linearization of the related problem which essentially simplifies the whole consideration.
In the framework of the infinitesimal approach, a (pseudo)group~$G$ of point transformations is replaced 
by the Lie algebra~$\mathfrak g$ of vector fields on the same space, which are generators of one-parametric local subgroups of~$G$. 

In particular, the vector fields in the space of $(x,u)$ generating one-parametric subgroups 
of the maximal point symmetry (pseudo)group $G_\theta$ of the system $\mathcal L_\theta$ form a Lie algebra~$\mathfrak g_\theta$
called the \emph{maximal Lie invariance algebra} of the system $\mathcal L_\theta$. 
Analogously to symmetry groups, 
the common part $\mathfrak g^\cap=\mathfrak g^\cap(\mathcal L|_{\mathcal S})=\bigcap_{\theta\in{\mathcal S}}\mathfrak g_\theta$ of
all $\mathfrak g_\theta$, $\theta\in\mathcal S$, is called the \emph{kernel of the maximal Lie invariance algebras} of systems from the class $\mathcal L|_{\mathcal S}$. 
It is the Lie algebra associated with the kernel group~$G^\cap$.

The equivalence algebra~$\mathfrak g^\sim$ is the Lie algebra formed by generators of one-parametric groups of equivalence transformations for the class~$\mathcal L|_{\mathcal S}$. 
These generators are vector fields in the space of $(x,u_{(p)},\theta)$ which are projectable to the space of $(x,u_{(p')})$ for any $0\le p'\le p$ 
and whose projections to the space of $(x,u_{(p)})$ are the $p$-th order prolongations of the corresponding projections to the space of $(x,u)$.

An infinitesimal analogue of Proposition~\ref{pro:OnKernelGroupAsANormalSubgroup} is the following assertion.

\begin{corollary}\label{cor:OnKernelAlgebraAsIdeal1}
The trivial prolongation~$\hat{\mathfrak g}^\cap$ of the kernel algebra~$\mathfrak g^\cap$ to the arbitrary elements is an ideal in the equivalence algebra~$\mathfrak g^\sim$.
\end{corollary}

By definition, any element of the algebra~$\hat{\mathfrak g}^\cap$ formally has the same form as the associated element from~$\mathfrak g^\cap$, 
but in fact is a vector field in the different space augmented with the arbitrary elements.

It is convenient to characterize and estimate transformational properties of classes of differential equations in terms of normalization. 

\begin{definition}\label{def:NormalizationIbragimov}
A class of differential equations~$\mathcal L|_{\mathcal S}$ is \emph{normalized} if its set of admissible transformations is induced by transformations of its equivalence group~$G^\sim$, 
meaning that for any triple $(\theta,\tilde\theta,\varphi)$ from $\mathrm T(\mathcal L|_{\mathcal S})$ there exists a transformation $\Phi$ from~$G^\sim$ 
such that $\tilde\theta =\Phi\theta$ and $\varphi=\Phi|_{(x,u)}$.
\end{definition}

\begin{definition}\label{def:SemiNormalizationIbragimov}
A class of differential equations~$\mathcal L|_{\mathcal S}$ is called \textit{semi-normalized} 
if its set of admissible transformations is induced by transformations from its equivalence group~$G^\sim$ and the maximal point symmetry groups of its equations, 
meaning that for any triple $(\theta,\tilde\theta,\varphi)$ from $\mathrm T(\mathcal L|_{\mathcal S})$ there exist a transformation $\Phi$ from~$G^\sim$ 
and a transformation~$\tilde\varphi$ from the maximal point symmetry group~$G_\theta$ of the system~$\mathcal L_\theta$, 
such that $\tilde\theta =\Phi\theta$ and $\varphi = \Phi|_{(x,u)}\circ \tilde\varphi$.
\end{definition}

In other words, a class of differential equations is semi-normalized if arbitrary similar systems from the class are related via transformations from the equivalence group of this class. 

Normalized and semi-normalized classes of differential equations have a number of interesting properties which essentially simplify the study of such classes. 
In particular, if the class $\mathcal L|_{\mathcal S}$ is normalized in the usual sense, 
its kernel algebra~$\mathfrak g^\cap$ is an ideal of the maximal Lie invariance algebra~$\mathfrak g_\theta$ for each $\theta\in\mathcal S$. 
In general, this claim is not true even if the class is only semi-normalized. See Example~1 in~\cite{card11Ay}.

The above notion of normalization (resp.\ semi-normalization) relies on the finite admissible transformations. 
A weaker version of normalization is defined in infinitesimal terms~\cite{popo10Cy}.

\begin{definition}\label{def:WeakNormalizationIbragimov}
A class of differential equation~$\mathcal L|_{\mathcal S}$ is \textit{weakly normalized} if the union and, therefore, the span of maximal Lie invariance algebras $\mathfrak g_\theta$ of all systems~$\mathcal L_\theta$ from the class is contained in the projection of the equivalence algebra $\mathfrak g^\sim$ of the class to vector fields in the space of independent and dependent variables, i.e.\
\[
 \bigcup_{\theta\in\mathcal S}\mathfrak g_\theta\subset \mathrm P\mathfrak g^\sim\quad  
 (\mbox{or}\quad \langle\mathfrak g_\theta\mid\theta\in\mathcal S\rangle \subset \mathrm P\mathfrak g^\sim).
\]
Here by~$\mathrm P$ we denote the projection operator that acts on vector fields of the general form $Q=\xi^i(x,u)\p_{x_i}+\eta^a(x,u)\p_{u^a}+\varphi^s(x,u,\theta)\p_{\theta^s}$ in the space of variables $x$, $u$ and $\theta$ yielding the vector fields of the form $\mathrm P Q = \xi^i \p_{x_i}+\eta^a\p_{u^a}$, which are defined on the space of variables $x$ and~$u$.
\end{definition}

It is obvious that any normalized class of differential equations is both semi-normalized and weakly normalized.

In general, the normalization of a class of differential equations can be checked by computing the set of admissible transformations of the class and its equivalence group 
(e.g.\ using the direct method) and testing whether the condition from Definition~\ref{def:NormalizationIbragimov} is satisfied. 
It is often convenient to begin with a normalized superclass and construct a hierarchy of normalized subclasses of the superclass or a simple chain of such nested subclasses, 
which contain the class under consideration \cite{popo10Ay,popo08Ay,popo10By}.
The weak normalization property in turn can be verified 
by finding the equivalence algebra of the class and an inspection of the determining equations for Lie symmetries of systems from the class 
(see the next section). 
As the computations related to checking weak normalization involve solving of only linear partial differential equations 
(in contrast to the computations using the direct method of finding equivalence and admissible transformations), 
they can be realized in an algorithmic way even for quite cumbersome classes of multidimensional partial differential equations. 
At the same time, the established presence of the usual normalization property is more useful and allows one to obtain deeper results than involving its weak infinitesimal analogue.

\section{Algebraic method of group classification}\label{sec:AlgMethodOfGroupClassification}

Now that we have introduced necessary notions related to point transformations within classes of differential equations, 
we can go on with the general discussion of the framework of group classification in some more detail. 

The solution of the group classification problem by Lie--Ovsiannikov 
for a class~$\mathcal L|_{\mathcal S}$ of differential equations 
should include the construction of the following elements: 
\begin{itemize}\itemsep=0ex
\item 
the equivalence group~$G^\sim$ of the class~$\mathcal L|_{\mathcal S}$, 
\item 
the kernel algebra~$\mathfrak g^\cap=\mathfrak g^\cap(\mathcal L|_{\mathcal S})=\bigcap_{\theta\in{\mathcal S}}\mathfrak g_\theta$ of the class~$\mathcal L|_{\mathcal S}$, 
i.e., the intersection of the maximal Lie invariance algebras of systems from this class,
\item 
an exhaustive list of $G^\sim$-equivalent extensions of the kernel algebra~$\mathfrak g^\cap$ in the class~$\mathcal L|_{\mathcal S}$, i.e., 
an exhaustive list of  $G^\sim$-equivalent values of~$\theta$ with the corresponding maximal Lie invariance algebras~$\mathfrak g_\theta$ 
for which $\mathfrak g_\theta\ne\mathfrak g^\cap$. 
\end{itemize}
More precisely, the classification list consists of pairs $(\mathcal S_\gamma,\{\mathfrak g_\theta,\theta\in\mathcal S_\gamma\})$, $\gamma\in\Gamma$.
For each $\gamma\in\Gamma$ $\mathcal L|_{\mathcal S_\gamma}$ is a subclass of~$\mathcal L|_{\mathcal S}$, 
$\mathfrak g_\theta\ne \mathfrak g^\cap$ for any $\theta\in\mathcal S_\gamma$ and 
the structures of the algebras $\mathfrak g_\theta$ are similar for all $\theta\in\mathcal S_\gamma$.
In particular, the algebras $\mathfrak g_\theta$, $\theta\in\mathcal S_\gamma$, have the same dimension or
display the same arbitrariness of algebra parameters in the infinite-dimensional case.
Moreover, for any $\theta\in\mathcal S$ with $\mathfrak g_\theta\ne \mathfrak g^\cap$ 
there exists $\gamma\in\Gamma$ such that $\theta\in\mathcal S_\gamma\bmod G^{\sim}$.
All elements from $\bigcup_{\gamma\in\Gamma}\mathcal S_\gamma$ are $G^\sim$-inequivalent.
Note that in all examples of group classification presented in the literature the set~$\Gamma$ was finite. 

The procedure of group classification can be supplemented by deriving auxiliary systems of differential equations
for the arbitrary elements, providing extensions of Lie symmetry, cf.\ Remark~\ref{rem:SubclassAndItsComplementOfIbragimovClass}.
In other words, for each $\gamma\in\Gamma$ one should explicitly describe the subset $\bar{\mathcal S}_\gamma$ of~$\mathcal S$ 
which is the union of $G^{\sim}$-orbits of elements from~${\mathcal S}_\gamma$.
Although this step is usually neglected, it may lead to nontrivial results (see, e.g., \cite{boro06Ay}).

If the class~$\mathcal L|_{\mathcal S}$ is not semi-normalized,
the classification list may include equations similar with respect to
point transformations which do not belong to $G^\sim$.
The knowledge of such \emph{additional} equivalences allows one to substantially simplify the
further symmetry analysis of the class~$\mathcal L|_S$.
Their construction can be considered as one further step of the algorithm of group classification~\cite{ivan10Ay,popo04Ay,vane09Ay}.
Often it can be implemented using empiric tools, e.g., the fact that similar equations have similar maximal invariance algebras.
A more systematical way is to describe the complete set of admissible transformations.

In practice, the \emph{procedure of group classification} within the Lie--Ovsiannikov approach can be realized by implementing a few consecutive steps. 

Given a class~$\mathcal L|_{\mathcal S}$, it is convenient to start the procedure by the computation of the \emph{equivalence algebra}. This can be done either using the infinitesimal method~\cite{akha91Ay,ovsi82Ay} or simply by deriving the set of generators for the one-parametric groups of the equivalence group, provided that the latter is known. Computing the equivalence algebra independently from the equivalence group is important, as it gives a test and a tool for the calculation of the equivalence group. In particular, often only the connected component of unity in the equivalence group is found using the knowledge of the equivalence algebra. The equivalence algebra also plays a distinct role in the course of applying the algebraic method of group classification. 

The most powerful tool for the construction of the \emph{equivalence group}, which is the next step of the procedure, is the direct method involving finite point transformations. Such a construction can be understood as the final stage in the preliminary investigation of the set of admissible transformations of the class~$\mathcal L|_{\mathcal S}$ and allows finding both continuous and discrete equivalence transformations. Due to involving finite point transforms the related calculations are cumbersome and lead to a nonlinear system of partial differential equations. An alternative approach in order to at least restrict the form of point equivalence transformations is based on the condition that any point equivalence transformation induces an automorphism of the equivalence algebra, cf.\ Section~\ref{sec:CalculationOfEquivAlgebraByAlgebraicMethod}.

The \emph{system of determining equations} on the coefficients of Lie symmetry operators of a system~$\mathcal L_\theta$ from the class~$\mathcal L|_{\mathcal S}$ follows from the infinitesimal invariance criterion~\cite{blum89Ay,olve86Ay,ovsi82Ay}, stating that
\[
Q_{(p)} L(x,u_{(p)},\theta_{(q)}(x,u_{(p)}))\big|_{\mathcal L^p_\theta}=0 
\]
holds for any operator $Q=\xi^i(x,u)\p_{x_i}+\eta^a(x,u)\p_{u^a}$ from $\mathfrak g_\theta$, where the arbitrary elements~$\theta$ play the role of parameters.
In what follows we assume the summation for repeated indices.
The indices~$i$ and~$a$ run from~1 to~$n$ and from~1 to~$m$, respectively.
$Q_{(p)}$ denotes the standard $p$-th prolongation of the operator~$Q$, 
\[
Q_{(p)}:=Q+\sum_{0<|\alpha|{}\leqslant  p} 
\Bigl(D_1^{\alpha_1}\ldots D_n^{\alpha_n}\bigl(\eta^a(x,u)-\xi^i(x,u)u^a_i\bigr)+\xi^iu^a_{\alpha,i}\Bigr)\p_{u^a_\alpha}.
\]
$ D_i=\p_i+u^a_{\alpha,i}\p_{u^a_\alpha}$ is the operator of total differentiation with respect to the variable~$x_i$. 
The tuple $\alpha=(\alpha_1,\ldots,\alpha_n)$ is a multi-index,  
$\alpha_i\in\mathbb{N}\cup\{0\}$, $|\alpha|\mbox{:}=\alpha_1+\cdots+\alpha_n$.
The variable $u^a_\alpha$ of the jet space $J^{(p)}$ is identified with the derivative 
$\p^{|\alpha|}u^a/\p x_1^{\alpha_1}\ldots\p x_n^{\alpha_n}$, 
and $u^a_{\alpha,i}:=\p u^a_\alpha/\p x_i$. 
Some determining equations do not involve the arbitrary elements and thus can be integrated immediately. The remaining determining equations explicitly depending on the arbitrary elements are referred to as the classifying equations.

Varying the arbitrary elements~$\theta$, we can split the determining equations with respect to different derivatives of~$\theta$. The additional splitting results in equations for those symmetries that are admitted for any value of the arbitrary elements and form the \textit{kernel of maximal Lie invariance algebras}.

The further analysis of the determining equations is usually much more intricate. The classifying equations are inspected for specific values of the arbitrary elements~$\theta$, which give extensions of the solution sets of the determining equations, associated with symmetry extensions of the kernel algebra. The sets of values found for~$\theta$ should be factorized with respect to the equivalence relation requested. Still, it is the complexity of this analysis that led to the development of a great variety of specialized techniques of group classification, which are conventionally partitioned into two approaches. 

The first method is the direct compatibility analysis and integration of the determining equations up to the equivalence relation, that depends on the values of the arbitrary elements. It is mostly suitable for classes with arbitrary elements that are constants or functions of single arguments. Algorithms of group classification that are realized in present day's computer algebra packages for the calculation of Lie symmetries are based on this method~\cite{ande96Ay,carm00Ay,head93Ay,vu07Ay,witt04Ay}.

The other method is of algebraic nature. It is based on the following two properties: For each fixed value of the arbitrary elements the solution space of the determining equations is associated with a Lie algebra of vector fields. Additionally, if systems of differential equations are similar with respect to a point transformation then its push-forward relates the corresponding maximal Lie invariance algebras. This is why any version of the algebraic method of group classification existing in the literature involves, in some way, the classification of algebras of vector fields up to certain equivalence induced by point transformations. The key question is what set of vector fields should be classified and what kind of equivalence should be used.

It is obvious that for each equation $\mathcal L_\theta$ from the class $\mathcal L|_{\mathcal S}$ its maximal Lie invariance algebra $\mathfrak g_\theta$ is contained in the union $\mathfrak g^\cup =\bigcup_{\theta\in\mathcal S} \mathfrak g_\theta$. The definition of $\mathfrak g^\cup$ implies that this set consists of vector fields for which the system of determining equations is consistent with respect to the arbitrary elements with the auxiliary system of the class~$\mathcal L|_{\mathcal S}$. Therefore, the set~$\mathfrak g^\cup$ can be obtained at the onset of group classification, independently from deriving the maximal Lie invariance algebras of equations from the class~$\mathcal L|_{\mathcal S}$. As it is not convenient to select linear subspaces in the set~$\mathfrak g^\cup$ in the general case, we can extend~$\mathfrak g^\cup$ to its linear span $\mathfrak g^{\langle\rangle}=\langle\mathfrak g_\theta|\theta\in\mathcal S\rangle$, but fortunately we often have $\mathfrak g^\cup=\mathfrak g^{\langle\rangle}$. Via push-forwarding of vector fields, equivalence (resp.\ admissible) point transformations for the class~$\mathcal L|_{\mathcal S}$ induce an equivalence relation on algebras contained in~$\mathfrak g^\cup$. Such an algebra is called appropriate if it is the maximal Lie invariance algebra of an equation from the class~$\mathcal L|_{\mathcal S}$. We should classify, up to the above equivalence relation, only appropriate algebras. They satisfy additional constraints. The simplest restriction for appropriate subalgebras is that each of them contains the kernel algebra~$\mathfrak g^\cap$. The condition that the algebras are really maximal Lie invariance algebras for equations from the class~$\mathcal L|_{\mathcal S}$ is more nontrivial to verify. 

Substituting the basis elements of each appropriate algebra obtained in the course of the algebra classification into the determining equations gives a compatible system for values of the arbitrary elements associated with Lie symmetry extensions within the class~$\mathcal L|_{\mathcal S}$. Solving the last system completes the group classification within the most general framework of the algebraic method. This whole construction is based on the following assertion: 

\begin{proposition}\label{pro:OnBasisOfAlgMethod}
Let $\mathcal S_i$ be the subset of~$\mathcal S$ that consists of all arbitrary elements 
for which the corresponding equations from $\mathcal L|_{\mathcal S}$ are invariant with respect to the same algebra of vector fields, $i=1,2$.
Then the algebras~$\mathfrak g^\cap(\mathcal L|_{\mathcal S_1})$ and $\mathfrak g^\cap(\mathcal L|_{\mathcal S_2})$ are similar with respect to push-forwards of vector fields by transformations from~$G^\sim$ (resp.\ point transformations) if and only if the subsets~$\mathcal S_1$ and~$\mathcal S_2$ are mapped to each other by transformations from~$G^\sim$ (resp.\ point transformations).
\end{proposition}

If the class~$\mathcal L|_{\mathcal S}$ is weakly normalized, the union~$\mathfrak g^\cup$ (resp.\ the span~$\mathfrak g^{\langle\rangle}$) is well agreed with $G^\sim$-equivalence. 
As a result, the algebraic method is appropriate for \emph{complete group classification} of the class~$\mathcal L|_{\mathcal S}$.
This is not the case when the main part of~$\mathfrak g^\cup$ does not lie in the projection~$\mathrm P\mathfrak g^\sim$. 
Then the approach of \emph{preliminary group classification} \cite{akha91Ay,ibra91Ay} is relevant to give a partial solution of the group classification problem for the class~$\mathcal L|_{\mathcal S}$ by the algebraic method. 
Preliminary group classification essentially rests on the following two propositions (they were first formulated in~\cite{ibra91Ay} in the particular case of the class~\eqref{eq:IbragimovClass}; see~\cite{card11Ay} for their general formulation and proofs):

\begin{proposition}\label{pro:OnPreliminaryGroupClassification1Ibragimov}
Let $\mathfrak a$ be a subalgebra of the equivalence algebra~$\mathfrak g^\sim$ of the class $\mathcal L|_{\mathcal S}$, $\mathfrak a\subset\mathfrak g^\sim$, and let $\theta^0(x,u_{(r)})\in\mathcal S$ be a value of the tuple of arbitrary elements~$\theta$ for which the algebraic equation $\theta=\theta^0(x,u_{(r)})$ is invariant with respect to~$\mathfrak a$. Then the differential equation $\mathcal L_{\theta^0}$ is invariant with respect to the projection of $\mathfrak a$ to the space of variables $(x,u)$.
\end{proposition}

\begin{proposition}\label{pro:OnPreliminaryGroupClassification2Ibragimov}\looseness=-1
Let $\mathcal S_i$ be the subset of~$\mathcal S$ that consists of tuples of arbitrary elements 
for which the corresponding algebraic equations are invariant with respect to the same subalgebra of the equivalence algebra~$\mathfrak g^\sim$ and
let $\mathfrak a_i$ be the maximal subalgebra of~$\mathfrak g^\sim$ for which $\mathcal S_i$ satisfies this property, $i=1,2$.
Then the subalgebras~$\mathfrak a_1$ and~$\mathfrak a_2$ are equivalent with respect to the adjoint action of~$G^\sim$ if and only if the subsets~$\mathcal S_1$ and~$\mathcal S_2$ are mapped to each other by transformations from~$G^\sim$.
\end{proposition}

Roughly speaking, in the course of preliminary group classification of the class~$\mathcal L|_{\mathcal S}$ 
we classify subalgebras of~$\mathfrak g^\sim$ instead of algebras of vector fields contained in~$\mathfrak g^\cup$. 
Then the objects to be classified (subalgebras of~$\mathfrak g^\sim$) are well agreed with the equivalence relation used ($G^\sim$-equivalence).
If a proper subalgebra~$\mathfrak s$ of~$\mathfrak g^\sim$ is fixed 
and then only subalgebras of~$\mathfrak s$ instead of the entire algebra~$\mathfrak g^\sim$ are classified up to the internal equivalence relation of subalgebras in~$\mathfrak s$ 
and used within the framework of the algebraic method, 
we call this approach \emph{partial preliminary group classification}. 

In view of Definition~\ref{def:WeakNormalizationIbragimov} and Proposition~\ref{pro:OnPreliminaryGroupClassification1Ibragimov} the following assertion is obvious.

\begin{corollary}
For a class of differential equations that is weakly normalized, complete preliminary group classification and complete group classification coincide.
\end{corollary} 

In fact, only certain subalgebras of the equivalence algebra~$\mathfrak g^\sim$ whose projections are contained in $\mathfrak g^\cup\cap\mathrm P\mathfrak g^\sim$ should be classified. 

\begin{definition}\label{def:AppropriateSubalgebras}
Within the framework of preliminary group classification, we will call a subalgebra~$\mathfrak a$ of the equivalence algebra~$\mathfrak g^\sim$ \emph{appropriate} if its projection~$\mathrm P\mathfrak a$ to the space of equation variables is a maximal Lie invariance algebra of an equation from the class~$\mathcal L|_{\mathcal S}$.
\end{definition}

Appropriate subalgebras of~$\mathfrak g^\sim$ satisfy restrictions similar to those for appropriate algebras contained in~$\mathfrak g^\cup$.
As the kernel is included in the maximal Lie invariance algebra of any equation from the class, in view of Corollary~\ref{cor:OnKernelAlgebraAsIdeal1} any appropriate subalgebra~$\mathfrak a$ of~$\mathfrak g^\sim$ should contain, as an ideal, the trivial prolongation~$\hat{\mathfrak g}^\cap$ of the kernel algebra~$\mathfrak g^\cap$ to the arbitrary elements.
The condition that the projection~$\mathrm P\mathfrak a$ of~$\mathfrak a$ is a Lie invariance algebra of a system from~$\mathcal L|_{\mathcal S}$ can be checked by two obviously equivalent ways: 
It is sufficient to prove that there exists a value $\theta^0(x,u_{(r)})\in\mathcal S$ of the tuple of arbitrary elements~$\theta$ for which the algebraic equation $\theta=\theta^0(x,u_{(r)})$ is invariant with respect to~$\mathfrak a$. The other way is to study the consistence of the system $\mathrm{DE}_{\mathfrak a}$ with the auxiliary system of the class~$\mathcal L|_{\mathcal S}$ with respect to the arbitrary elements. 
By $\mathrm{DE}_{\mathfrak a}$ we denote the system obtained by the substitution of the coefficients of each basis element of~$\mathrm P\mathfrak a$ into the determining equations of the class~$\mathcal L|_{\mathcal S}$.
Simultaneously we should verify the condition if the projection~$\mathrm P\mathfrak a$ is really the maximal Lie invariance algebra for some systems from~$\mathcal L|_{\mathcal S}$. 

\begin{remark}
Often the equivalence algebra can be represented as a semi-direct sum of the ideal associated with the kernel algebra and a certain subalgebra. To obtain preliminary group classification in this case, we in fact need to classify only inequivalent subalgebras of the complement of the kernel ideal. Projections of these subalgebras to the space of equation variables will give all possible inequivalent extensions of the kernel. This was the case for the class of generalized diffusion equations investigated in~\cite{card11Ay}. In the present paper, the situation will be different, see Remark~\ref{rem:OnStructureOfequivAlgebra}.
\end{remark}

The importance of semi-normalization of a class of differential equations for the optimal solution of the group classification problem for this class 
is connected with the following property of semi-normalized classes. 

\begin{proposition}\label{pro:OnGroupClassificationsInSemiNormClass}
For a class of differential equations that is semi-normalized, the group classification up to equivalence generated by the corresponding equivalence group coincides 
with the group classification up to general point equivalence.
\end{proposition}

In other words, in a semi-normalized class of differential equations there are no additional equivalence transformations 
between cases of Lie symmetry extensions which are inequivalent with respect to the corresponding equivalence group. 
This results in a clear representation of the final classification list. 
As normalized classes of differential equations are both semi-normalized and weakly normalized, 
it is especially convenient to carry out group classification in such classes by the algebraic method.  
This is why the normalization property can be used as a criterion for selecting classes of differential equations to be classified 
or for splitting of such classes into subclasses which are appropriate for group classification.

\section{Equivalence algebra}\label{sec:EquivalenceAlgebraIbragimovClass}

The equivalence algebra of the entire class of equation of the general form~\eqref{eq:IbragimovClass} was already computed in~\cite{ibra91Ay}. 
It coincides with the equivalence algebra of the class considered in the present paper, which consists of purely nonlinear equations of the above form. 
This is why we only represent generating elements of this algebra in a convenient form and refer the reader to~\cite{ibra91Ay} for more details.
The equivalence algebra~$\mathfrak g^\sim$ of class~\eqref{eq:IbragimovClass} is generated by the vector fields
\begin{gather}\label{eq:EquivalenceAlgebraGenWaveEqs}
\begin{split}&
\DDD^u=u\p_u+u_x\p_{u_x}+g\p_g,\quad
\DDD^t=t\p_t-2f\p_f-2g\p_g,\quad
\PP^t=\p_t,\\&
\DDD(\varphi)=\varphi\p_x-\varphi_xu_x\p_{u_x}+2\varphi_xf\p_f+\varphi_{xx}u_xf\p_g,\\&
\GG(\psi)=\psi\p_u+\psi_x\p_{u_x}-\psi_{xx}f\p_g,\quad
\FF^1=t\p_u, \quad
\FF^2=t^2\p_u+2\p_g,
\end{split}
\end{gather}
where~$\varphi=\varphi(x)$ and~$\psi=\psi(x)$ run through the set of smooth functions of~$x$.
The nonvanishing commutation relations between the these vector fields are exhausted by
\begin{align*}
 &[\GG(\psi),\DDD^u]=\GG(\psi), \quad [\FF^1,\DDD^u]=\FF^1, \quad [\FF^2,\DDD^u]=\FF^2,\\
 &[\DDD^t,\FF^1]=\FF^1, \quad [\DDD^t,\FF^2]=2\FF^2,\\
 &[\PP^t,\DDD^t]=\PP^t, \quad [\PP^t,\FF^1]=\GG(1), \quad [\PP^t,\FF^2]=2\FF^1, \\
 &[\DDD(\varphi^1),\DDD(\varphi^2)]=\DDD(\varphi^1\varphi^2_x-\varphi^1_x\varphi^2), \quad [\DDD(\varphi),\GG(\psi)]=\GG(\varphi\psi_x).
\end{align*}

{\sloppy
In fact, in~\eqref{eq:EquivalenceAlgebraGenWaveEqs} we present only the projections of generating elements of~$\mathfrak g^\sim$ to the space of $(t,x,u,u_x,f,g)$ 
instead of the whole elements, which are vector fields in the space of $(t,x,u_{(2)},f,g)$.  
The terms of generating vector fields, which are associated with derivatives of~$u$,
can be computed via prolongation from the coefficients of~$\p_t$, $\p_x$ and~$\p_u$ and, therefore, are not essential. 
However, it is necessary to include the terms with~$\p_{u_x}$ in the representation of these vector fields 
in order to ensure proper commutation relations between them. 
Moreover, the derivative~$u_x$ is a significant argument of the parameter-functions~$f$ and~$g$ and hence the minimal space 
on which equivalence transformations can be correctly restricted is the space of the variables $(t,x,u,u_x,f,g)$. 
This is why at least the projections to the same space should be given for vector fields from~$\mathfrak g^\sim$.

}

The form~(3.16) of the equivalence algebra given in~\cite{ibra91Ay} differs from~\eqref{eq:EquivalenceAlgebraGenWaveEqs}. Namely, the operators $\GG(1)$ and $\GG(x)$ were singled out from the family~$\{\GG(\psi)\}$. Additionally, we combined the operators from~\cite{ibra91Ay} to separate scalings with respect to $u$, which gives simpler commutation relations between generating vector fields.

In order to compute the complete equivalence group of class~\eqref{eq:IbragimovClass} by the algebraic method
(cf. Section~\ref{sec:CalculationOfEquivAlgebraByAlgebraicMethod}) 
we need to know a set of megaideals of the equivalence algebra~$\mathfrak g^\sim$ of this class. 

Given a Lie algebra~$\mathfrak g$, a \emph{megaideal} $\mathfrak i$ is a vector subspace of $\mathfrak g$ that is invariant under any transformation from the automorphism group $\mathrm{Aut}(\mathfrak g)$ of~$\mathfrak g$ 
\cite{bihl11Cy,popo03Ay}.
That is, we have $\mathfrak T \mathfrak i=\mathfrak i$ for a megaideal~$\mathfrak i$ of~$\mathfrak g$, whenever $\mathfrak T$ is a transformation from $\mathrm{Aut}(\mathfrak g)$. Any megaideal of~$\mathfrak g$ is an ideal and a characteristic ideal of~$\mathfrak g$. Both the improper subalgebras of~$\mathfrak g$ (the zero subspace and $\mathfrak g$ itself) are megaideals of~$\mathfrak g$. 
If $\mathfrak i_1$ and $\mathfrak i_2$ are megaideals of~$\mathfrak g$ then so are $\mathfrak i_1+\mathfrak i_2,$ $\mathfrak i_1\cap \mathfrak i_2$ and $[\mathfrak i_1,\mathfrak i_2]$, i.e., sums, intersections and Lie products of megaideals are again megaideals. 
If $\mathfrak i_2$ is a megaideal of $\mathfrak i_1$ and $\mathfrak i_1$ is a megaideal of $\mathfrak g$ then $\mathfrak i_2$ is a megaideal of $\mathfrak g$, i.e., megaideals of megaideals are also megaideals.
The centralizer of a megaideal is a megaideal. 
All elements of the derived, upper and lower central series of a Lie algebra are its megaideals. In particular, the center and the derivative of a Lie algebra are its megaideals.
The radical~$\mathfrak r$ and nil-radical~$\mathfrak n$ (i.e., the maximal solvable and nilpotent ideals, respectively) of~$\mathfrak g$ 
as well as different Lie products, sums and intersections involving~$\mathfrak g$, $\mathfrak r$ and~$\mathfrak n$ 
($[\mathfrak g,\mathfrak r]$, $[\mathfrak r,\mathfrak r]$, $[\mathfrak g,\mathfrak n]$, $[\mathfrak r,\mathfrak n]$,  $[\mathfrak n,\mathfrak n]$, etc.) are megaideals of~$\mathfrak g$.

Here we only prove that, roughly speaking, the centralizer of a megaideal is a megaideal. 

\begin{proposition}
If~$\mathfrak i$ is a megaideal of~$\mathfrak g$ 
then the centralizer $\mathrm C_{\mathfrak g}(\mathfrak i)$ of~$\mathfrak i$ in~$\mathfrak g$ 
is also a megaideal of~$\mathfrak g$. 
\end{proposition}

\begin{proof}
Consider an arbitrary $\mathfrak T\in\mathrm{Aut}(\mathfrak g)$, $v\in\mathfrak i$ and
$w\in\mathrm C_{\mathfrak g}(\mathfrak i)$. Then 
$[\mathfrak Tw,v]=[\mathfrak Tw,\mathfrak T\mathfrak T^{-1}v]=\mathfrak T[w,\mathfrak T^{-1}v]=0$ 
as $\mathfrak T^{-1}v\in\mathfrak i$ and hence $[w,\mathfrak T^{-1}v]=0$. 
This means that $\mathfrak Tw\in\mathrm C_{\mathfrak g}(\mathfrak i)$, i.e., 
$\mathrm C_{\mathfrak g}(\mathfrak i)$ is a megaideal of~$\mathfrak g$.
\end{proof}

Let $\mathfrak g=\mathfrak g^\sim$. It is easy to compute the following megaideals of~$\mathfrak g^\sim$: 
\begin{gather*}
\mathfrak g'=\langle\PP^t,\DDD(\varphi),\GG(\psi),\FF^1,\FF^2\rangle,\quad 
\mathfrak g''=\langle\DDD(\varphi),\GG(\psi),\FF^1\rangle,\quad 
\mathfrak g'''=\langle\DDD(\varphi),\GG(\psi)\rangle,\\
\mathrm C_{\mathfrak g}(\mathfrak g''')=\langle\DDD^t,\PP^t,\GG(1),\FF^1,\FF^2\rangle,\quad
\mathrm C_{\mathfrak g'}(\mathfrak g''')=\langle\PP^t,\GG(1),\FF^1,\FF^2\rangle,\\
\mathrm C_{\mathfrak g'}(\mathfrak g'')=\langle\GG(1),\FF^1,\FF^2\rangle,\quad
\mathrm Z_{\mathfrak g''}=\langle\GG(1),\FF^1\rangle,\quad
\mathrm Z_{\mathfrak g'}=\langle\GG(1)\rangle,\\
\mathrm R_{\mathfrak g}=\langle\DDD^u,\DDD^t,\PP^t,\GG(\psi),\FF^1,\FF^2\rangle,\quad 
\mathrm R_{\mathfrak g'''}=\langle\GG(\psi)\rangle, 
\end{gather*}
where $\mathfrak a'$, $\mathrm R_{\mathfrak a}$, $\mathrm Z_{\mathfrak a}$ and $\mathrm C_{\mathfrak a}(\mathfrak b)$ 
denote the derivative, the radical and the center of a Lie algebra~$\mathfrak a$ 
and the centralizer of a subalgebra~$\mathfrak b$ in~$\mathfrak a$, respectively. 
We present proofs only for the last two equalities. 

The linear span $\mathfrak s_1=\langle\DDD^u,\DDD^t,\PP^t,\GG(\psi),\FF^1,\FF^2\rangle$ 
obviously is a solvable ideal of~$\mathfrak g$. 
Moreover, it is the maximal solvable ideal of~$\mathfrak g$. 
Indeed, suppose that $\mathfrak s_1\subsetneq\mathfrak i$ and $\mathfrak i$ is an ideal of~$\mathfrak g$. 
Then there exists a smooth function~$\zeta$ of~$x$ which does not identically vanish
such that the vector field $\DDD(\zeta)$ belongs to~$\mathfrak i$. 
As $\mathfrak i$ is an ideal of~$\mathfrak g$, 
for an arbitrary smooth function~$\varphi$ of~$x$ the commutator $[\DDD(\zeta),\DDD(\varphi)]=\DDD(\zeta\varphi_x-\zeta_x\varphi)$ belongs to~$\mathfrak i$. 
If $\zeta$ is not a constant function, we define the following series of operators:
\[
R^{0k}=k^{-1}[\DDD(\zeta),\DDD(\zeta^{k+1})],\quad 
R^{jk}=k^{-1}[R^{j-1,1},R^{j-1,k+1}],\quad 
j,k=1,2,\dots.
\] 
It is possible to prove by induction that $R^{j-1,k}=\DDD(\zeta^k(\zeta\zeta_x)^{2^j-1})\ne0$, $j,k=1,2,\dots$. 
Moreover, as $R^{0k}\in\mathfrak i$, we have $R^{jk}\in\mathfrak i^{(j)}$, $j,k=1,2,\dots$, i.e., $\mathfrak i^{(j)}\ne\{0\}$ for any nonnegative integer~$j$. 
This means that the ideal~$\mathfrak i$ is not solvable. 
If $\zeta$ is a constant function, we can set $\zeta\equiv1$. 
We choose any smooth function~$\varphi$ of~$x$ with $\varphi_{xx}\not\equiv0$ and denote~$\varphi_x$ by~$\tilde\zeta$.
As the commutator $[\DDD(1),\DDD(\varphi)]=\DDD(\tilde\zeta)$ belongs to~$\mathfrak i$, 
the consideration for the previous case again implies that the ideal~$\mathfrak i$ is not solvable. 
Therefore, $\mathfrak s_1$ is really the maximal solvable ideal of~$\mathfrak g$, i.e., $\mathrm R_{\mathfrak g}=\mathfrak s_1$.

The linear span $\mathfrak s_2=\langle\GG(\psi)\rangle$ is an Abelian and, therefore, solvable ideal of~$\mathfrak g'''$. 
The maximality of this solvable ideal is proved in the same way as for~$\mathfrak s_1$. 
Hence $\mathrm R_{\mathfrak g'''}=\mathfrak s_2$.

The same megaideals can be obtained in different ways. For example, $\langle\GG(1)\rangle=\mathrm Z_{\mathfrak g'}=\mathrm Z_{\mathfrak g'''}$.

To find more megaideals, we can calculate the automorphism group~$\mathrm{Aut}(\mathfrak m)$ of a finite-dimensional megaideal~$\mathfrak m$ 
and then determine megaideals of~$\mathfrak m$ as subspaces of~$\mathfrak m$ invariant with respect to~$\mathrm{Aut}(\mathfrak m)$. 
In the course of calculating the automorphisms we can use knowledge about simple megaideals of~$\mathfrak m$. 
Consider $\mathfrak m=\mathrm C_{\mathfrak g}(\mathfrak g''')=\langle\GG(1),\FF^1,\FF^2,\PP^t,\DDD^t\rangle$. 
Then 
\[
\mathfrak m'=\langle\GG(1),\FF^1,\FF^2,\PP^t\rangle,\quad 
\mathfrak m''=\langle\GG(1),\FF^1\rangle,\quad
\mathrm Z_{\mathfrak m}=\langle\GG(1)\rangle,\quad
\mathrm C_{\mathfrak m}(\mathfrak m'')=\langle\GG(1),\FF^1,\FF^2\rangle.
\]
The presence of the above set of nested megaideals is equivalent to that for any automorphism~$A$ of~$\mathfrak m$,
its matrix in the basis $\{\GG(1),\FF^1,\FF^2,\PP^t,\DDD^t\}$ is upper triangular with nonzero diagonal elements. 
In particular, 
\begin{gather*}
A\PP^t=a_{14}\GG(1)+a_{24}\FF^1+a_{34}\FF^2+a_{44}\PP^t, \\ 
A\DDD^t=a_{15}\GG(1)+a_{25}\FF^1+a_{35}\FF^2+a_{45}\PP^t+a_{55}\DDD^t,
\end{gather*}
where $a_{44}a_{55}\ne0$. 
As $[\PP^t,\DDD^t]=\PP^t$ and $A$ is an automorphism of~$\mathfrak m$, we should have $[A\PP^t,A\DDD^t]=A\PP^t$. 
Collecting coefficients of basis elements in the last equality, we derive a system of equations with respect to~$a$'s 
which implies, in view of the condition $a_{44}\ne0$, that $a_{55}=1$, $a_{34}=0$, $a_{24}=a_{44}a_{35}$ and 
$a_{14}=a_{44}a_{25}-a_{45}a_{24}$. 
As a result, we find one more megaideal $\langle\GG(1),\FF^1,\PP^t\rangle$ of~$\mathfrak m$ and, therefore, 
$\mathfrak g^\sim$.
We will use this megaideal in the course of the computation of the complete equivalence group of class~\eqref{eq:IbragimovClass} by the algebraic method
in Section~\ref{sec:CalculationOfEquivAlgebraByAlgebraicMethod}.

\section{Preliminary study of admissible transformations}\label{sec:PreliminaryStudyOfAdmTrans}

The infinitesimal invariance criterion allows finding of all continuous equivalence transformations by means of solving a linear system of partial differential equations. In order to determine the complete point equivalence group (including both continuous and discrete equivalence transformations) and the set of admissible transformations, it is necessary to apply the direct method. We will start our consideration with a preliminary investigation of the set of admissible transformations, which will give relevant information also on the equivalence group of the class~\eqref{eq:IbragimovClass}. That is, we directly seek for all point transformations
\begin{equation}\label{eq:GeneralEquivalenceTransformation}
    \tilde t = T(t,x,u), \quad \tilde x = X(t,x,u), \quad \tilde u = U(t,x,u),
\end{equation}
for which the Jacobian $\mathrm J=\p(T,X,U)/\p(t,x,u)$ does not vanish, that map a fixed equation of the form~\eqref{eq:IbragimovClass} to an equation of the same form,
\[
    \tilde u_{\tilde t\tilde t} = \tilde f(\tilde x,\tilde u_{\tilde x})\tilde u_{\tilde x\tilde x} + \tilde g(\tilde x,\tilde u_{\tilde x}).
\]
To carry out this transformation in practice, it is necessary to find the transformation rules for the various derivatives of $\tilde u$ with respect to $\tilde t$ and $\tilde x$. In order to obtain them we apply the total derivative operators $\DD_t$ and $\DD_x$, respectively, to the expression $\tilde u(\tilde t,\tilde x) = U(t,x,u)$, assuming $\tilde t = T(t,x,u)$ and $\tilde x = X(t,x,u)$. This gives
\begin{align*}
 &\tilde u_{\tilde t}\DD_tT + \tilde u_{\tilde x}\DD_tX - \DD_tU = 0, \\
 &\tilde u_{\tilde t}\DD_xT + \tilde u_{\tilde x}\DD_xX - \DD_xU = 0, \\
 &\tilde u_{\tilde t\tilde t}(\DD_tT)^2 + 2\tilde u_{\tilde t\tilde x}(\DD_tX)(\DD_tT) + \tilde u_{\tilde x\tilde x}(\DD_tX)^2+\tilde u_{\tilde t}\DD_t^2T+\tilde u_{\tilde x}\DD_t^2X - \DD_t^2U = 0, \\
 &\tilde u_{\tilde t\tilde t}(\DD_xT)^2 + 2\tilde u_{\tilde t\tilde x}(\DD_xX)(\DD_xT) + \tilde u_{\tilde x\tilde x}(\DD_xX)^2+\tilde u_{\tilde t}\DD_x^2T+\tilde u_{\tilde x}\DD_x^2X - \DD_x^2U = 0.
\end{align*}
Solving the last two equations for $u_{tt}$ and $u_{xx}$, respectively, and substituting the results into~\eqref{eq:IbragimovClass}, we obtain
\begin{align}\label{eq:TransformationOfIbragimovClass}
\begin{split}
&\tilde u_{\tilde t\tilde t}(\DD_tT)^2+ 2\tilde u_{\tilde t\tilde x}(\DD_tT)(\DD_tX)+ \tilde u_{\tilde x\tilde x}(\DD_tX)^2
+\tilde u_{\tilde t}V^tT+\tilde u_{\tilde x}V^tX-V^tU\\
&{}= f\bigl(\tilde u_{\tilde t\tilde t}(\DD_xT)^2+2\tilde u_{\tilde t\tilde x}(\DD_xT)(\DD_xX) +\tilde u_{\tilde x\tilde x}(\DD_xX)^2
+\tilde u_{\tilde t}V^xT+\tilde u_{\tilde x}V^xX-V^xU\bigr)\\
&{}-g(\tilde u_{\tilde t} T_u+\tilde u_{\tilde x}X_u-U_u),
\end{split}
\end{align}
where we use the notation $V^t=\p_{tt}+2u_t\p_{tu}+u_t^{\;2}\p_{uu}$ and $V^x=\p_{xx}+2u_x\p_{xu}+u_x^{\ 2}\p_{uu}$ and additionally have to set $\tilde u_{\tilde t\tilde t}=\tilde f\tilde u_{\tilde x\tilde x} + \tilde g$ wherever it occurs.
As the derivative $\tilde u_{\tilde t\tilde x}$ does not appear in the transformed form of equations from the class~\eqref{eq:IbragimovClass}, 
the associated coefficient in \eqref{eq:TransformationOfIbragimovClass} vanishes,~i.e.
\begin{align}\label{eq:DeterminingEqsEquivalenceTransformations1}
 (T_t+T_uu_t)(X_t+X_uu_t)=f(T_x+T_uu_x)(X_x+X_uu_x).
\end{align}
Eq.~\eqref{eq:DeterminingEqsEquivalenceTransformations1} involves only original (untilded) variables and is a polynomial in~$u_t$. 
Therefore, we can split it with respect to~$u_t$ by collecting the coefficients of different powers of this derivative. 
(Note that we cannot as directly split Eq.~\eqref{eq:DeterminingEqsEquivalenceTransformations1} with respect to the derivative~$u_x$, which is an argument of the function~$f$.) 
As a result, we derive that
\begin{gather}
\label{eq:DeterminingEqsEquivalenceTransformations1ut2}
u_t^2\colon\quad T_uX_u=0,
\\ \label{eq:DeterminingEqsEquivalenceTransformations1ut1}
u_t^1\colon\quad T_uX_t+T_tX_u=0,
\\ \label{eq:DeterminingEqsEquivalenceTransformations1ut0}
u_t^0\colon\quad T_tX_t = f(T_xX_x+(T_uX_x+T_xX_u)u_x).
\end{gather}
Multiplying Eq.~\eqref{eq:DeterminingEqsEquivalenceTransformations1ut1} by~$T_u$ (resp.\ $X_u$), we obtain in view of Eq.~\eqref{eq:DeterminingEqsEquivalenceTransformations1ut2} that $T_uX_t=0$ (resp.\ $T_tX_u=0$). 
We apply the trick with the multiplication by~$T_u$ (resp.\ $X_u$) also to Eq.~\eqref{eq:DeterminingEqsEquivalenceTransformations1ut0} 
and take into account the equations $T_uX_u=0$, $T_uX_t=0$ and $T_tX_u=0$ already derived and the inequality~$f\ne0$. 
This gives equations which involve no arbitrary elements and hence can be further split with respect to~$u_x$. 
Therefore, these equations are equivalent to the equations $T_uX_x=0$ and $X_uT_x=0$, respectively. 
The system $T_uX_t=0$, $T_uX_x=0$, $T_uX_u=0$ (resp.\ $X_uT_t=0$, $X_uT_x=0$, $X_uT_u=0$) implies that $T_u=0$ (resp.\ $X_u=0$)
since otherwise the Jacobian~$\mathrm J$ of the point transformation~\eqref{eq:GeneralEquivalenceTransformation} vanishes. 
The condition 
\[
T_u=X_u=0
\]
means that any admissible point transformation of the class~\eqref{eq:IbragimovClass} is fiber-preserving. 
In view of this condition, Eqs.~\eqref{eq:DeterminingEqsEquivalenceTransformations1ut2} and~\eqref{eq:DeterminingEqsEquivalenceTransformations1ut1} are identically satisfied 
and the remainder of Eq.~\eqref{eq:DeterminingEqsEquivalenceTransformations1ut0} is 
\begin{equation}\label{eq:DeterminingEqsEquivalenceTransformations1Reminder}
T_tX_t = fT_xX_x. 
\end{equation}

After substituting $\tilde u_{\tilde t\tilde t}=\tilde f\tilde u_{\tilde x\tilde x} + \tilde g$, 
we can also split~\eqref{eq:TransformationOfIbragimovClass} with respect to $\tilde u_{\tilde x\tilde x}$, which gives, in view of $T_u=X_u=0$, the equation
\begin{equation}\label{eq:DeterminingEqsEquivalenceTransformations2}
f\tilde f T_t^{\;2}+X_t^{\;2} = f(\tilde fT_x^{\;2}+X_x^{\;2}).
\end{equation}
Unfortunately, the direct splitting  with respect to other derivatives in Eq.~\eqref{eq:TransformationOfIbragimovClass} is not possible.
The remaining part of~\eqref{eq:TransformationOfIbragimovClass} therefore~is
\begin{align}\label{eq:DeterminingEqsEquivalenceTransformations3}
\begin{split}
 &\tilde g T_t^2 + \tilde u_{\tilde t} T_{tt} + \tilde u_{\tilde x}X_{tt} -(U_{tt}+2U_{tu}u_t + U_{uu}u_t^2)\\&{}
= f(\tilde g T_x^2+\tilde u_{\tilde t}T_{xx} + \tilde u_{\tilde x}X_{xx}-(U_{xx}+2U_{xu}u_x+U_{uu}u_x^2))+ gU_u.
\end{split}
\end{align}
The additional condition to keep in mind is the nondegeneracy of transformations~\eqref{eq:GeneralEquivalenceTransformation}, 
which in view of the conditions $T_u=X_u=0$ is reduced to the inequality $U_u(T_tX_x-T_xX_t)\ne0$ and hence $(T_tX_x-T_xX_t)\ne0$ and $U_u\ne0$.

\section{Equivalence group}\label{sec:EquivGroup}

At this point, we continue the consideration by computing the equivalence group as it is needed even for the analysis of the determining equations for coefficients of Lie symmetry operators and the exhaustive description of admissible transformations. In the case of equivalence transformations, the arbitrary elements~$f$ and~$g$ should be varied. We can therefore split 
the equations~\eqref{eq:DeterminingEqsEquivalenceTransformations1Reminder}, \eqref{eq:DeterminingEqsEquivalenceTransformations2} and~\eqref{eq:DeterminingEqsEquivalenceTransformations3} 
also with respect to the arbitrary elements. 
Eq.~\eqref{eq:DeterminingEqsEquivalenceTransformations1Reminder} and the nondegeneracy constraint $T_tX_x-T_xX_t\ne0$ imply 
that either $X_x=T_t=0$, $T_x\ne0$ and $X_t\ne0$ or $X_t=T_x=0$, $X_x\ne0$ and $T_t\ne0$.

For $X_x=T_t=0$, Eq.~\eqref{eq:DeterminingEqsEquivalenceTransformations2} is simplified to $X_t^2=f\tilde fT_x^2$.  
As the expression for the derivative~$u_x$ in the new variables is $u_x=(T_x\tilde u_{\tilde t}-U_x)/U_u$, 
i.e., it does not involve~$\tilde u_{\tilde x}$, the equality $X_t^2=f\tilde fT_x^2$ can be split into the two equations  
$T_x=0$ and $X_t=0$, which contradict the nondegeneracy condition.

Therefore we necessarily have $X_t=T_x=0$ and thus $T=T(t)$, $X=X(x)$, where $X_x\ne0$ and $T_t\ne0$. 
Then Eq.~\eqref{eq:DeterminingEqsEquivalenceTransformations2} is reduced to $\tilde fT_t^2=fX_x^2$ 
and the differentiation of this equation with respect to~$t$ yields
\begin{equation}\label{eq:DeterminingEqsEquivalenceTransformations2a}
 2T_tT_{tt}\tilde f + \tilde f_{\tilde u_{\tilde x}}\frac{U_{tx}+U_{tu}u_x}{X_x} = 0.
\end{equation}
Since Eq.~\eqref{eq:DeterminingEqsEquivalenceTransformations2a} holds for all~$\tilde f$, we can split it and derive $T_{tt}=0$, $U_{xt}=0$ and $U_{tu}=0$. Collecting coefficients of $u_t^2$ in Eq.~\eqref{eq:DeterminingEqsEquivalenceTransformations3} we moreover find that $U_{uu}=0$. Taking all the constraints derived into account, Eq.~\eqref{eq:DeterminingEqsEquivalenceTransformations3} reads
\[
 \tilde g T_t^2 - U_{tt} = f\left(\frac{U_uu_x+U_x}{X_x}X_{xx}-U_{xx}-2U_{xu}u_x\right)+gU_u.
\]
Differentiating this equation with respect to~$u$ and $t$ allows deriving that $U_{xu}=0$ and $U_{ttt}=0$.

Collecting all the restrictions derived up to now, any equivalence transformation must satisfy the following system of differential equations
\begin{align}\label{eq:DeterminingEqsEquivalenceTransformationsFinalSystem}
\begin{split}
&T_u=T_x=T_{tt}=0, \quad X_u=X_t=0, \\
&U_{uu} = U_{tu} = U_{xu} = U_{tx} = U_{ttt} = 0.
\end{split}
\end{align}
Integrating the above system in view of the nondegeneracy condition $\mathrm J\ne0$, we proved the following theorem:

\begin{theorem}\label{thm:EquivalenceGroupIbragrimovClass}
The equivalence group~$G^\sim$ of the class~\eqref{eq:IbragimovClass} consists of the transformations
\begin{align}\label{eq:EquivalenceGroupIbragimovClass}
\begin{split}
 &\tilde t = c_1t+c_0, \quad \tilde x=\varphi(x), \quad \tilde u = c_2u+c_4t^2+c_3t+\psi(x), \quad \tilde u_{\tilde x}=\frac{c_2u_x+\psi_x}{\varphi_x},\\
 &\tilde f = \frac{\varphi_x^2}{c_1^2}f, \quad \tilde g = \frac{1}{c_1^2}\left(c_2g+\frac{c_2u_x+\psi_x}{\varphi_x}\varphi_{xx}f-\psi_{xx}f+2c_4\right),
\end{split}
\end{align}
where $c_0$, \dots, $c_4$ are arbitrary constants satisfying the condition $c_1c_2\ne0$ and $\varphi$ and~$\psi$ run through the set of smooth functions of~$x$, $\varphi_x\ne0$.
\end{theorem}

Comparing the equivalence algebra~\eqref{eq:EquivalenceAlgebraGenWaveEqs} and the equivalence group~\eqref{eq:EquivalenceGroupIbragimovClass} the following corollary is evident:
\begin{corollary}
The class of equations~\eqref{eq:IbragimovClass} admits three independent discrete equivalence transformations, which are given by $(t,x,u,f,g)\mapsto(-t,x,u,f,g)$, $(t,x,u,f,g)\mapsto(t,-x,u,f,g)$ and $(t,x,u,f,g)\mapsto(t,x,-u,f,-g)$.
\end{corollary}

Theorem~\ref{thm:EquivalenceGroupIbragrimovClass} implies that any transformation~$\mathcal T$ from~$G^\sim$ of the class~\eqref{eq:IbragimovClass} can be represented as the composition
\[
  \mathcal T = \mathscr D^t(c_1)\mathscr P^t(c_0)\mathscr D(\varphi)\mathscr D^u(c_2)\mathscr F^1(c_4)\mathscr F^2(c_3)\mathscr G(\psi),
\]
where
\[
\arraycolsep=0ex
\begin{array}{lllllll}
\mathscr P^t(c_0)  \colon\ & \tilde t=t+c_0,\ & \tilde x=x,      \quad& \tilde u=u,       \ & \tilde u_{\tilde x}=u_x,            \ & \tilde f=f,           \ & \tilde g=g,\\
\mathscr D^t(c_1)  \colon\ & \tilde t=c_1t, \ & \tilde x=x,      \quad& \tilde u=u,       \ & \tilde u_{\tilde x}=u_x,            \ & \tilde f=c_1^{-2}f,   \ & \tilde g=c_1^{-2}g,\\
\mathscr D(\varphi)\colon\ & \tilde t=t,    \ & \tilde x=\varphi,\quad& \tilde u=u,       \ & \tilde u_{\tilde x} = u_x/\varphi_x,\ & \tilde f=\varphi^2_xf,\ & \tilde g=g+\varphi_{xx}u_xf/\varphi_x,\\
\mathscr D^u(c_2)  \colon\ & \tilde t=t,    \ & \tilde x=x,      \quad& \tilde u=c_2u,    \ & \tilde u_{\tilde x}=c_2u_x,         \ & \tilde f=f,           \ & \tilde g=c_2g,\\
\mathscr F^1(c_3)  \colon\ & \tilde t=t,    \ & \tilde x=x,      \quad& \tilde u=u+c_3t,  \ & \tilde u_{\tilde x}=u_x,            \ & \tilde f=f,           \ & \tilde g=g,\\
\mathscr F^2(c_4)  \colon\ & \tilde t=t,    \ & \tilde x=x,      \quad& \tilde u=u+c_4t^2,\ & \tilde u_{\tilde x}=u_x,            \ & \tilde f=f,           \ & \tilde g=g+2c_4,\\
\mathscr G(\psi)   \colon\ & \tilde t=t,    \ & \tilde x=x,      \quad& \tilde u=u+\psi,  \ & \tilde u_{\tilde x}=u_x+\psi_x,     \ & \tilde f=f,           \ & \tilde g=g-\psi_{xx}f,
\end{array}
\]
are (families of) one-parameter equivalence transformations and the nondegeneracy requires that $c_1c_2\varphi_x\ne0$. These transformations are shifts and scalings in $t$, arbitrary transformations in $x$, scalings of $u$, gauging transformations of $u$ with square polynomials in $t$ and arbitrary functions of $x$.

\section{Calculation of equivalence group by the algebraic method}
\label{sec:CalculationOfEquivAlgebraByAlgebraicMethod}

It is well known that any point symmetry transformation~$\mathcal T$ of a differential equation 
(resp.\ a system of differential equations) $\mathcal L$ 
generates an automorphism of the maximal Lie invariance algebra of~$\mathcal L$ via push-forwarding of vector fields in the space of equation variables. 
This condition implies constraints for the transformation~$\mathcal T$ which are then taken into account in further calculations using the direct method~\cite{bihl11Cy,hydo00By}. 
The set of transformations found in the way described form the complete point symmetry group of the equation~$\mathcal L$
including both continuous and discrete point transformations. 
The above algebraic method can be easily extended to the framework of equivalence transformations. 
A basis for this is given by the following simple proposition. 

\begin{proposition}
Let $\mathcal L|_{\mathcal S}$ be a class of (systems of) differential equations, 
$G^\sim$ and~$\mathfrak g^\sim$ the equivalence group and the equivalence algebra of this class (of the same type, namely, either usual or generalized ones). 
Any transformation~$\mathcal T$ from~$G^\sim$ generates an automorphism of~$\mathfrak g^\sim$ via push-forwarding of vector fields in the space of equation variables, appropriate derivatives and arbitrary elements.
\end{proposition}

\begin{proof}
Consider an arbitrary vector field $Q\in\mathfrak g^\sim$. 
The local one-parameter transformation group $G=\{\exp(\ve Q)\}$ associated with $Q$ 
is contained in~$G^\sim$. 
Then the one-parameter transformation group $\tilde G=\{\mathcal T\exp(\ve Q)\mathcal T^{-1}\}$, 
which is similar to~$G$ with respect to~$\mathcal T$ and whose generator is $\mathcal T_*Q$, 
is also contained in~$G^\sim$. 
This means that the vector field $\mathcal T_*Q$ belongs to~$\mathfrak g^\sim$. 
An arbitrary push-forward saves the Lie bracket of vector fields, 
$[\mathcal T_*Q,\mathcal T_*Q']=\mathcal T_*[Q,Q']$ for any $Q,Q'\in\mathfrak g^\sim$. 
Therefore, $\mathcal T_*$ is an automorphisms of~$\mathfrak g^\sim$.
\end{proof}

Here we compute the usual equivalence group~$G^\sim$ of class~\eqref{eq:IbragimovClass} by the algebraic method. 
The purpose of this computation is dual: to test the results of the previous section and
to present an example of applying the algebraic method. 
The group~$G^\sim$ consists of nondegenerate point transformations 
in the joint space of variables~$t$, $x$ and $u$, 
the first derivatives~$u_t$ and~$u_x$ and the arbitrary elements~$f$ and~$g$, 
which are projectable to the variable space and 
whose components for first derivatives are defined via 
the first prolongation of their projections to the variable space.
Thus, the general form of a transformation~$\mathcal T$ from~$G^\sim$ is 
\begin{gather*}
\tilde t=T(t,x,u),\quad \tilde x=X(t,x,u),\quad \tilde u=U(t,x,u),\\ 
\tilde u_{\tilde t}=U^t(t,x,u,u_t,u_x),\quad \tilde u_{\tilde x}=U^x(t,x,u,u_t,u_x),\\
\tilde f=F(t,x,u,u_t,u_x,f,g),\quad \tilde g=G(t,x,u,u_t,u_x,f,g), 
\end{gather*}
where~$U^t$ and~$U^x$ are determined via~$T$, $X$ and~$U$ and 
the nondegeneracy condition should be additionally satisfied. 
To obtain the constrained form of $\mathcal T$, 
we will act by the push-forward~$\mathcal T_*$ induced by~$\mathcal T$ 
on the vector fields~\eqref{eq:EquivalenceAlgebraGenWaveEqs} additionally including the terms with~$\p_{u_t}$ 
and use megaideals of the equivalence algebra~$\mathfrak g^\sim$ of class~\eqref{eq:IbragimovClass}
and restrictions on automorphisms of~$\mathfrak g^\sim$ found in Section~\ref{sec:EquivalenceAlgebraIbragimovClass}. 
Note that the majority of these objects and properties of~$\mathfrak g^\sim$ are related to the finite-dimensional megaideal 
$\mathfrak m=\mathrm C_{\mathfrak g}(\mathfrak g''')=\langle\GG(1),\FF^1,\FF^2,\PP^t,\DDD^t\rangle$. 
Megaideals which are sums of other megaideals are not essential for the computation 
since they give weaker constraints than their summands. 
For example, the megaideal~$\mathfrak g''$ is the sum of~$\mathfrak g'''$ and~$\mathrm Z_{\mathfrak g'}$ 
and hence we do not use it in the further consideration.
It is sufficient to use the following equalities:
\begin{gather}
\label{eq:PushForwardOfT1}
\mathcal T_*\GG(1)=a_{11}\tilde\GG(1), 
\\\label{eq:PushForwardOfT2} 
\mathcal T_*\FF^1=a_{12}\tilde\GG(1)+a_{22}\tilde\FF^1, 
\\\label{eq:PushForwardOfT3} 
\mathcal T_*\FF^2=a_{13}\tilde\GG(1)+a_{23}\tilde\FF^1+a_{33}\tilde\FF^2, 
\\\label{eq:PushForwardOfT4} 
\mathcal T_*\PP^t=a_{14}\tilde\GG(1)+a_{24}\tilde\FF^1+a_{44}\tilde\PP^t, 
\\\label{eq:PushForwardOfT5} 
\mathcal T_*\DDD^t=a_{15}\tilde\GG(1)+a_{25}\tilde\FF^1+a_{35}\tilde\FF^2+a_{45}\tilde\PP^t+\tilde\DDD^t,
\\\label{eq:PushForwardOfT6}
\mathcal T_*\GG(\hat\psi)=\tilde\GG(\tilde\psi^{\hat\psi}), 
\\\label{eq:PushForwardOfT7} 
\mathcal T_*\DDD(\hat\varphi)=\tilde\GG(\tilde\psi^{\hat\varphi})+\tilde\DDD(\tilde\varphi^{\hat\varphi}), 
\end{gather}
where $a$'s are constants, $a_{11}a_{22}a_{33}a_{44}\ne0$ and 
$\hat\psi$ and $\hat\varphi$ are arbitrary smooth functions of~$x$. 
The constants~$a$'s completed with $a_{55}=1$ and $a_{ij}=0$, $1\leqslant i<j\leqslant5$, 
form a matrix of an automorphism of the megaideal~$\mathfrak m$.
Tildes over vector fields on the right hand sides of the above equations mean 
that these vector fields are written in terms of the transformed variables. 
The parameter-functions $\tilde\psi^{\hat\psi}$, $\tilde\psi^{\hat\varphi}$ and $\tilde\varphi^{\hat\varphi}$ 
are smooth functions of~$\tilde x$ whose values depend on the values of $\hat\psi$ or $\hat\varphi$ 
which is indicated by the corresponding superscripts.
We will derive constraints for~$\mathcal T$, consequently equating coefficients of vector fields 
in conditions \eqref{eq:PushForwardOfT1}--\eqref{eq:PushForwardOfT7} 
and taking into account constraints obtained in previous steps. 
As the coefficients of vector fields and components of the transformation~$\mathcal T$ 
associated with the derivatives $u_t$ and~$u_x$ are defined via first-order prolongation 
involving the similar values related to the variables~$t$, $x$ and~$u$, 
the coefficients of~$\p_{u_t}$ and~$\p_{u_x}$ give no essentially new equations 
in comparison with the coefficients of~$\p_t$, $\p_x$ and~$\p_u$. 
This is why we will not equate the coefficients of~$\p_{u_t}$ and~$\p_{u_x}$.
To have the same representation of the final result as in Theorem~\ref{thm:EquivalenceGroupIbragrimovClass}, 
we will re-denote certain values in an appropriate way. 

Thus, Eq.~\eqref{eq:PushForwardOfT1} implies that $T_u=X_u=0$, $U_u=c_2$ and $F_u=G_u=0$, 
where the nonvanishing constant~$a_{11}$ is re-denoted by~$c_2$. 
Then we derive from Eq.~\eqref{eq:PushForwardOfT2} that 
$tU_u=a_{22}T+a_{12}$, i.e.\ $T=c_1t+c_0$ where $c_1=c_2/a_{22}\ne0$ and $c_0=-a_{12}/a_{22}$, and $F_{u_t}=G_{u_t}=0$.
The consequence $t^2U_u=a_{33}T^2+a_{23}T+a_{13}$ of Eq.~\eqref{eq:PushForwardOfT3} gives only relations between $a$'s. 
In particular, $a_{33}=c_2/c_1^2$. 
The other consequences of Eq.~\eqref{eq:PushForwardOfT3} then are $F_g=0$ and~$G_g=c_2/c_1^2$. 
The essential consequences of Eq.~\eqref{eq:PushForwardOfT4} are exhausted by $X_t=0$, $U_t=a_{24}T+a_{14}$ and $F_t=G_t=0$.
Therefore, $X=\varphi(x)$ and $U=c_2+c_4t^2+c_3t+\psi(x)$, where $\varphi_x\ne0$, $c_4=a_{24}c_1/2$ and $c_3=a_{14}+a_{24}c_0$.

As we have already derived the precise expressions for the components of~$\mathcal T$ corresponding to the variables 
(cf.\ Eq.~\eqref{eq:EquivalenceGroupIbragimovClass}), 
at this point we can interrupt the computation of equivalence transformations by the algebraic method 
and calculate the expressions for~$F$ and~$G$ by the direct method. 
At the same time, all the determining equations for transformations from the equivalence group~$G^\sim$ 
of class~\eqref{eq:IbragimovClass} follow from restrictions for automorphisms of the equivalence algebra~$\mathfrak g^\sim$. 
This is not a common situation when the algebraic method is applied. 
Usually it gives only a part of the determining equations simplifying the subsequent application of the direct method. 
See, e.g., the computations of the complete point symmetry groups of 
the barotropic vorticity equation and the quasi-geostrophic two-layer model 
in \cite[Section~3]{bihl11Cy} and \cite[Section~4]{bihl11By}, respectively. 
This is why we complete the consideration of the equivalence group~$G^\sim$ within the framework of the algebraic method. 

From Eq.~\eqref{eq:PushForwardOfT5} we obtain in particular that $tU_t=a_{35}T^2+a_{25}T+a_{15}$, $fF_f=F$ and $fG_f+gG_g=G-a_{35}$, 
where $a_{35}=2c_4/c_1^2$ in view of the first of these consequences.

Eq.~\eqref{eq:PushForwardOfT6} implies the equations 
\begin{equation}\label{eq:ConditionsForArbitraryElementsByAlgabraicMethod1}
\hat\psi U_u=\tilde\psi^{\hat\psi},\quad 
\hat\psi_xF_{u_x}=0,\quad 
\hat\psi_xG_{u_x}-\hat\psi_{xx}fG_g=\tilde\psi^{\hat\psi}_{\tilde x\tilde x}F.
\end{equation}
The first and second equations of~\eqref{eq:ConditionsForArbitraryElementsByAlgabraicMethod1} are equivalent to 
$\tilde\psi^{\hat\psi}=c_2\hat\psi$ and $F_{u_x}=0$. 
Then we can express the derivative $\tilde\psi^{\hat\psi}_{\tilde x\tilde x}$ via derivatives of~$\hat\psi$,
$\tilde\psi^{\hat\psi}_{\tilde x\tilde x}=c_2\varphi_x^{-3}(\varphi_x\hat\psi_{xx}-\varphi_{xx}\hat\psi_x)$,
substitute the expression into the third equation of~\eqref{eq:ConditionsForArbitraryElementsByAlgabraicMethod1} 
and split with respect to the derivatives $\hat\psi_x$ and~$\hat\psi_{xx}$, 
as the function~$\hat\psi$ is arbitrary. 
As a result, we obtain $F=c_1^{-2}\varphi_x{}^2f$ and 
$G_{u_x}=c_2\varphi_x^{-3}\varphi_{xx}F$, i.e., $G_{u_x}=c_2c_1^{-2}\varphi_x^{-1}\varphi_{xx}f$. 
The expression for~$F$ coincides with the transformation component for~$f$ presented in Theorem~\ref{thm:EquivalenceGroupIbragrimovClass}.

The last essential equation for~$G$ is given by Eq.~\eqref{eq:PushForwardOfT7}. 
Collecting coefficients of $\p_x$, $\p_u$ and $\p_g$ in Eq.~\eqref{eq:PushForwardOfT7}, we have that 
$\tilde\varphi^{\hat\varphi}(\tilde x)=\varphi_x\hat\varphi$, $\tilde\psi^{\hat\varphi}(\tilde x)=\psi_x\hat\varphi$ and 
\begin{equation}\label{eq:ConditionsForArbitraryElementsByAlgabraicMethod2}
\hat\varphi G_x-\hat\varphi_xu_xG_{u_x}+2\hat\varphi_xfG_f+\hat\varphi_{xx}u_xfG_g=
\tilde\varphi^{\hat\varphi}_{\tilde x\tilde x}\tilde u_{\tilde x}F-\tilde\psi^{\hat\varphi}_{\tilde x\tilde x}F,
\end{equation} 
respectively.
We proceed in a way analogous to the previous step.
Namely, we express the derivatives $\tilde\varphi^{\hat\varphi}_{\tilde x\tilde x}$ 
and~$\tilde\psi^{\hat\varphi}_{\tilde x\tilde x}$ via derivatives of~$\hat\varphi$, 
substitute the expressions into Eq.~\eqref{eq:ConditionsForArbitraryElementsByAlgabraicMethod2} and  
split with respect to derivatives of $\hat\varphi$ because the function~$\hat\psi$ is arbitrary.
Equating the coefficients of $\hat\varphi_x$ leads to the equation
$fG_f=c_1^{-2}\varphi_x^{-1}(c_2u_x+\psi_x\varphi_{xx}-\psi_{xx}\varphi_x)$.

The simultaneous integration of all the equations obtained for~$G$ precisely results in 
the transformation component for~$g$ from Theorem~\ref{thm:EquivalenceGroupIbragrimovClass}.

\section{Determining equations for Lie symmetries}\label{sec:DetEqsForLieSymsIbragrimovClass}

Suppose that the vector field $Q = \tau(t,x,u)\p_t + \xi(t,x,u)\p_x+ \eta(t,x,u)\p_u$ belongs to the maximal Lie invariance algebra~$\mathfrak g^{\max}$ of an equation $\mathcal L$: $L=0$ from the class~\eqref{eq:IbragimovClass}, i.e.\ it is the generator of a one-parameter Lie symmetry group of the equation~$\mathcal L$. The criterion for infinitesimal invariance of $\mathcal L$ with respect to $Q$ is implemented using the second prolongation of $Q$, which reads
\[
 Q_{(2)} = Q + \eta^t\p_{u_t} + \eta^x\p_{u_x} + \eta^{tt}\p_{u_x} + \eta^{tx}\p_{u_{tx}}+\eta^{xx}\p_{u_{xx}}.
\]
The coefficients $\eta^t$, $\eta^x$, $\eta^{tt}$, $\eta^{xx}$ in $Q_{(2)}$ can be determined from the general prolongation formula for vector fields, see, e.g.~\cite{blum89Ay,olve86Ay,ovsi82Ay}. Using the second prolongation of $Q$, the infinitesimal invariance criterion reads $Q_{(2)}L|_{L=0}=0$, where the notation $|_{L=0}$ means that the condition $Q_{(2)}L$ is required to hold only on equations of class~\eqref{eq:IbragimovClass}. Applying the infinitesimal invariance condition to the class~\eqref{eq:IbragimovClass} then yields
\begin{equation}\label{eq:InfinitesimalInvarianceCriterionIbragimovClass}
 \eta^{tt}-(\xi f_x+\eta^xf_{u_x})u_{xx}-f\eta^{xx}-\xi g_x - \eta^x g_{u_x} =0 \qquad \textup{for} \qquad u_{tt}=fu_{xx}+g,
\end{equation}
where
\begin{gather*}
\eta^x = \DD_x(\eta-\tau u_t-\xi u_x) + \tau u_{tx}+\xi u_{xx}, \\
\eta^{xx} = \DD_x^2(\eta-\tau u_t-\xi u_x) + \tau u_{txx}+\xi u_{xxx}, \\
\eta^{tt} = \DD_t^2(\eta-\tau u_t-\xi u_x) + \tau u_{ttt}+\xi u_{ttx},
\end{gather*}
$\DD_t$ and $\DD_x$ denote the operators of total differentiation with respect to $t$ and $x$, respectively, which in the present case of one dependent variable are given by
\begin{gather*}
 \DD_t = \p_t + u_t\p_{u}+u_{tt}\p_{u_{t}} + u_{tx}\p_{u_{x}}+\cdots, \\ 
 \DD_x = \p_x + u_x\p_{u}+u_{tx}\p_{u_{t}} + u_{xx}\p_{u_{x}}+\cdots.
\end{gather*}
Expanding the infinitesimal invariance condition~\eqref{eq:InfinitesimalInvarianceCriterionIbragimovClass} we obtain
\begin{gather}\label{eq:InfinitesimalInvarianceCriterionIbragimovClassExpanded}
\begin{split}
 &\DD_t^2\eta - u_t\DD_t^2\tau - u_x\DD_t^2\xi - 2u_{tt}\DD_t\tau - 2u_{tx}\DD_t\xi\\ 
 &=f(\DD_x^2\eta - u_t\DD_x^2\tau - u_x\DD_x^2\xi - 2u_{tx}\DD_x\tau - 2u_{xx}\DD_x\xi)\\ 
 & + (\xi f_x + (\DD_x\eta -u_t\DD_x\tau - u_x\DD_x\xi)f_{u_x})u_{xx} + \xi g_x + (\DD_x\eta-u_t\DD_x\tau-u_x\DD_x\xi)g_{u_x}  
\end{split}
\end{gather}
where we have to substitute $u_{tt}=fu_{xx}+g$. The above equation can be split with respect to the derivatives of~$u_{tx}$, $u_{xx}$ and $u_t$. Collecting the coefficients of $u_{tx}u_t$, $u_{tx}$ and $u_{xx}u_t$, we produce
\[
    \xi_u=0, \quad \xi_t = f(\tau_t+\tau_u u_x), \quad 2f\tau_u = (\tau_x+\tau_u u_x)f_{u_x}.
\]
Supposing that $\xi_t=0$, it immediately follows from the second equation that $\tau_u=0$. Otherwise, for $\xi_t\ne0$ we can solve the second equation for~$f$ and substitute the obtained expression into the third equation. After simplification we have that $\xi_t\tau_u=0$, i.e.\ $\tau_u=0$. Therefore, we always have
\[
 \xi_u=0, \quad \tau_u=0, \quad \xi_t = f\tau_x, \quad \tau_xf_{u_x}=0.
\]
Further splitting of Eq.~\eqref{eq:InfinitesimalInvarianceCriterionIbragimovClassExpanded} and taking into account the above restrictions gives
\begin{equation}\label{eq:DeterminingEquationsIbragimovClass}
\arraycolsep=0ex
\begin{array}{ll}
 u_t^2\colon&\quad \eta_{uu}=0 \\[1ex]
 u_{xx}\colon&\quad 2(\tau_t-\xi_x)f+\xi f_x+(\eta_x+(\eta_u -\xi_x) u_x)f_{u_x}=0 \\[1ex]
 u_t\colon&\quad 2\eta_{tu}-\tau_{tt}+\tau_{xx}f+\tau_xg_{u_x} = 0 \\[1ex]
 \textup{Rest}\colon\quad &\quad \eta_{tt}-\xi_{tt}u_x - (\eta_{xx}+(2\eta_{ux}-\xi_{xx})u_x)f +(\eta_u-2\tau_t)g\\[1ex]
 &\qquad {}-\xi g_x - (\eta_x+(\eta_u-\xi_x)u_x)g_{u_x} = 0.
\end{array}
\end{equation}
The equations $\xi_u=0$, $\tau_u=0$ and $\eta_{uu}=0$ neither involve the arbitrary elements~$f$ or~$g$ nor their derivatives. This is why they can be immediately integrated and give restricting conditions on Lie symmetries valid for all equations of the form~\eqref{eq:IbragimovClass}. In particular, we have $\eta=\eta^1(t,x)u+\eta^0(t,x)$. 

In order to derive the kernel of Lie symmetries of class~\eqref{eq:IbragimovClass}, we can further split the classifying part of the determining equations~\eqref{eq:DeterminingEquationsIbragimovClass} with respect to the arbitrary elements and their derivatives. This immediately gives that the kernel algebra is 
\begin{equation}\label{eq:KernelIbragimovClass}
  \mathfrak g^\cap =\langle\p_t, \p_u, t\p_u\rangle,
\end{equation}
which is a realization of the three-dimensional (nilpotent) Heisenberg algebra. Consequently, the Lie symmetries admitted by each equation from the class~\eqref{eq:IbragimovClass} are exhausted by transformations of the form $(t,x,u)\mapsto(t+\ve_1,x,u+\ve_2+\ve_3t)$, where $\ve_1$, $\ve_2$ and $\ve_3$ are arbitrary constants.

Up to this point the nonlinearity of the equations under consideration was of no importance. Only the general form~\eqref{eq:IbragimovClass} was essential.
Now we should start to exploit the nonlinearity condition~$(f_{u_x},g_{u_xu_x})\ne(0,0)$, which is included in the definition of class~\eqref{eq:IbragimovClass}.

First assume that~$f_{u_x}=0$ and therefore~$g_{u_xu_x}\ne0$. 
Differentiating the third equation of system~\eqref{eq:DeterminingEquationsIbragimovClass} with respect to $u_x$, we then immediately find that $\tau_x=0$. In view of the equation $\xi_t=f\tau_x$ we also have $\xi_t=0$. 
Upon differentiating the second equation of~\eqref{eq:DeterminingEquationsIbragimovClass} with respect to~$t$ we obtain~$\tau_{tt}=0$. 
The third equation of~\eqref{eq:DeterminingEquationsIbragimovClass} then implies $\eta_{tu}=0$. 
Finally, we differentiate the last equation in~\eqref{eq:DeterminingEquationsIbragimovClass} with respect to $u$ and $u_x$ (resp.\ $t$ and $u_x$, resp.\ $t$). 
This gives $\eta_{xu}=0$ (resp.\ $\eta_{xt}=0$, resp.\  $\eta_{ttt}=0$).

Now we assume that~$f_{u_x}\ne0$. In this case, the equation~$\xi_t=f\tau_x$ can be split to yield $\xi_t=\tau_x=0$. The third equation of system~\eqref{eq:DeterminingEquationsIbragimovClass} then implies $2\eta_{tu}=\tau_{tt}$. 
Differentiating the second equation of~\eqref{eq:DeterminingEquationsIbragimovClass} with respect to~$u$ we obtain~$\eta_{xu}=0$. 
The differentiation of the last equation in~\eqref{eq:DeterminingEquationsIbragimovClass} with respect to~$u$ then yields~$\eta_{ttu}=0$. 
In view of the equation $2\eta_{tu}=\tau_{tt}$ we obviously have~$\tau_{ttt}=0$.
Differentiating the second equation of~\eqref{eq:DeterminingEquationsIbragimovClass} twice with respect to~$t$ leads to~$\eta_{ttx}=0$.

Collecting the results from the above two cases, for the class~\eqref{eq:IbragimovClass}, whose definition includes the condition~$(f_{u_x},g_{u_xu_x})\ne(0,0)$, we always have
\begin{equation}\label{eq:DeterminingEquationsIbragimovClassSimplified}
  \tau_u=\tau_x=\xi_u=\xi_t=\eta_{uu}=\eta_{xu}=\eta_{ttx}=\tau_{ttt}=0, \quad 2\eta_{tu}=\tau_{tt}.
\end{equation}
Hence only the second and fourth equations of~\eqref{eq:DeterminingEquationsIbragimovClass} 
essentially involve arbitrary elements and are really classifying determining equations for the class~\eqref{eq:IbragimovClass}.
They must be solved up to the equivalence relation induced by transformations from~$G^\sim$. 
Also note that for $f_{u_x}\ne0$ and~$\tau_{tt}=0$ we find by differentiating both these equations with respect to~$t$ that $\eta_{tx}=0$ and $\eta_{ttt}=0$.

This completes the proof of the following proposition:

\begin{proposition}
For each equation from class~\eqref{eq:IbragimovClass}, any symmetry operator~$Q$ with $\tau_{tt}=0$ lies 
in the projection of the equivalence algebra~$\mathfrak g^\sim$ to the space of equation variables, i.e.\ $Q\in\mathrm P\mathfrak g^\sim$.
\end{proposition}

It thus remains to investigate the case in which $f_{u_x}\ne0$ and the corresponding maximal Lie invariance algebra~$\mathfrak g^{\max}=\mathfrak g^{\max}(f,g)$ contains a vector field~$\breve Q$ with $\tau_{tt}\ne0$. 
In view of system~\eqref{eq:DeterminingEquationsIbragimovClassSimplified} the general form of vector fields from~$\mathfrak g^{\max}$ is 
\[Q=(a_2t^2+a_1t+a_0)\p_t+\xi(x)\p_x+((a_2t+b_1)u+\eta^0(t,x))\p_u,\] 
where the constants $a_0$, $a_1$,  $a_2$ and~$b_1$ and the functions $\xi=\xi(x)$ and $\eta^0=\eta^0(t,x)$, where $\eta^0_{ttx}=0$,
are additionally constrained in such a way that 
the coefficients $\tau=a_2t^2+a_1t+a_0$, $\xi=\xi(x)$ and $\eta=(t+b_1)u+\eta^0(t,x)$ also satisfy the second and last equations of system~\eqref{eq:DeterminingEquationsIbragimovClass}. 
For convenience, we will mark the values of coefficients and parameters corresponding to the vector field~$\breve Q$ by breve. 
As $\breve a_2\ne0$, by scaling of~$\breve Q$ we can set $\breve a_2=1$. 
As the vector field~$\p_t$ belongs to~$\mathfrak g^\cap$, it necessarily is in~$\mathfrak g^{\max}$ for any~$f$ and~$g$. 
Therefore, the algebra~$\mathfrak g^{\max}$ contains also the commutator $[\p_t,\breve Q]$. 
Linearly combining~$\breve Q$ with~$\p_t$ and $\tilde Q=[\p_t,\breve Q]$ we can also set $\breve a_0=\breve a_1=0$. 
Hence $\tilde Q=2t\p_t+(u+\breve\eta^0_t)\p_u$.

Substituting the coefficients of~$\tilde Q$ into the second equation of system~\eqref{eq:DeterminingEquationsIbragimovClass}, 
we obtain the equation $4f+(u_x+\beta)f_{u_x}=0$, where $\beta=\breve\eta^0_{tx}$ should be a smooth function depending at most on~$x$ since $\eta^0_{ttx}=0$ for any operator~$Q$ from~$\mathfrak g^{\max}$. 
The general solution of this equation is $f=\alpha(x)(u_x+\beta(x))^{-4}$, where~$\alpha$ is an arbitrary function of~$x$. 
Using transformations from the equivalence group~$G^\sim$ we can simplify $f$ and set $\alpha=\pm1$ and $\beta=0$. 
If we plug the form $f=\pm u_x^{-4}$ into the second equation of system~\eqref{eq:DeterminingEquationsIbragimovClass}, 
we obtain that $\tau_t-\xi_x = 2(\eta_u-\xi_x)+2\eta_xu_x^{-1}$ for an arbitrary operator from~$\mathfrak g^{\max}$. 
From this condition, we can immediately conclude that $\eta_x=0$ and $\xi_x=2b_1$, i.e.\ $\xi=2b_1x+b_0$ for some constant~$b_0$. 

As $\breve\eta^0_x=0$, the substitution of the coefficients of~$\tilde Q$ into the last equation of system~\eqref{eq:DeterminingEquationsIbragimovClass} 
gives the equation $u_xg_{u_x}+3g=\breve\eta^0_{ttt}$ with separated variables. 
Both the sides of this equation are equal to a constant which can be set to zero by a transformation $\mathscr F^2(c_4)$ from~$G^\sim$.
The equation $u_xg_{u_x}+3g=0$ is equivalent to the representation $g=\mu(x)u_x^{-3}$, where $\mu$ is an arbitrary function of~$x$. 
For this expression of~$g$ the last equation of system~\eqref{eq:DeterminingEquationsIbragimovClass} takes the form
\[
 \eta^0_{tt}-2b_1\mu u_x^{-3} - (2b_1x+b_0)\mu_xu_x^{-3}=0
\]
and the subsequent splitting with respect to~$u_x$ implies that $\eta^0_{tt}=0$ and 
$
(\mu(2b_1x+b_0))_x=0.
$
We now distinguish the following cases for values of~$b_0$ and~$b_1$ depending on a value of~$\mu$: 

0. $\mu$ is arbitrary. In this case $b_1=b_0=0$.
 
1. $\mu$ is a nonzero constant. Then $b_0$ is arbitrary and $b_1=0$. Using an equivalence transformation, we can scale $\mu$ to one. 

2. $\mu = \nu x^{-1}\bmod G^\sim$, where $\nu$ is a nonzero constant. (A constant summand of~$x$ can be set equal to 0 by a shift of~$x$.) 
For this value of~$\mu$ we have $b_0=0$ and $b_1$ is arbitrary. 

3. $\mu=0$. This implies that $b_1$ and $b_2$ are arbitrary.

We denote by~$\mathcal K$ the subclass of equations from the class~\eqref{eq:IbragimovClass}, 
which are $G^\sim$-equivalent to equations with $f=\pm u_x^{-4}$ and $g=\mu(x)u_x^{-3}$ and by~$\bar{\mathcal K}$ the complement of this subclass in the class~\eqref{eq:IbragimovClass}.
The above consideration shows that only equations from the subclass~$\mathcal K$ admit Lie symmetry operators that are not contained in $\mathrm P \mathfrak g^\sim$. 
In other words, the following theorem is true:

\begin{theorem}\label{thm:WeakNormalizationIbragimovClass}
The subclass $\bar{\mathcal K}$ of class~\eqref{eq:IbragimovClass} that is singled out by the condition
\[
    (f,g)\ne(\pm u_x^{-4},\mu(x)u_x^{-3})\bmod G^\sim,
\]
where~$\mu(x)$ is an arbitrary function of $x$ is weakly normalized.
\end{theorem}

\begin{remark}\label{rem:SubclassAndItsComplementOfIbragimovClass}
The sets~$\mathcal K$ and~$\bar{\mathcal K}$ of equations are really subclasses of the class~\eqref{eq:IbragimovClass} since the condition $(f,g)=(\pm u_x^{-4},\mu(x)u_x^{-3})\bmod G^\sim$ and its negation are equivalent to systems of equations and/or inequalities with respect to the arbitrary elements~$f$ and~$g$. Indeed, by acting on the arbitrary elements $f=\pm u_x^{-4}$ and $g=\mu(x)u_x^{-3}$ with transformations from $G^\sim$ and eliminating the involved group parameters and the parameter-function~$\mu$, we arrive at a system of differential equations in~$f$ and~$g$ characterizing the subclass~$\mathcal K$. Namely, the subclass~$\mathcal K$ is singled out from the class~\eqref{eq:IbragimovClass} by the system
\[
V_{u_x}=1,\quad W_{xu_x}(V^3)_{u_x}-W_{u_x}(V^3)_{xu_x}=0,\quad W_{u_xu_x}(V^3)_{u_x}-W_{u_x}(V^3)_{u_xu_x}=0,
\]
where $V=-4f/f_{u_x}$ and $W=V^3(g+fV_x+f_xV/2)$. This implies that the subclass~$\bar{\mathcal K}$ as the complement of~$\mathcal K$ is defined by the inequality
\[
\big(V_{u_x}-1\big)^2+\big(W_{xu_x}(V^3)_{u_x}-W_{u_x}(V^3)_{xu_x}\big)^2+\big(W_{u_xu_x}(V^3)_{u_x}-W_{u_x}(V^3)_{u_xu_x}\big)^2\ne0.
\]
\end{remark}

The above Cases 0--3 represent the complete group classification of equations from the subclass~$\mathcal K$ up to $G^\sim$-equivalence.
Recall that by the definition of the subclass~$\mathcal K$ any equation from this subclass is $G^\sim$-equivalent to an equation with $f=\pm u_x^{-4}$ and $g=\mu(x)u_x^{-3}$.

\begin{lemma}\label{lem:IbragimovClassSpecialLieSymExts1}
A complete list of $G^\sim$-inequivalent Lie symmetry extensions for equations of the general form
\begin{equation}\label{eq:IbragimovClassSpecialLieSymExts1Subclass}
u_{tt} = \pm u_x^{-4}u_{xx} + \mu(x)u_x^{-3},
\end{equation}
where $\mu$ runs through the set of smooth functions depending on~$x$, is exhausted by the following cases:
\begin{gather}\label{eq:IbragimovClassSpecialLieSymExts1SubclassList}\hspace*{-\arraycolsep}%
\begin{array}{lll}
0.& \text{arbitrary}\ \mu\colon & \mathfrak g^\cap_1 = \mathfrak g^\cap + \langle t^2\p_t+tu\p_u,\,2t\p_t+u\p_u\rangle,\\
1.& \mu=1\colon & \mathfrak g^\textup{max}=\mathfrak g^\cap_1 +\langle\p_x\rangle,\\
2.& \mu=\nu x^{-1},\ \nu\ne0\colon & \mathfrak g^\textup{max}=\mathfrak g^\cap_1 +\langle2x\p_x+u\p_u\rangle,\\
3.& \mu=0\colon & \mathfrak g^\textup{max}=\mathfrak g^\cap_1 +\langle\p_x,\,2x\p_x+u\p_u\rangle.
\end{array}
\end{gather}
\end{lemma}

\begin{remark}
We can use equations of the general form
\begin{equation}\label{eq:IbragimovClassSpecialLieSymExts1EquivSubclass}
u_{tt}=\theta(x)u_x^{-4}u_{xx} 
\end{equation}
as canonical representatives of elements from the class~$\mathcal K$ instead of~\eqref{eq:IbragimovClassSpecialLieSymExts1Subclass}. 
Indeed, each equation from the subclass~\eqref{eq:IbragimovClassSpecialLieSymExts1Subclass} 
is mapped to an equation from the subclass~\eqref{eq:IbragimovClassSpecialLieSymExts1EquivSubclass} by the transformation~$\mathscr D(\varphi)$, 
where $\varphi_{xx}\pm\mu\varphi_x=0$ and $\theta(\tilde x)=\pm(\varphi_x(x))^{-2}$. 
(Here and in what follows all $\pm$ and $\mp$ are consistent with those from Lemma~\ref{lem:IbragimovClassSpecialLieSymExts1}.)
In other words, we construct a point-transformation mapping~\cite{vane09Ay} between 
the subclasses~\eqref{eq:IbragimovClassSpecialLieSymExts1Subclass} and~\eqref{eq:IbragimovClassSpecialLieSymExts1EquivSubclass}
which is generated by a family of equivalence transformations parameterized by the arbitrary element~$\mu$. 
Hence, mapping and rearrangement of the classification list~\eqref{eq:IbragimovClassSpecialLieSymExts1SubclassList} lead to the equivalent list based on 
the canonical representative form~\eqref{eq:IbragimovClassSpecialLieSymExts1EquivSubclass}:
\begin{gather}\label{eq:IbragimovClassSpecialLieSymExts1EquivSubclassList}\hspace*{-\arraycolsep}%
\begin{array}{lll}
0.& \text{arbitrary}\ \theta\colon & \mathfrak g^\cap_1 = \mathfrak g^\cap + \langle t^2\p_t+tu\p_u,\,2t\p_t+u\p_u\rangle,\\
1.& \theta=\pm e^{2x}\colon & \mathfrak g^\textup{max}=\mathfrak g^\cap_1 +\langle2\p_x+u\p_u\rangle,\\
2.& \theta=\pm|x|^{2p},\ p\ne0\colon & \mathfrak g^\textup{max}=\mathfrak g^\cap_1 +\langle2x\p_x+(p+1)u\p_u\rangle,\\
3.& \theta=\pm1\colon & \mathfrak g^\textup{max}=\mathfrak g^\cap_1 +\langle\p_x,\,2x\p_x+u\p_u\rangle.
\end{array}
\end{gather}
Cases~0, 1, $2|_{\nu=\pm1}$, $2|_{\nu\ne\pm1}$ and 3 of the list~\eqref{eq:IbragimovClassSpecialLieSymExts1SubclassList} are mapped to 
Cases~0, $2|_{p=-1}$, 1, $2|_{p=\nu/(\nu\mp1)}$ and 3 of the list~\eqref{eq:IbragimovClassSpecialLieSymExts1EquivSubclassList}, respectively. 
Each of the classification lists has certain advantages. 
Thus, the form~\eqref{eq:IbragimovClassSpecialLieSymExts1EquivSubclass} is more compact than~\eqref{eq:IbragimovClassSpecialLieSymExts1Subclass}. 
At the same time, basis elements of the algebras presented in the list~\eqref{eq:IbragimovClassSpecialLieSymExts1SubclassList} do not depend, 
in contrast to Case~2 of~\eqref{eq:IbragimovClassSpecialLieSymExts1EquivSubclassList}, on equation parameters. 
The equation associated with Case~1 of the list~\eqref{eq:IbragimovClassSpecialLieSymExts1SubclassList}
does not explicitly involve the independent variable~$x$, 
as opposed to its image given in Case~$2|_{p=-1}$ of the list~\eqref{eq:IbragimovClassSpecialLieSymExts1EquivSubclassList}
whose value of the arbitrary element~$\theta$ equals $\pm x^{-2}$. 
\end{remark}

\begin{remark}
Due to Theorem~\ref{thm:WeakNormalizationIbragimovClass}, to complete the group classification of the class~\eqref{eq:IbragimovClass} it is enough to investigate symmetry extensions induced by subalgebras of the equivalence algebra~$\mathfrak g^\sim$. The corresponding Lie symmetry generators satisfy the following simplified determining equations:
\begin{align}\label{eq:DeterminingEqsSimplifiedIbragimovClass}
\begin{split}
 &\tau_u=\tau_x=\tau_{tt}=\xi_u=\xi_t=\eta_{uu}=\eta_{xu}=\eta_{tx}=\eta_{tu}=\eta_{ttt}=0,\\
 &\xi f_x+((\eta_u-\xi_x)u_x+\eta_x)f_{u_x}=2(\xi_x-\tau_t)f, \\
 &\xi g_x+((\eta_u-\xi_x)u_x+\eta_x)g_{u_x}=(\eta_u-2\tau_t)g+(\xi_{xx}u_x-\eta_{xx})f +\eta_{tt}.
\end{split}
\end{align}
\end{remark}

\begin{remark}
Lemma~\ref{lem:IbragimovClassSpecialLieSymExts1} obviously implies that the entire class~\eqref{eq:IbragimovClass} is not weakly normalized. 
This can also be proved without the study of the subclass structure, 
by the direct computation of the union $\mathfrak g^\cup$ of the maximal Lie invariance algebras of equations from the class~\eqref{eq:IbragimovClass}.
The set~$\mathfrak g^\cup$ consists of vector fields of the form $\tau(t,x,u)\p_t + \xi(t,x,u)\p_x+ \eta(t,x,u)\p_u$ 
for which the whole system of determining equations and the nonvanishing condition $(f_{u_x},g_{u_xu_x})\ne(0,0)$
are consistent with respect to the functions $f=f(x,u_x)$ and $g=g(x,u_x)$. 
The consistency condition is the joint system of~\eqref{eq:DeterminingEquationsIbragimovClassSimplified} and 
\begin{equation}\label{eq:IbragimovClassSpecificSystemForUnionOfMIA}
\begin{split}
&\eta_{tx}(\eta_u-\xi_x)=\eta_x\eta_{tu}+\xi\eta_{txx},\\
&\eta_{tu}(\xi_{xx}+\eta_{xx})=\eta_{txx}(\eta_u-2\tau_t+\xi),\\
&\eta_{ttt}(\eta_u-2\tau_t)=\eta_{tt}(\eta_{tu}-2\tau_{tt}).
\end{split}
\end{equation}
It is clear that $\mathfrak g^\cup$ is not contained in the projection $\mathrm P\mathfrak g^\sim$ of the equivalence algebra~$\mathfrak g^\sim$, 
which is associated with the solution set of the system given in the first row of~\eqref{eq:DeterminingEqsSimplifiedIbragimovClass}.
\end{remark}

\section{Set of admissible transformations}\label{sec:SetOfAdmTrans}

After we have established the equivalence group of the class~\eqref{eq:IbragimovClass}, we can describe the set of admissible transformations of this class in terms of its normalized subclasses. 
Theorem~\ref{thm:WeakNormalizationIbragimovClass} and its proof give us hints on feasible ways for the classification of admissible transformations.

First we assume that $f_{u_x}=0$ and therefore $g_{u_xu_x}\ne0$. We differentiate Eq.~\eqref{eq:DeterminingEqsEquivalenceTransformations3} with respect to $\tilde u_{\tilde t}$ and $\tilde u_{\tilde x}$ and take into account Eq.~\eqref{eq:DeterminingEqsEquivalenceTransformations1Reminder}. From the obtained equation
\[
 g_{u_xu_x}\frac{T_xX_x}{U_u^{\,2}}=0
\]
and the inequality $g_{u_xu_x}\ne0$ we can conclude that $T_xX_x=0$. Then Eq.~\eqref{eq:DeterminingEqsEquivalenceTransformations1Reminder} also implies $T_tX_t=0$.

Suppose that $T_t=0$. Consequently, in view of the nondegeneracy condition of point transformations we have $T_x\ne0$ and $X_t\ne0$ and therefore $X_x=0$.
The expressions of~$u_t$ and~$u_x$ via~$\tilde u_{\tilde t}$ and~$\tilde u_{\tilde x}$ take the form 
$u_t=(X_t\tilde u_{\tilde x}-U_t)/U_u$ and $u_x=(T_x\tilde u_{\tilde t}-U_x)/U_u$, i.e., 
the expression of~$u_t$ (resp.\ $u_x$) does not involve~$\tilde u_{\tilde t}$ (resp.\ $\tilde u_{\tilde x}$). 
We differentiate Eq.~\eqref{eq:DeterminingEqsEquivalenceTransformations3} twice with respect to $\tilde u_{\tilde x}$ and once with respect to~$u$.
In view of the supposition $T_t=0$, this gives $(U_{uu}/U_u^{\,2})_u=0$.
Then we differentiate Eq.~\eqref{eq:DeterminingEqsEquivalenceTransformations3} twice with respect to $\tilde u_{\tilde t}$:
\begin{equation}\label{eq:DeterminingEqsEquivalenceTransformations3utut}
 g_{u_xu_x} = 2f\frac{U_{uu}}{U_u}.
\end{equation}
The subsequent differentiation of Eq.~\eqref{eq:DeterminingEqsEquivalenceTransformations3utut} with respect to~$u$ gives the equation $(U_{uu}/U_u)_u=0$, which together with $(U_{uu}/U_u^{\,2})_u=0$ implies that $U_{uu}=0$. Then Eq.~\eqref{eq:DeterminingEqsEquivalenceTransformations3utut} is reduced to $g_{u_xu_x}=0$ and therefore leads to a contradiction.

This is why we necessarily have $T_tX_x\ne0$ and consequently $X_t=T_x=0$. In view of Eq.~\eqref{eq:DeterminingEqsEquivalenceTransformations2}, we also obtain the equation $\tilde f = fX_x^2/T_t^2$ from which we can conclude by differentiation with respect to~$t$ that $T_{tt}=0$. It is also evident that $\tilde f_{\tilde u_{\tilde x}}=0$.

Owing to the restrictions derived so far, it is now possible to split Eq.~\eqref{eq:DeterminingEqsEquivalenceTransformations3} with respect to~$u_t$. The coefficient of $u_t^2$ gives $U_{uu}=0$ and that of $u_t$ leads to $U_{tu}=0$. The rest of Eq.~\eqref{eq:DeterminingEqsEquivalenceTransformations3} is
\[
  \tilde g T_t^2 -U_{tt}=f(\tilde u_{\tilde x}X_{xx}-U_{xx}-2U_{xu}u_x)+ gU_u.
\]
This obviously implies that $\tilde g_{\tilde u_{\tilde x}\tilde u_{\tilde x}}\ne0$ since $g_{u_xu_x}\ne0$. We will successively differentiate the above rest with respect to three combinations of variables, $(u,\tilde u_{\tilde x})$, $(t,\tilde u_{\tilde x})$ and $t$, which gives $U_{xu}=0$, $U_{tx}=0$ and $U_{ttt}=0$, respectively.

Summing up, for the components $T$, $X$ and $U$ of admissible transformations of any equation with $f_{u_x}=0$ and $g_{u_xu_x}\ne0$ within the class~\eqref{eq:IbragimovClass} we derive the same system of determining equations as in the case of equivalence transformations, cf.~\eqref{eq:DeterminingEqsEquivalenceTransformationsFinalSystem}. Moreover, the conditions $f_{u_x}=0$ and $g_{u_xu_x}\ne0$ are saved by the admissible transformations. In this way, we have established the following theorem:

\begin{theorem}\label{thm:IbragimovClassNormalizedSubclass}
The subclass of class~\eqref{eq:IbragimovClass} which is singled out by the constraints $f_{u_x}=0$ and $g_{u_xu_x}\ne0$ is saved by admissible point transformations within the class~\eqref{eq:IbragimovClass}. This subclass is normalized and its equivalence group coincides with the equivalence group $G^\sim$ of the entire class~\eqref{eq:IbragimovClass}.
\end{theorem}

It now remains to investigate the case $f_{u_x}\ne0$. Eq.~\eqref{eq:DeterminingEqsEquivalenceTransformations1Reminder} immediately implies that $T_tX_t=T_xX_x=0$.

Supposing $T_x\ne0$, we obtain that $X_x=0$, $X_t\ne0$ and hence $T_t=0$. In view of these conditions Eq.~\eqref{eq:DeterminingEqsEquivalenceTransformations2} is reduced to $X_t^2=\tilde ffT_x^2$. Differentiating the last equation with respect to~$u_x$ leads to the equation $\tilde ff_{u_x}T_x^2=0$, which is equivalent to the equation $T_x=0$, contradicting the initial supposition.

Therefore, we have $T_x=0$, $T_t\ne0$, $X_t=0$ and $X_x\ne0$ and Eq.~\eqref{eq:DeterminingEqsEquivalenceTransformations2} reads
\begin{gather}\label{eq:DeterminingEqsEquivalenceTransformations5}
 \tilde fT_t^2 = fX_x^2.
\end{gather}
As the transformation rules for the first derivatives are simplified to
\[
 \tilde u_{\tilde t} = \frac{U_t+U_uu_t}{T_t},\quad \tilde u_{\tilde x} = \frac{U_x+U_uu_x}{X_x},
\]
we can conclude from Eq.~\eqref{eq:DeterminingEqsEquivalenceTransformations5} that $\tilde f_{\tilde u_{\tilde x}}=0$ if and only if $f_{u_x}=0$.

Differentiating Eq.~\eqref{eq:DeterminingEqsEquivalenceTransformations5} with respect to~$u$ gives
$
 (U_{xu}+U_{uu}u_x)\tilde f_{\tilde u_{\tilde x}} =0,
$
and therefore $U_{uu}=0$ and $U_{xu}=0$. Differentiating Eq.~\eqref{eq:DeterminingEqsEquivalenceTransformations5} with respect to~$t$ results in the equation
\begin{equation}\label{eq:DeterminingEqsEquivalenceTransformations5t}
 \left(\frac{U_{tu}}{X_x}\left(\frac{X_x\tilde u_{\tilde x}-U_x}{U_u}\right)+\frac{U_{xt}}{X_x}\right)\tilde f_{\tilde u_{\tilde x}}+2\tilde f\frac{T_{tt}}{T_t} = 0.
\end{equation}
Taking into account the simplifications obtained so far, we represent Eq.~\eqref{eq:DeterminingEqsEquivalenceTransformations3} in the reduced form
\[
 \tilde g T_t^2 + \frac{U_t+U_uu_t}{T_t}T_{tt} - U_{tt}-2U_{tu}u_t =  f\left(\frac{U_x+U_uu_x}{X_x} X_{xx}-U_{xx}\right)+ gU_u.
\]
This last equation can be split with respect to $u_t$, giving the equations
\begin{equation}\label{eq:DeterminingEqsEquivalenceTransformations6}
\arraycolsep=0ex
\begin{array}{l}
\dfrac{U_u}{T_t}T_{tt} =2U_{tu},\\[1ex]
\tilde g T_t^2 + \dfrac{U_t}{T_t}T_{tt} - U_{tt} =  f\left(\dfrac{U_x+U_uu_x}{X_x} X_{xx}-U_{xx}\right)+ gU_u.
\end{array}
\end{equation}
We now distinguish the two cases $T_{tt}=0$ and $T_{tt}\ne0$.

In the case of $T_{tt}=0$, the first of the above equations implies $U_{tu}=0$. The corresponding form of Eq.~\eqref{eq:DeterminingEqsEquivalenceTransformations5t} then leads to $U_{tx}=0$. Differentiating the second equation of~\eqref{eq:DeterminingEqsEquivalenceTransformations6} with respect to $t$ yields $U_{ttt}=0$. Collecting all the results for this case implies that the transformation belongs to the equivalence group $G^\sim$.

We now investigate the case of $T_{tt}\ne0$. Then, we solve the first equation of~\eqref{eq:DeterminingEqsEquivalenceTransformations6} with respect to $U_{tu}/U_u$ and plug the resulting expression into Eq.~\eqref{eq:DeterminingEqsEquivalenceTransformations5t}. This yields
\[
 \left(\frac{T_{tt}}{2T_t}\left(\tilde u_{\tilde x}-\frac{U_x}{X_x}\right)+\frac{U_{xt}}{X_x}\right)\tilde f_{\tilde u_{\tilde x}}+2\frac{T_{tt}}{T_t}\tilde f = 0,
\]
or,
\[
\left(\tilde u_{\tilde x}+\frac{U_{xt}}{X_x}\frac{2T_t}{T_{tt}}-\frac{U_x}{X_x}\right)\tilde f_{\tilde u_{\tilde x}}+4\tilde f = 0.
\]
The difference of the second and third terms in the bracket can be encapsulated as a function of~$x$ (or, equivalently, $\tilde x$), i.e.\ we can write
$
(\tilde u_{\tilde x}+\tilde \alpha(\tilde x))\tilde f_{\tilde u_{\tilde x}}+4\tilde f = 0.
$
This implies that
\begin{equation}\label{eq:DeterminingEqsEquivalenceTransformationsExpressionf}
    \tilde f = \frac{\tilde \beta(\tilde x)}{(\tilde u_{\tilde x}+\tilde \alpha(\tilde x))^4}.
\end{equation}
We now differentiate the second equation of~\eqref{eq:DeterminingEqsEquivalenceTransformations6} with respect to~$u$, which gives
\[
U_{ttu} = \frac{U_{ut}T_{tt}}{T_t} = \frac12\left(\frac{U_uT_{tt}}{T_t}\right)_t,
\]
where the second equality holds upon differentiating the first equation in~\eqref{eq:DeterminingEqsEquivalenceTransformations6} with respect to $t$. This implies that
\[
 \frac{U_{tu}T_{tt}}{T_t} - U_u\left(\frac{T_{tt}}{T_t}\right)_t = 0,
\]
which is equivalent to $\left(U_uT_t/T_{tt}\right)_t=0$. 
Integrating this equation gives an expression for $U_u$: $U_u = \varkappa T_{tt}/T_t$, where $\varkappa$ is a constant.
We substitute the expression for $U_u$ into the first equation of~\eqref{eq:DeterminingEqsEquivalenceTransformations6} to obtain
$2T_{ttt}T_t-3T_{tt}^2=0$.
The general solution of the last equation is
\[
T=\frac{a_1t+a_0}{a_3t+a_2},
\]
were $a_i$, $i=0,\dots,3$, are constants with $a_1a_2-a_0a_3\ne0$ which are determined up to a common nonvanishing multiplier. As $T_{tt}\ne0$, we moreover have $a_3\ne0$ and can assume $a_3=1$ due to the indeterminacy of the constant multiplier. Then we successively gauge $a_2$, $a_0$ and $a_1$ to 0, 1 and 0 by a shift of~$t$, a scaling of~$t$ and a shift of~$\tilde t$, respectively. All the above transformations belong to the group~$G^\sim$. In other words, $T=1/t\bmod G^\sim$. Plugging the expression obtained for~$T$ into the equation $U_u=\varkappa T_{tt}/T_t$ allows deriving that $U_u=\hat q/t$, where $\hat q$ is a nonzero constant.

Combining Eq.~\eqref{eq:DeterminingEqsEquivalenceTransformations5} with the expression for $\tilde f$ established in Eq.~\eqref{eq:DeterminingEqsEquivalenceTransformationsExpressionf} yields
\[
 f=\frac{T_t^2}{X_x^2}\frac{\tilde\alpha (X)X_x^4}{(U_uu_x+U_x+\tilde\beta (X)X_x)^4}
 =\frac{\alpha(x)}{(u_x+\beta(x))^4},
\]
where $\beta(x) := (U_x+\tilde\beta(X)X_x)/U_u$ and $\alpha(x):=T_t^2X_x\tilde\alpha(X)/U_u^4$. Furthermore, upon using transformations from the equivalence group~$G^\sim$, we can set $\tilde\beta=\beta=0$, which consequently implies that $U_x=0$. By means of equivalence transformations, we can also set $\beta,\tilde \beta\in\{-1,1\}$ and as the multiplier relating $\alpha$ and $\tilde\alpha$ is strictly positive, we have that $\tilde\alpha = \alpha$. As the transformation of~$X$ only depends on~$x$ it also follows from $T_t^2X_x/U_u^4=1$ that $X_x=\const$. Upon scaling this constant and translations of~$x$, which belong to~$G^\sim$, we can choose $X=x$. Therefore, $T_t^2/U_u^4=1$. As $T=1/t$ and thus $U_u=\hat q/t$, this means that $\hat q=1$, i.e., $U_u=1/t$ and hence $U=u/t+U^0(t)$ and $\tilde u_{\tilde x}=u_x/t$. Here $U^0=U^0(t)$ is a smooth function arising after integration with respect to~$u$ and depending only on~$x$ in view of the condition $U_x=0$.

The remaining part of Eq.~\eqref{eq:DeterminingEqsEquivalenceTransformations3} can be represented as
\begin{equation}\label{eq:DeterminingEqsEquivalenceTransformations3Rest}
\frac{\tilde g}{t^3} - 2U^0_t-tU^0_{tt} = g,
\end{equation}
where $U=u/t+U^0(t)$. The differentiation of Eq.~\eqref{eq:DeterminingEqsEquivalenceTransformations3Rest} with respect to~$t$ yields
$\tilde g_{\tilde u_{\tilde x}}\tilde u_{\tilde x} + 3\tilde g + t^4(tU^0_{tt}+2U^0_t)_t=0$.
The first two terms do not depend on~$t$ and the last summand depends only on~$t$. Thus, we can separate variables and set $t^4(tU^0_{tt}+2U^0_t)_t=-3\tilde\varkappa=\const$, where the factor of $-3$ was introduced for the sake of convenience. Integration of this equation yields $tU^0_{tt}+2U^0_t=\tilde\varkappa/t^3+\varkappa$, where $\varkappa=\const$. The general solution of this equation is $U^0=\hat\varkappa/(2t^2)-\varkappa t/2-\sigma_1/t+\sigma_0$, where $\sigma_1, \sigma_2=\const$. We also have $\tilde g_{\tilde u_{\tilde x}}\tilde u_{\tilde x} + 3\tilde g=\tilde \varkappa$, which upon integration leads to $\tilde g= \tilde \mu(\tilde x)/\tilde u_{\tilde x}^3+\tilde \varkappa$. Plugging these results into Eq.~\eqref{eq:DeterminingEqsEquivalenceTransformations3Rest} gives
\[
g = \frac{\tilde \mu(X)}{u_x^3}+\frac{\tilde\varkappa}{t^3}-(tU^0_{tt}+2U^0_t)=\frac{\tilde \mu(X)}{u_x^3}+\varkappa.
\]
Using equivalence transformations, we can put $\hat\varkappa=\varkappa=0$. This is why we have $\tilde f=\delta/\tilde u_{\tilde x}^4$, $f=\delta/u_x^4$, $\tilde g=\delta/\tilde u_{\tilde x}^3$ and $g=\delta/u_x^3$, where $\delta=\pm1$. That is, the equivalence transformations for this case reduce to symmetry transformations.

Owing to the above computations, we can formulate the following theorem:

\begin{theorem}\label{thm:IbragimovClassSemiNormalizedSubclass}
The subclass~$\mathcal K$ of the class~\eqref{eq:IbragimovClass}, 
that consists of equations $G^\sim$-equivalent to equations of the form~\eqref{eq:IbragimovClassSpecialLieSymExts1Subclass}, 
is semi-normalized with respect to $G^\sim$. Any admissible transformation in this subclass is generated by~$G^\sim$ or is represented as a composition of the transformations $(\theta_1,\theta_2,T_1)$, $(\theta_2,\theta_2,T_2)$ and $(\theta_2,\theta_3,T_3)$, where $\theta_1=(f,g)$, $\theta_2=(\pm u_x^{-4},\mu u_x^{-3})$, $\theta_3=(\tilde f,\tilde g)$ and $T_1$, $T_3$ are equivalence transformations and $T_2=1/t$ is a symmetry transformation of $\mathcal L_{\theta_2}$.
The complement~$\bar{\mathcal K}$ of~$\mathcal K$ in the class~\eqref{eq:IbragimovClass}
(as well as the complement of~$\mathcal K$ in the subclass of~\eqref{eq:IbragimovClass} singled out by the condition $f_{u_x}\ne0$) is normalized with respect to $G^\sim$. The usual equivalence group of the subclass~$\bar{\mathcal K}$ coincides with~$G^\sim$.
\end{theorem}

\begin{corollary}\label{col:IbragimovClassEquiv}
The entire class~\eqref{eq:IbragimovClass} is semi-normalized.
Hence the group classification of the class~\eqref{eq:IbragimovClass} up to $G^\sim$-equivalence coincides with the group classification of this class up to general point equivalence.
\end{corollary}

\begin{remark}\label{rem:OnIbragimovClassNormalizationOfGaugedSpecialSubclass}
It can be proved using the above consideration that the class~\eqref{eq:IbragimovClassSpecialLieSymExts1Subclass} is normalized. 
The equivalence group~$G^\sim_1$ of this class consists of the transformations of the general form 
\[
\tilde t=\frac{a_1t+a_0}{a_3t+a_2},\quad 
\tilde x=b_1t+b_0,\quad 
\tilde u=\frac{\pm\sqrt{|b_1A|}\,u+b_3t+b_2}{a_3t+a_2},\quad 
\tilde\mu=\frac{\mu}{b_1},
\]
were $a_i$, $i=0,\dots,3$, are constants with $A=a_1a_2-a_0a_3\ne0$ which are determined up to a common nonvanishing multiplier 
and $b_i$, $i=0,\dots,3$, are arbitrary constants with $b_1\ne0$.
The group~$G^\sim_1$ can be represented as the product of its two subgroups. 
The first subgroup is the ideal associated with the kernel group of the class~\eqref{eq:IbragimovClassSpecialLieSymExts1Subclass} and 
formed by the transformations from~$G^\sim_1$ with $b_1=1$ and $b_0=0$.
The second subgroup corresponds to the subgroup of~$G^\sim$ whose elements save equations of the form~\eqref{eq:IbragimovClassSpecialLieSymExts1Subclass} 
and consists of the transformations from~$G^\sim_1$ with $a_3=1$ and $a_2=0$.
This is why the list presented in Lemma~\ref{lem:IbragimovClassSpecialLieSymExts1} is an exhaustive list of Lie symmetry extensions 
in the class~\eqref{eq:IbragimovClassSpecialLieSymExts1Subclass} up to both $G^\sim_1$-equivalence and general point equivalence. 
\end{remark}

\begin{remark}\label{rem:OnInequivOfLinAndNonlinCasesOfIbragimovClass}
It follows from the above consideration that the entire class of equations of the general form~\eqref{eq:IbragimovClass} is partitioned into three subclasses 
associated with the additional constraints $f_{u_x}\ne0$, $f_{u_x}=0$ and $g_{u_xu_x}\ne0$, and $f_{u_x}=g_{u_xu_x}=0$, respectively. 
Equations from different subclasses of this partition are not mapped to each other by point transformations. 
This is the main reason why it is natural to separate nonlinear equations of the form~\eqref{eq:IbragimovClass} from linear ones, which are well studied and form the last subclass. 
\end{remark}

\begin{remark}
In order to simplify calculations, we could use Theorem~4.4b of Ref.~\cite{king98Ay}, 
describing form-preserving transformations between $(1+1)$-dimensional second-order partial differential equations 
of the quite general form $u_{tt}=H(t,x,u,u_x,u_{xx})$, where $H_{u_{xx}}\ne0$. 
This theorem directly implies the simplest constraints $T_u=T_x=X_u=X_t=0$ 
for admissible transformations of the class~\eqref{eq:IbragimovClass}, 
in view of which the coefficients of any Lie symmetry operator~$Q=\tau\p_t+\xi\p_x+\eta\p_u$ of each equation 
from the class~\eqref{eq:IbragimovClass} satisfy the determining equations \mbox{$\tau_u=\tau_x=\xi_u=\xi_t=0$}. 
A~partial repetition of computations in the present paper was necessary in order to finding the appropriate partition 
of the class~\eqref{eq:IbragimovClass} into subclasses. 
\end{remark}

\section{Classification of inequivalent appropriate subalgebras}\label{sec:ClassificationSubalgebrasIbragimovClass}

In order to classify subalgebras of the equivalence algebra~$\mathfrak g^\sim$, 
we need to describe the adjoint action of the equivalence group~$G^\sim$, 
which consists of transformations of the form~\eqref{eq:EquivalenceGroupIbragimovClass}, 
on the generating vector fields~\eqref{eq:EquivalenceAlgebraGenWaveEqs} of~$\mathfrak g^\sim$. 
This adjoint action can be determined by solving the Cauchy problem
\[
 \dd{\ww}{\ve} = \ad\ \vv(\ww) := [\vv,\ww],\qquad \ww(0) = \ww_0
\]
for each pair $(\vv,\ww_0)$ of generating vector fields of~$\mathfrak g^\sim$,
which is equivalent to computing the convergent Lie series~\cite{olve86Ay}
\[
  \ww(\ve) = \mathrm{Ad}(e^{\ve\vv})\ww_0 := \sum_{n=0}^\infty\frac{\ve^n}{n!}(\ad\ \vv)^n(\ww_0).
\]
An alternative way is the direct computation of actions of transformations from~$G^\sim$ on elements of~$\mathfrak g^\sim$ 
via pushforward of vector fields by these transformations~\cite{card11Ay}. 
Stated in another way, the second method uses the usual transformation rule of vector fields under point transformations. 
As this method properly works for infinite-dimensional Lie algebra, we will pursue it below.

Employing elementary equivalence transformations (cf.\ the end of Section~\ref{sec:EquivGroup}), we can compute the nonidentical adjoint actions using the respective push-forwards. This yields
\begin{align*}
 &\mathscr F^2_*(c_4) \DDD^t = \DDD^t+2c_4\FF^2,                              &&\mathscr D^t_*(c_1) \FF^2 = c_1^{-2}\FF^2,\\
 &\mathscr G_*(\psi) \DDD^u = \DDD^u- \GG(\psi),                              &&\mathscr D^u_*(c_2) \GG(\psi) = c_2\GG(\psi),\\
 &\mathscr F^2_*(c_4) \DDD^u = \DDD^u-c_4\FF^2,                               &&\mathscr D^u_*(c_2) \FF^2 = c_2\FF^2, \\
 &\mathscr G_*(\psi) \DDD(\varphi)=\DDD(\varphi) + \GG(\varphi\psi_x),        &&\mathscr D_*(\theta) \GG(\psi) = \GG(\psi(\hat\theta)),\\
 &\mathscr D_*(\theta) \DDD(\varphi) = \DDD(\varphi(\hat\theta)/\hat\theta_x),&&
\end{align*}
where $\hat\theta=\hat\theta(x)$ is the inverse of the function $\theta$.
It should be stressed that there are more nonidentical adjoint actions of transformations from~$G^\sim$ on generating vector fields of~$\mathfrak g^\sim$ than listed above, namely those related with actions on the trivial prolongation~$\hat{\mathfrak g}^\cap$ of the kernel algebra~$\mathfrak g^\cap$ to the arbitrary elements, which is an ideal in~$\mathfrak g^\sim$, and those involving $\mathscr P^t_*(c_0)$ and $\mathscr F^1_*(c_3)$. These adjoint actions, however, do not yield simplifications in the course of classification of extensions of the kernel algebra. 

We will only classify appropriate subalgebras of~$\mathfrak g^\sim$. Any appropriate subalgebra~$\mathfrak s$ of~$\mathfrak g^\sim$ should contain $\hat{\mathfrak g}^\cap=\langle\PP^t,\FF^1,\GG(1)\rangle$. For the class~\eqref{eq:IbragimovClass} we have two specific representations of~$\mathfrak s$, which are given by~$\mathfrak s = \hat {\mathfrak g}^\cap+\langle Q^1,\dots,Q^k\rangle=\langle\PP^t,\FF^1\rangle+\langle \GG(1),Q^1,\dots,Q^k\rangle$, 
where ``$+$'' denotes the direct sum of vector spaces, 
$\hat {\mathfrak g}^\cap$ is an ideal of~$\mathfrak s$ (since it is an ideal of the entire~$\mathfrak g^\sim$) and 
$\langle \GG(1),Q^1,\dots,Q^k\rangle$ is a subalgebra of~$\mathfrak s$. 
$Q^1$, \dots, $Q^k$ are basis elements from the complement of $\hat{\mathfrak g}^\cap$ in~$\mathfrak s$ and their projections to the space of equation variables yield a proper Lie symmetry extension of~$\mathfrak g^\cap$ in the class~\eqref{eq:IbragimovClass}.

\begin{remark}\label{rem:OnStructureOfequivAlgebra}
The double representation of appropriate subalgebras is related with the representation of the whole equivalence algebra~$\mathfrak g^\sim$ in the form $\mathfrak g^\sim=\hat{\mathfrak g}^\cap+\bar{\mathfrak g}$, where 
$\hat{\mathfrak g}^\cap$ and $\bar{\mathfrak g}=\langle\DDD^u,\DDD^t,\DDD(\varphi),\GG(\psi),\FF^2\rangle$ 
are an ideal and a subalgebra of~$\mathfrak g^\sim$ but the sum is not direct even in the sense of vector spaces 
since $\hat{\mathfrak g}^\cap\cap\bar{\mathfrak g}=\langle\GG(1)\rangle$. 
Unfortunately, the equivalence algebra~$\mathfrak g^\sim$ does not possess a representation as a semi-direct sum of the ideal~$\hat{\mathfrak g}^\cap$ associated with the kernel algebra and a certain subalgebra, 
which additionally complicates the group classification of the class~\eqref{eq:IbragimovClass}.
\end{remark}

This is why it is necessary to classify only subalgebras of~$\mathfrak g^\sim$ which are contained in~$\bar{\mathfrak g}$ and contain $\langle\GG(1)\rangle$.
The classification should be carried out up to $G^\sim_0$-equivalence, where 
$G^\sim_0$ is a subgroup of~$G^\sim$ formed by the transformations~\eqref{eq:EquivalenceGroupIbragimovClass} 
with $c_0=c_3=0$. 
In fact, we will present the classification results in terms of extensions of~$\hat{\mathfrak g}^\cap$ 
excluding $\GG(1)$ from the corresponding bases. 

The determining equations for Lie symmetries of equations from the class~\eqref{eq:IbragimovClass} impose more restrictions on appropriate subalgebras.

\begin{lemma}\label{lem:OnAppropriateSubalgebras1}
$\mathfrak s\cap\langle\DDD^u,\GG(\psi),\FF^2\rangle=\mathfrak s\cap\langle\DDD^t,\FF^2\rangle=\{0\}$
for any appropriate subalgebra~$\mathfrak s$.
\end{lemma}

\begin{proof}
Suppose that an appropriate subalgebra $\mathfrak s$ of~$\mathfrak g^\sim$ contains an operator~$Q=b\DDD^u+\GG(\psi)+c\FF^2$, where at least one of the constants~$b$ and~$c$ or the derivative~$\psi_x$ of the function~$\psi=\psi(x)$ does not vanish. Then the operator~$\mathrm P Q$ is a Lie symmetry operator for an equation from the class~\eqref{eq:IbragimovClass}. Substituting the coefficients of operator~$Q$ into the determining equations~\eqref{eq:DeterminingEqsSimplifiedIbragimovClass} implies the following conditions for the arbitrary elements~$f$ and~$g$:
\[
(bu_x+\psi_x)f_{u_x}=0,\quad (bu_x+\psi_x)g_{u_x} = 2c-\psi_{xx}f+bg.
\]
For both the cases $b\ne0$ and $\psi_x\ne0$ it follows that $f_{u_x}=0$ and $g_{u_xu_x}=0$, which contradicts the definition of class~\eqref{eq:IbragimovClass}. The case $b=0$, $\psi_x=0$ and $c\ne0$ leads to a contradiction.
Therefore, any appropriate subalgebra does not contain an operator of the form considered.

Analogously, an operator~$\DDD^t+c\FF^2$, where~$c$ is an arbitrary constant, gives the condition $f=0$, which is also inconsistent with the definition of the class~\eqref{eq:IbragimovClass}.
\end{proof}

\begin{lemma}\label{lem:OnAppropriateSubalgebras2}
$\dim\big(\mathfrak s\cap\langle\DDD(\varphi),\GG(\psi),\FF^2\rangle\big)\leqslant2$ for any appropriate subalgebra~$\mathfrak s$.
\end{lemma}

\begin{proof}

Suppose that~$\mathfrak s$ is an appropriate subalgebra of $G^\sim$
and $\dim\big(\mathfrak s\cap\langle\DDD(\varphi),\GG(\psi),\FF^2\rangle\big)\geqslant2$.
This means that the subalgebra~$\mathfrak s$ contains (at least) two operators~$Q^i=\DDD(\varphi^i)+\GG(\psi^i)+c_i\FF^2$, where the functions~$\varphi^i$, $i=1,2$, should be linearly independent in view of Lemma~\ref{lem:OnAppropriateSubalgebras1}. In other words, the projections $\mathrm P Q^i$ of $Q^i$ simultaneously are Lie symmetry operators of an equation from the class~\eqref{eq:IbragimovClass}. By~$W$ we denote the Wronskian of the functions $\varphi^1$ and $\varphi^2$, $W=\varphi^1\varphi^2_x-\varphi^2\varphi^1_x$. $W\ne0$ as the functions~$\varphi^1$ and $\varphi^2$ are linearly independent.

Plugging the coefficients of~$\mathrm P Q^i$ into the first classifying equation from the system~\eqref{eq:DeterminingEqsSimplifiedIbragimovClass} gives two equations with respect to~$f$ only,
\begin{equation}\label{eq:DetermingEquationsSimplifiedTwice1}
 (\varphi^i_x u_x-\psi^i_x)f_{u_x}-\varphi^if_x + 2\varphi^i_xf = 0.
\end{equation}
We multiply the equation corresponding to $i=2$ by $\varphi^1$ and subtract it from the equation for $i=1$ multiplied by $\varphi^2$. Dividing the resulting equation by~$W$, we obtain the ordinary differential equation
\[
(u_x+\beta)f_{u_x}+2f = 0,
\]
where $\beta=\beta(x):=(\varphi^2\psi^1_x-\varphi^1\psi^2_x)/W$ and the variable~$x$ plays the role of a parameter. It is possible to set $\beta=0$ by means of an equivalence transformation, $\mathscr G(-\beta)$. Indeed, this transformation preserves the form of the operators~$Q^i$, only changing the values of the functional parameters $\psi^i$. In particular, it does not affect the linear independency of the functions $\varphi^i$. The integration of the above equation for $\beta=0$ yields that $f=\alpha u_x^{-2}$, where $\alpha=\alpha(x)$ is a nonvanishing function of~$x$. In view of the derived form of~$f$, splitting of equations~\eqref{eq:DetermingEquationsSimplifiedTwice1} with respect to~$u_x$ leads to $\varphi^1\alpha_x=0$ and $\varphi^i_x=0$, i.e.\ $\psi^i_x=0$ and $\alpha_x=0$. As $\GG(1)\in\mathfrak s$, we can assume up to linear combining of elements of $\mathfrak s$ that $\psi^i=0$. The constant $\alpha$ can be scaled to $\alpha=\pm1$ by an equivalence transformation.

In a similar manner, consider the last equation from system~\eqref{eq:DeterminingEqsSimplifiedIbragimovClass}, taking into account the restrictions set on parameter-functions and the form of~$f$. For each~$Q^i$, this classifying equation gives an equation with respect to~$g$,
\begin{equation}\label{eq:DetermingEquationsSimplifiedTwice2}
 \varphi^i_xu_xg_{u_x} - \varphi^i g_x= -\varphi^i_{xx}\alpha u_x^{-1} - 2c_i.
\end{equation}
Again, we multiply the equation corresponding to $i=2$ with $\varphi^1$ and subtract it from the equation for $i=1$ multiplied by $\varphi^2$, divide the resulting equation by~$W$ and thereby obtain that
$
 g_{u_x} = \mu^2 u_x^{-2} + \mu^1 u_x^{-1},
$
where $\mu^2=\mu^2(x):=\alpha W_x/W$ and $\mu^1=\mu^1(x):=2(c_1\varphi^2-c_2\varphi^1)/W$. Integration with respect to~$u_x$ directly gives $g=\mu^2 u_x^{-1} + \mu^1\ln|u_x|+\mu^0$, where $\mu^0=\mu^0(x)$ is a smooth function of~$x$. The parameter-function~$\mu^2$ can be set equal to zero by the equivalence transformation~$\mathscr D(\zeta)$, where the function $\zeta=\zeta(x)$ is a solution of the equation $\alpha\zeta_{xx} +\mu^2\zeta_x=0$. Substituting the derived form of~$g$ into equations~\eqref{eq:DetermingEquationsSimplifiedTwice2} and splitting with respect to~$u_x$, we find that~$\mu^1_x=0$, $\varphi^i_{xx}=0$, $\varphi^i\mu^0_x=\varphi^i_x\mu^1+2c_i$. Therefore, $\mu^1$ is a constant and the functions~$\varphi^1$ and~$\varphi^2$ can be set to $1$ and $x$, respectively, upon linear combining of~$Q^i$. Then, we have $\mu^0_x=2c_1$, $x\mu^0_x=2c_2+\mu^1$, i.e.\ $c_1=0$, $c_2=-1/2$ and $\mu^0$ is a constant that can be set to zero by the equivalence transformation~$\mathscr F(-\mu^0/2)$.

Summing up, we have proved that any equation of class~\eqref{eq:IbragimovClass} admitting (at least) two operators~$\mathrm P Q^i$ is $G^\sim$-equivalent to an equation of the form
\[
 u_{tt} = \pm u_x^{-2}u_{xx} + \mu^1\ln|u_x|,
\]
where $\mu^1=\const$. However, the determining equations~\eqref{eq:DeterminingEqsSimplifiedIbragimovClass} in this case yield~$\eta_x=0$, $\eta_u=\tau_t$, $\xi_{xx}=0$, $\mu^1\eta_u=0$, $\eta_{tt}=\mu^1(\tau_t-\xi_x)$. This obviously implies that the number of such operators~$Q^i$ cannot exceed two.
\end{proof}

\begin{corollary}\label{cor:OnAppropriateSubalgebras1}
There are two $G^\sim$-inequivalent cases of Lie symmetry extensions in class~\eqref{eq:IbragimovClass} involving two linearly independent operators of the form~$\mathrm P Q^i$, where $Q^i=\DDD(\varphi^i)+\GG(\psi^i)+c_i\FF^2$,
\begin{gather*}\hspace*{-\arraycolsep}
\begin{array}{lll}
1.& u_{tt} = \pm u_x^{-2}u_{xx} + 2\ln|u_x|\colon & \mathfrak g^{\max} = \mathfrak g^\cap + \langle\mathrm P \DDD(1),\mathrm P \DDD(x)-\mathrm P \FF^2\rangle, \\
2.& u_{tt} = \pm u_x^{-2}u_{xx}\colon & \mathfrak g^{\max} = \mathfrak g^\cap + \langle\mathrm P \DDD(1),\mathrm P \DDD(x),\mathrm P \DDD^t+\mathrm P \DDD^u\rangle.
\end{array}
\end{gather*}
\end{corollary}

\begin{proof}
For $\mu^1\ne0$, we have that $\eta_u=\tau_t=0$, $\xi_{xx}=0$ and, after scaling of $\mu^1$ to two by an equivalence transformation, $\eta_{tt}=-2\xi_x$. This directly gives the first case. If $\mu^1=0$, we obviously recover the second case.
\end{proof}

Now that we have computed the essential adjoint actions and classified all appropriate subalgebras in Corollary~\ref{cor:OnAppropriateSubalgebras1} for which $\dim\big(\mathfrak s\cap\langle\DDD(\varphi),\GG(\psi),\FF^2\rangle\big)=2$, we should go on with the computation of inequivalent appropriate extensions of $\hat{\mathfrak g}^\cap$, 
which contain at most one linearly independent operator of the form $\DDD(\varphi)+\GG(\psi)+c\FF^2$, 
where $\varphi=\varphi(x)$ is a nonvanishing function. 
In view of Lemma~\ref{lem:OnAppropriateSubalgebras1} it is obvious that the dimension of such extensions cannot be greater than three. 
Here we select candidates for such extensions using only restrictions on appropriate subalgebras presented in Lemmas~\ref{lem:OnAppropriateSubalgebras1} and~\ref{lem:OnAppropriateSubalgebras2}. 
As there exist specific restrictions for two- and three-dimensional extensions, 
we will make an additional selection of appropriate extensions from the set of candidates directly in the course of the construction of invariant equations.

The result of the classification is formulated in the subsequent lemmas.

\begin{lemma}\label{lem:1DimInequivExtsForIbragimovClass}
A complete list of $G^\sim$-inequivalent appropriate one-dimensional extensions of~$\hat{\mathfrak g}^\cap$ in~$\mathfrak g^\sim$ is given by
\begin{align}\label{eq:OneDimensionalSubalgebrasIbragimovClass}
\begin{split}
 &\langle\DDD^u+\tfrac12\DDD^t+\DDD(\ve)+\FF^2\rangle, \quad \langle\DDD^u-p\DDD^t+\DDD(\ve)\rangle, \quad \langle\DDD^t-\DDD(1)\rangle, \\
 &\langle\DDD^t-\GG(x)\rangle, \quad \langle\DDD(1)+\ve\FF^2\rangle,
\end{split}
\end{align}
where $\ve\in\{0,1\}$ and $a$ is an arbitrary constant.
\end{lemma}

\begin{proof}
The classification of the appropriate one-dimensional extensions can be carried out effectively by simplifying a general element of the linear span $\langle\DDD^u,\DDD^t,\DDD(\varphi),\GG(\psi),\FF^2\rangle$,
\[
 Q = a_1\DDD^u+a_2\DDD^t+\DDD(\varphi)+\GG(\psi)+a_4\FF^2,
\]
using push-forwards of transformations from~$G^\sim$. For this aim, it is necessary to distinguish multiple cases, subject to which of the constants $a_i$ or functions $\varphi$, $\psi$ are nonzero. We note in the beginning that owing to the push-forward $\mathscr D_*(\theta)$ we can always set $\varphi=a_3=\const$.

For $a_1\ne0$ we can scale the vector field $Q$ to achieve $a_1=1$. Using the push-forwards of a suitable transformation $\mathscr G(\chi)$, we can set $\psi=0$. The further possibilities for simplification depend crucially on the value of $a_2$. For $a_2=1/2$, the sum $\DDD^u+a_2\DDD^t$ is invariant under the push-forward~$\mathscr F^2_*(c_4)$ and therefore it is not possible to set $a_4=0$. The actions of the push-forwards of the transformations $\mathscr D(x)$ and $\mathscr D^u$ allow scaling of $a_3$ and $a_4$ to $\{0,1\}$.  If $a_4=1$, by denoting $a_3=\ve$ we obtain the first case from the list~\eqref{eq:OneDimensionalSubalgebrasIbragimovClass}.

For $a_2\ne1/2$ we can use the push-forward~$\mathscr F^2_*(c_4)$ to additionally set $a_4=0$, which gives, jointly with the case $a_2=1/2$ and $a_4=0$, the second extension listed, where $a_2$ is denoted by~$-p$.

If $a_1=0$ and $a_2\ne0$, we scale $a_2=1$ and can use the push-forward~$\mathscr F^2_*(c_4)$ to set $a_4=0$. For $a_3\ne0$, we can scale $a_3=-1$ by means of the action of $\mathscr D^u_*(c_2)$ and additionally put~$\psi=0$ upon using the push-forward of the transformation~$\mathscr G(\chi)$. If $a_3=0$, we have $\psi_x\ne0$ in view of Lemma~\ref{lem:OnAppropriateSubalgebras1} and hence we can use the action of $\mathscr D_*(\theta)$ to set $\psi=-x$. This gives the third and the fourth case of the list~\eqref{eq:OneDimensionalSubalgebrasIbragimovClass}, respectively.

In case of $a_1=a_2=0$ but $a_3\ne0$, we can set $a_3=1$ and use the push-forward~$\mathscr G_*(\psi)$ for a certain $\psi$ to arrive at $\psi=0$. The action of $\mathscr D^u_*(c_2)$ on the resulting vector field allows us to scale the coefficient $a_4$ so that we have $a_4\in\{0,1\}$, which yields the fifth element of the above list of one-dimensional inequivalent subalgebras.

In view of Lemma~\ref{lem:OnAppropriateSubalgebras1}, the case $a_1=a_2=a_3=0$ is not appropriate.
\end{proof}

\begin{lemma}\label{lem:2DimInequivExtsForIbragimovClass}
Up to $G^\sim$-equivalence, any appropriate two-dimensional extension of $\hat{\mathfrak g}^\cap$ in~$\mathfrak g^\sim$, 
which contains at most one linearly independent operator of the form $\DDD(\varphi)+\GG(\psi)+c\FF^2$, belongs to the following list: 
\begin{align}\label{eq:TwoDimensionalSubalgebrasIbragimovClass}
\begin{split}
    &\langle\DDD^u+\DDD(1),\,\DDD^t+\DDD(b)\rangle,\quad \langle\DDD^u+\DDD(1),\,\DDD^t+\GG(e^x)\rangle, \\
    &\langle a_1\DDD^u+a_2\DDD^t+a_3\DDD(x)+\ve_0\GG(x)+\ve_1\FF^2,\,\DDD(1)+\ve_2\FF^2\rangle,
\end{split}
\end{align}
where $b$, $a_1$, $a_2$, $a_3$, $\ve_0$, $\ve_1$ and $\ve_2$ are constants with 
$b\ne0$, $(a_1,a_2)\ne(0,0)$, $(a_2,a_3)\ne(0,0)$, $(a_1,a_3,\ve_0)\ne(0,0,0)$ and $(a_1-2a_2-a_3)\ve_2=0$.
Due to scaling of the first basis element and $G^\sim$-equivalence we can also assume that 
one of $a$'s equals 1, $(a_1-a_3)\ve_0=0$, $(a_1-2a_2)\ve_1=0$, $\ve_0,\ve_1\in\{0,1\}$ 
and $\ve_2\in\{-1,0,1\}$ or, if $\ve_0=0$, $\ve_2\in\{0,1\}$.
\end{lemma}

\begin{proof}
The general strategy is to take two arbitrary linearly independent elements~$Q^1$ and~$Q^2$
from the linear span $\langle\DDD^u,\DDD^t,\DDD(\varphi),\GG(\psi),\FF^2\rangle$ such that 
$\mathfrak s=\hat{\mathfrak g}^\cap+\langle Q^1,Q^2\rangle$ is a five-dimensional subalgebra of~$\mathfrak g^\sim$ 
satisfying the restriction on elements of the form $\DDD(\varphi)+\GG(\psi)+c\FF^2$ and Lemma~\ref{lem:OnAppropriateSubalgebras1} and simplify~$Q^1$ and~$Q^2$ as much as possible by linear combining of elements of~$\mathfrak s$ and 
push-forwards of transformations from~$G^\sim$. 
The proof is split into two parts. 

First, we consider possible extensions not involving operators of the form $\DDD(\varphi)+\GG(\psi)+c\FF^2$. 
In view of this additional restriction and Lemma~\ref{lem:OnAppropriateSubalgebras1}, 
up to linear combining we can take the elements~$Q^1$ and~$Q^2$ in the initial form
\[
Q^1=\DDD^u+\DDD(\varphi^1)+\GG(\psi^1)+c_1\FF^2, \quad Q^2=\DDD^t+\DDD(\varphi^2)+\GG(\psi^2)+c_2\FF^2,
\]
where $\varphi^1\ne0$. 
We set $\varphi^1=1$, $\psi^1=0$ and $c_1=0$ using $\mathscr D_*(\theta)$, $\mathscr G_*(\chi)$ with suitable functions $\theta$ and $\chi$ and $\mathscr F^2_*(c_1)$, respectively, i.e.\ $Q^1=\DDD^u+\DDD(1)$. 
As the subalgebra~$\mathfrak s$ is closed with respect to the Lie bracket of vector fields, we have 
$[Q^1,Q^2]=\DDD(\varphi^2_x)+\GG(\psi^2_x-\psi^2)-c_2\FF^2\in\langle\GG(1)\rangle$ and hence 
$\varphi^2_x=0$, $c_2=0$ and $\psi^2_x-\psi^2=\const$. 
Integrating the equations for~$\varphi^2$ and~$\psi^2$ gives that $\varphi^2=b$ and $\psi^2=d_1e^x+d_0$ 
for some constants~$b$, $d_1$ and~$d_0$. 
The constant~$d_0$ can be always set equal to zero by linear combining with the operator~$\GG(1)$ belonging to~$\hat{\mathfrak g}^\cap$. The further simplification of~$Q^2$ depends on the value of~$b$. 
If $b\ne0$, the push-forward of $\mathscr G(-d_1b^{-1}e^x)$ does not change~$Q^1$ and leads to~$d_1=0$. 
If $b=0$, the parameter~$d_1$ is nonzero in view of Lemma~\ref{lem:OnAppropriateSubalgebras1} and, therefore, 
can be scaled to~1 by $\mathscr D_*(d_1^{-1})$. 
As a result, we obtain the first two elements of the list~\eqref{eq:TwoDimensionalSubalgebrasIbragimovClass}.

Now we investigate the case $\dim\big(\mathfrak s\cap\langle\DDD(\varphi),\GG(\psi),\FF^2\rangle\big)=1$. 
Then basis operators of the extension of~$\hat{\mathfrak g}^\cap$ can be chosen in the form 
\[
Q^1=a_1\DDD^u+a_2\DDD^t+\DDD(\varphi^1)+\GG(\psi^1)+c_1\FF^2,\quad 
Q^2=\DDD(\varphi^2)+\GG(\psi^2)+c_2\FF^2,
\]
where $(a_1,a_2)\ne(0,0)$ and $\varphi^2\ne0$.
We set $\varphi^2=1$ and $\psi^2=0$ using $\mathscr D_*(\theta)$ and $\mathscr G_*(\chi)$ 
with suitably chosen functions $\theta$ and $\chi$, respectively. 
As $\mathfrak s$ is a Lie algebra, we have that $[Q^2,Q^1]=\DDD(\varphi^1_x)+\GG(\psi^1_x)+(a_1-2a_2)c_2\FF^2=a_3Q_2+\GG(d)$ 
for some constants~$a_3$ and~$d$. 
Therefore, $\varphi^1_x=a_3$, $(a_1-2a_2-a_3)c_2=0$ and $\psi^1_x=c_0$. 
Up to combing $Q^1$ with $Q^2$ and $\GG(1)$ we obtain that $\varphi^1=a_3x$ and $\psi^1=c_0x$. 
Up to $G^\sim$-equivalence we can assume that $(a_1-a_3)c_0=0$ and  $(a_1-2a_2)c_1=0$. 
Indeed, if $a_1-2a_2\ne0$, we can set $c_1=0$ using $\mathscr F^2_*(\tilde c_1)$ with an appropriate constant~$\tilde c_1$. 
To set $c_0=0$ in the case $a_1-a_3\ne0$, we simultaneously act by $\mathscr G_*(\tilde c_0x)$ 
with an appropriate constant~$\tilde c_0$ and linearly combine the operator~$Q^2$ with~$\DDD(1)$. 
Using push-forwards of scalings of variables and alternating their signs, 
we can independently scale the constant parameters $c_0$, $c_1$ and~$c_2$ 
and change sings of $c_1$ and, simultaneously, $c_0$ and~$c_2$. 
Additionally we can multiply the whole vector field~$Q^1$ by a nonvanishing constant to scale one of nonvanishing $a$'s to one. 
The conditions $(a_2,a_3)\ne(0,0)$ and $(a_1,a_3,\ve_0)\ne(0,0,0)$ follow from Lemma~\ref{lem:OnAppropriateSubalgebras1}. 
This yields the third case of the list~\eqref{eq:TwoDimensionalSubalgebrasIbragimovClass} 
and thereby completes the proof of the theorem.
\end{proof}

\begin{lemma}\label{lem:3DimInequivExtsForIbragimovClass}
Up to $G^\sim$-equivalence, any appropriate three-dimensional extension of $\hat{\mathfrak g}^\cap$ in~$\mathfrak g^\sim$, 
which contains at most one linearly independent operator of the form $\DDD(\varphi)+\GG(\psi)+c\FF^2$, has one of the following forms: 
\begin{equation}\label{eq:3Dextensions}
\begin{split}
&\langle\DDD^u+p_1\DDD(x),\,\DDD^t+p_2\DDD(x),\,\DDD(1)+\ve\FF^2\rangle,\\
&\langle\DDD^u+\DDD(x)+d\GG(x),\,\DDD^t-\GG(x),\,\DDD(1)\rangle,
\end{split}
\end{equation}
where $p_1$, $p_2$ and $d$ are constants, $\ve\in\{0,1\}$, $(p_1,p_2)\ne(1,0)$ and $\ve(p_1-1)=\ve(p_2+2)=0$.
\end{lemma}

\begin{proof}
In view of Lemma~\ref{lem:OnAppropriateSubalgebras1},
any three-dimensional appropriate extension of~$\hat{\mathfrak g}^\cap$, which contains at most one linearly independent operator of the form $\DDD(\varphi)+\GG(\psi)+c\FF^2$, is spanned by vector fields 
$Q^1=\DDD^u+\DDD(\varphi^1)+\GG(\psi^1)+c_1\FF^2$, 
$Q^2=\DDD^t+\DDD(\varphi^2)+\GG(\psi^2)+c_2\FF^2$ and 
$Q^3=       \DDD(\varphi^3)+\GG(\psi^3)+c_3\FF^2$, 
where~$\varphi^i$ and~$\psi^i$ are smooth functions of~$x$, $c_i$ are constants, $\varphi^1,\varphi^3\ne0$ 
and $(\varphi^2,\psi^2)\ne(0,0)$.

Using $\mathscr F_*(c_1)$, $\mathscr D_*(\theta)$ and $\mathscr G_*(\chi)$ 
with suitably chosen functions $\theta$ and $\chi$ of~$x$ and, if $c_3\ne0$, $\mathscr D^u_*(c_3^{-1})$,  
we set  $c_1=0$, $\varphi^3=1$, $\psi^3=0$ and $c_3=\ve\in\{0,1\}$.
The commutation relations of~$Q^3$ with~$Q^1$ and~$Q^2$ are 
\begin{gather*}
[Q^3,Q^1]=\DDD(\varphi^1_x)+\GG(\psi^1_x)+c_3\FF^2=p_1Q^3+d_1\GG(1),\\
[Q^3,Q^2]=\DDD(\varphi^2_x)+\GG(\psi^2_x)-2c_3\FF^2=p_2Q^3+d_2\GG(1)
\end{gather*}
for some constants $p_i$ and~$d_i$, $i=1,2$. 
These commutation relations imply the conditions $\varphi^i_x=p_i$, $\psi^i_x=d_i$ and $(p_1-1)\ve=(p_2+2)\ve=0$. 
Therefore, up to combining $Q^i$ with~$Q^3$ and~$\GG(1)$ we obtain $\varphi^i=p_ix$ and $\psi^i=d_ix$. 
Then the commutation relation 
\[
[Q^2,Q^1]=(d_2+p_2d_1-p_1d_2)\GG(x)+c_2\FF^2=0
\]
yields $c_2=0$ and $p_2d_1=(p_1-1)d_2$. 
If $p_1\ne1$, we can set $d_1=0$ using $\mathscr G_*(-(p_1-1)^{-1}d_1x)$ 
and then the equality $p_2d_1=(p_1-1)d_2$ is reduced to $d_2=0$. 
Analogously, in the case $p_2\ne0$ we can set $d_2=0$ using $\mathscr G_*(-p_2^{-1}d_2x)$ 
and then the equality $p_2d_1=(p_1-1)d_2$ is equivalent to $d_1=0$. 
Summing up, we always have $d_1=d_2=0\bmod G^\sim$ if $(p_1,p_2)\ne(1,0)$. 
This gives the first extension in~\eqref{eq:3Dextensions}. 
Otherwise, $p_1=1$, $p_2=0$ and hence $\ve=0$ and $d_2\ne0$. 
Setting $d_2=-1$ by $\mathscr D^u_*(-d_2^{-1})$, we obtain the second extension in~\eqref{eq:3Dextensions}. 
This completes the proof of the theorem.
\end{proof}

\section{Result of group classification}
\label{sec:GroupClassificationIbragimovClass}

To describe equations whose Lie invariance algebras contain the projection~$\mathrm P\mathfrak s$ of a certain appropriate subalgebra~$\mathfrak s$ of~$\mathfrak g^\sim$ to the variable space, 
we can use two equivalent ways, which lead to the same system of partial differential equations in the arbitrary elements~$f$ and~$g$: 
For each basis element~$\mathcal Q$ of~$\mathfrak s$ we should 
either substitute the coefficients of~$\mathrm P\mathcal Q$ into the last two equations of system~\eqref{eq:DeterminingEqsSimplifiedIbragimovClass} 
or write the condition of invariance of the functions~$f$ and~$g$ with respect to~$\mathcal Q$. 
Then we should solve the joint system of the equations derived. 
Simultaneously we should check whether the projection~$\mathrm P\mathfrak s$ is really the maximal Lie invariance algebra 
for obtained values of the arbitrary elements~$f$ and~$g$.

All the candidates for one-dimensional appropriate extensions listed in Lemma~\ref{lem:1DimInequivExtsForIbragimovClass} are really appropriate. 
For each representative of the list we have an uncoupled system of two equations in~$f$ and~$g$, which is easy to be solved. 
As a result, we obtain the following list of equations from class~\eqref{eq:IbragimovClass} that admit 
one-dimensional Lie symmetry extensions of~$\mathfrak g^\cap$ related to $\mathfrak g^\sim$:
\begin{gather*}\hspace*{-\arraycolsep}
\begin{array}{lll}
1.1.& \DDD^u+\tfrac12\DDD^t+\DDD(\ve)+\FF^2\colon &u_{tt} = \tilde f(\omega)u_x^{-1}u_{xx}+\tilde g(\omega)+2\ln|u_x|   ,\\[.5ex]
1.2.& \DDD^u-p\DDD^t+\DDD(\ve)             \colon &u_{tt} = |u_x|^{2p}(\tilde f(\omega)u_{xx}+\tilde g(\omega)u_x)     ,\\[.5ex]
1.3.& \DDD^t-\DDD(1)                       \colon &u_{tt} = e^{2x}(\tilde f(u_x)u_{xx}+\tilde g(u_x))                  ,\\[.5ex]
1.4.& \DDD^t-\GG(x)                        \colon &u_{tt} = e^{2u_x}(\tilde f(x)u_{xx}+\tilde g(x))                    ,\\[.5ex]
1.5.& \DDD(1)+\ve\FF^2                     \colon &u_{tt} = \tilde f(u_x)u_{xx}+\tilde g(u_x)+2\ve x)                   ,
\end{array}\hspace*{-5ex}
\end{gather*}
where $\omega=x-\ve\ln|u_x|$, $\ve\in\{0,1\}$ and $p$ is an arbitrary constant. 
Here and in what follows, in each case we present only basis elements the corresponding subalgebra of~$\mathfrak g^\sim$ 
which belong to the complement of the basis $\{\PP^t,\GG(1),\FF^1\}$ of~$\hat{\mathfrak g}^\cap$.  

Calculations related to two-dimensional extensions are more complicated. 
We first present the result of the calculations and then give some explanations. 
\begin{gather*}\hspace*{-\arraycolsep}
\begin{array}{rl}
2.1. & \DDD^u-\DDD(p),\ \DDD^t-\DDD(1),\ p\ne0,                                              \colon \quad u_{tt}=\pm e^{-2x}|u_x|^{2p}(u_{xx}+\nu u_x)                        ,\\[.5ex] 
2.2. & \DDD^u+\DDD(x),\ \DDD^t-\GG(x)                                                        \colon \quad u_{tt}=\pm e^{2u_x}(x^2u_{xx}+\nu x)                               ,\\[.5ex] 
2.3. & \DDD^u+\DDD^t+\DDD(x),\ \DDD(1)                                                       \colon \quad u_{tt}=\tilde f(u_x)u_{xx}                                          ,\\[.5ex] 
2.4. & \DDD^u+\DDD^t+\DDD(x)+\GG(x),\ \DDD(1)                                                \colon \quad u_{tt}=\pm u_{xx}+e^{-u_x}                                          ,\\[.5ex] 
2.5. & 2\DDD^u+\DDD^t+2\DDD(x)+\GG(x)+\FF^2,\ \DDD(1)                                        \colon \quad u_{tt}=\pm e^{2u_x}u_{xx}+2u_x                                      ,\\[.5ex] 
2.6. & \DDD^u+\DDD(x)+\GG(x),\ \DDD(1)+\ve_2\FF^2,\ \ve_2\in\{-1,0,1\}                       \colon \quad u_{tt}=\pm e^{2u_x}u_{xx}+e^{u_x}+2\ve_2x                           ,\\[.5ex] 
2.7. & (2-q)\DDD^u+(1-q)\DDD^t+(2-q)\DDD(x)+\GG(x),\ \DDD(1),\ q\ne0,1                       \colon \\&   u_{tt}=\pm e^{2u_x}u_{xx}+e^{qu_x}                                  ,\\[.5ex] 
2.8. & (2+2p-q)\DDD^u+(1+p-q)\DDD^t+(1+2p-q)\DDD(x),\ \DDD(1),\ q\ne0                        \colon \\&   u_{tt}=\pm |u_x|^{2p}u_{xx}+|u_x|^q                                 ,\\[.5ex] 
2.9. & (3+2p)\DDD^u+\DDD^t+(1+2p)\DDD(x),\ \DDD(1)+\FF^2                                     \colon \\&   u_{tt}=\pm |u_x|^{2p}u_{xx}+\ve_3|u_x|^{p+1/2}+2x,\ \ve_3\in\{0,1\} ,\\[.5ex] 
2.10.& 2(1+p)\DDD^u+(1+p)\DDD^t+(1+2p)\DDD(x)+\FF^2,\ \DDD(1)                                \colon \quad u_{tt}=\pm |u_x|^{2p}u_{xx}+2\ln|u_x|                               ,\\[.5ex] 
2.11.& 2\DDD^u+\DDD^t+2\FF^2,\ \DDD(1)+\FF^2                                                 \colon \quad u_{tt}=\pm u_x^{-1}u_{xx}+2\ln|u_x|+2x                              .         
\noprint{
\\
2.8.& (2+2p-q)\DDD^u+(1+p-q)\DDD^t+(1+2p-q)\DDD(x),\ \DDD(1)+\ve_2\FF^2,\\                                                                   
    & (1+2p-2q)\ve_2=0,\ \ve_2,\ve_3\in\{0,1\},\ (p,\ve_3q)\ne(0,0),\                                                                              
      (\ve_2,\ve_3)\ne(0,0),\ q\ne0                                             \colon\\ 
    &                                                                                        u_{tt}=\pm |u_x|^{2p}u_{xx}+\ve_3|u_x|^q+\ve_2x    ,\\ 
2.9.& 2(1+p)\DDD^u+(1+p)\DDD^t+(1+2p)\DDD(x)+\ve_1\FF^2,\ \DDD(1)+\ve_2\FF^2,\\     
    & (1+2p)\ve_2=0,\ \ve_1,\ve_2\in\{0,1\},\ (p,\ve_1)\ne(0,0),   
      (\ve_1,\ve_2)\ne(0,0)                                                     \colon\\ 
    &                                                                                        u_{tt}=\pm |u_x|^{2p}u_{xx}+2\ve_1\ln|u_x|+\ve_2x  .\\
}
\end{array}\hspace*{-5ex}                                                                     
\end{gather*}
Nontrivial constraints for constant parameters which are imposed by the maximality condition for the corresponding extensions 
are discussed in detail after Theorem~\ref{thm:IbragimovClassGroupClassification}.

Cases~2.1 and 2.2 correspond to the first and second spans from Lemma~\ref{lem:2DimInequivExtsForIbragimovClass}, respectively. 
For the associated invariant equations to have a simpler form, 
these span are replaced by the equivalent spans $\langle\DDD^u-\DDD(p),\,\DDD^t-\DDD(1)\rangle$, where $p=-b^{-1}$, 
and $\langle\DDD^u+\DDD(x),\,\DDD^t-\GG(x)\rangle$, respectively. 
Note that we always can set a constant multiplier of the arbitrary element~$f$ to~$\pm1$, e.g., by scaling of~$t$.

The third span from Lemma~\ref{lem:2DimInequivExtsForIbragimovClass} in fact represents a multiparametric series of candidates for appropriate extensions, 
which is partitioned in the course of the construction of invariant equations into Cases~2.3--2.11. 
Not all values of series parameters give appropriate extensions. 
Additional constraints for parameters follow from the consistence conditions of the associated system in the arbitrary elements,     
\begin{gather*}
\hspace*{-\arraycolsep}\begin{array}{ll}
f_x=0      ,\quad & ((a_1-a_3)u_x+\ve_0)f_{u_x}=2(a_3-a_2)f,\\[.5ex]
g_x=2\ve_2 ,\quad & ((a_1-a_3)u_x+\ve_0)g_{u_x}=(a_1-2a_2)g-2\ve_2a_3x+2\ve_1,
\end{array}                                                                  
\end{gather*}
with the inequality $f\ne0$ and the requirement that the dimension of extensions should not exceed two. 

The above partition is carried out in the following way. 

If $a_1=a_3=a_2$, the common value of~$a$'s is nonzero and we can set it to be equal~1 by scaling the first basis elements of the span. 
We also have that $\ve_1=0\bmod G^\sim$ and $\ve_2=0$. Depending on either $\ve_0=0$ or $\ve_0=1$ (which is replaced by the equivalent value $\ve_0=-1$) we obtain Cases~2.3 and 2.4, respectively. 

If $a_1=a_3\ne a_2$, scaling the first basis elements of the span allows us to set $a_3-a_2=1$. 
The parameter~$\ve_0$ should be nonzero since otherwise $f=0$. Therefore, $\ve_0=1\bmod G^\sim$.
The conditions $a_2=1$, $a_2=0$ and $a_2\ne0,1$ lead to Cases~2.5, 2.6 and 2.7, respectively. 
In the last case we denote $1-a_2$ by~$q$. This value should also be nonzero since otherwise the extension dimension is greater than two. 

Let $a_1\ne a_3$. Then $\ve_0=0\bmod G^\sim$ 
and by scaling the first basis elements of the span we can also set $a_1-a_3=1$. 
Introducing the notation $p=a_3-a_2$ and $q=a_1-2a_2$, we obtain that $a_1=2+2p-q$, $a_2=1+p-q$ and $a_3=1+2p-q$. 
The further partition depends on values of~$\ve_2$, $q$ and~$\ve_1$. 
For $\ve_2=0$ the dimension of extension is not greater than two only if 
either $q\ne0$ and then $\ve_1=0\bmod G^\sim$ (Case~2.8) or $q=0$ and $\ve_1\ne0$ and then $\ve_1=1\bmod G^\sim$ (Case~2.10). 
The condition $\ve_2=1$ implies that $q=p+1/2$. 
If additionally either $q=\ve_1=0$ or $q\ne0$ (and then $\ve_1=0\bmod G^\sim$), we have Case~2.9. 
Case~2.11 corresponds to the additional constraints $q=0$ and $\ve_1\ne0$ (i.e.\ $\ve_1=1\bmod G^\sim$).

Consider the candidates for three-dimensional appropriate extensions listed in Lemma~\ref{lem:3DimInequivExtsForIbragimovClass}. 
The compatibility of the associated systems in the arbitrary elements, supplemented with the inequality $f\ne0$, implies 
$p_1+p_2=1$ and $d=0$ for the first and the second span of Lemma~\ref{lem:3DimInequivExtsForIbragimovClass}, respectively. 
The general solutions of these systems up to $G^\sim$-equivalence are $(f,g)=(\pm |u_x|^{2p},0)$ and $(f,g)=(\pm e^{2u_x},0)$. 
This gives the following cases of Lie symmetry extensions:
\begin{gather*}\hspace*{-\arraycolsep}
\begin{array}{lll}
3.1. & (1+p)\DDD^u+p\DDD(x),\ (1+p)\DDD^t+\DDD(x),\ \DDD(1),\ p\ne-2,-1,0\colon & u_{tt}=\pm |u_x|^{2p}u_{xx},\\[.5ex]
3.2. & \DDD^u+\DDD(x)+\GG(x),\ \DDD^t-\GG(x),\ \DDD(1)\colon & u_{tt}=\pm e^{2u_x}u_{xx}.
\end{array}
\end{gather*}

Special cases of Lie symmetry extensions in class~\eqref{eq:IbragimovClass} are presented before this section. 
More precisely, all inequivalent equations whose maximal Lie invariance algebras are not contained in the projection of the equivalence algebra~$\mathfrak g^\sim$ 
to the variable space are listed in Lemma~\ref{lem:IbragimovClassSpecialLieSymExts1}. 
Equations from class~\eqref{eq:IbragimovClass} which are invariant with respect to two linearly independent operators of the form~$\mathrm P Q^i$, 
where $Q^i=\DDD(\varphi^i)+\GG(\psi^i)+c_i\FF^2$, are described in Corollary~\ref{cor:OnAppropriateSubalgebras1}.
For convenience, we collect the derived cases in a single table 
and formulate the final result of group classification in the class~\eqref{eq:IbragimovClass} as a theorem.
Recall that within the class~\eqref{eq:IbragimovClass} $G^\sim$-equivalence coincides with the general point equivalence, cf.\ 
Corollary~\ref{col:IbragimovClassEquiv}. 

\begin{theorem}\label{thm:IbragimovClassGroupClassification}
All  $G^\sim$-inequivalent (resp.\ point-inequivalent) cases of Lie symmetry extensions of the kernel algebra~$\mathfrak g^\cap$ in the class~\eqref{eq:IbragimovClass} 
are exhausted by cases presented in Table~\ref{tab:IbragimovClassExtensions}.
\end{theorem}

\begin{table}[htb]\footnotesize\renewcommand{\arraystretch}{1.39}\belowcaptionskip=.7ex\abovecaptionskip=.0ex
{\centering
\caption{Lie symmetry extensions of the kernel algebra~$\mathfrak g^\cap=\langle\p_t,\p_u,t\p_u\rangle$ for the class~\eqref{eq:IbragimovClass}\label{tab:IbragimovClassExtensions}}
\begin{tabular}{|r|l|l|l|}
\hline
N\hfil &\hfil $f$ &\hfil $g$ &\hfil Basis of extension \\
\hline
\multicolumn{4}{|c|}{One-dimensional extensions}\\
\hline
1  & $\tilde f(x-\ve\ln|u_x|)u_x^{-1}$          & $\tilde g(x-\ve\ln|u_x|)+2\ln|u_x|\!\!$           & $t\p_t+2\ve\p_x+2(u+t^2)\p_u$ \\
2  & $\tilde f(x-\ve\ln|u_x|)|u_x|^{2p}$        & $\tilde g(x-\ve\ln|u_x|)|u_x|^{2p}u_x$            & $-pt\p_t + \ve\p_x+u\p_u$ \\
3  & $\tilde f(u_x)e^{2x}$                      & $\tilde g(u_x)e^{2x}$                             & $t\p_t-\p_x$ \\
4  & $\tilde f(x)e^{2u_x}$                      & $\tilde g(x)e^{2u_x}$                             & $t\p_t - x\p_u$ \\
5  & $\tilde f(u_x) $                           & $\tilde g(u_x)+2\ve x$                            & $\p_x + \ve t^2\p_u$ \\
\hline                                                                                           
\multicolumn{4}{|c|}{Two-dimensional extensions}\\                                               
\hline                                                                                           
6  & $\delta u_x^{-4}$                          & $\tilde g(x)u_x^{-3}$                             & $t^2\p_t+tu\p_u$, $2t\p_t+u\p_u$ \\
7  & $\delta e^{2x}|u_x|^{2p}$, $p\ne0,-2$      & $\nu e^{2x}|u_x|^{2p}u_x$, $\nu(p+1)\ne\delta$    & $p\p_x-u\p_u$, $t\p_t-\p_x$ \\
8  & $\delta x^2e^{2u_x}$                       & $\nu xe^{2u_x}$, $\nu\ne\delta$                   & $x\p_x+u\p_u$, $t\p_t-x\p_u$ \\
9  & $\tilde f(u_x)$                            & $0$                                               & $\p_x$, $t\p_t+x\p_x+u\p_u$ \\
10 & $\delta$                                   & $e^{-u_x}$                                        & $\p_x$, $t\p_t+x\p_x+(u+x)\p_u$ \\
11 & $\delta e^{2u_x}$                          & $2u_x$                                            & $\p_x$, $t\p_t+2x\p_x+(2u+x+t^2)\p_u$ \\
12 & $\delta e^{2u_x}$                          & $e^{u_x}{+}2\ve_2x$, $\ve_2{\in}\{-1,1\}$         & $x\p_x+(u+x)\p_u$, $\p_x+\ve_2t^2\p_u$ \\
13 & $\delta e^{2u_x}$                          & $e^{qu_x}$, $q\ne0$                               & $\p_x$, $(1-q)t\p_t+(2-q)x\p_x+((2-q)u+x)\p_u$ \\
14 & $\delta |u_x|^{2p}$                        & $|u_x|^q$,\quad *)                                & $\p_x$, $(1{+}p{-}q)t\p_t+(1{+}2p{-}q)x\p_x+(2{+}2p{-}q)u\p_u$ \\
15 & $\delta |u_x|^{2p}$                        & $\ve|u_x|^{p+1/2}{+}2x$                           & $\p_x+t^2\p_u$, $t\p_t+(1{+}2p)x\p_x+(3{+}2p)u\p_u$ \\
16 & $\delta |u_x|^{2p}$                        & $2\ln|u_x|$                                       & $\p_x$, $(1{+}p)t\p_t+(1{+}2p)x\p_x+(2(1{+}p)u+t^2)\p_u$ \\
17 & $\delta u_x^{-1}$                          & $2\ln|u_x|+2x$                                    & $\p_x+t^2\p_u$, $t\p_t+2(u+t^2)\p_u$ \\
\hline                                                                                           
\multicolumn{4}{|c|}{Three-dimensional extensions}\\                                             
\hline                                                                                           
18 & $\delta u_x^{-4}$                          & $u_x^{-3}$                                        & $t^2\p_t+tu\p_u$, $2t\p_t+u\p_u$, $\p_x$ \\
19 & $\delta u_x^{-4}$                          & $\nu x^{-1}u_x^{-3}$, $\nu\ne0$                   & $t^2\p_t+tu\p_u$, $2t\p_t+u\p_u$, $2x\p_x+u\p_u$ \\
20 & $\delta |u_x|^{2p}$, $p\ne-2,0$            & $0$                                               & $\p_x$,  $t\p_t+x\p_x+u\p_u$, $pt\p_t-u\p_u$ \\
21 & $\delta e^{2u_x}$                          & $0$                                               & $\p_x$, $t\p_t+x\p_x+u\p_u$, $t\p_t-x\p_u$ \\
\hline                                                                                           
\multicolumn{4}{|c|}{Four-dimensional extensions}\\                                              
\hline                                                                                           
22 & $\delta u_x^{-4}$                          & $0$                                               & $t^2\p_t+tu\p_u$, $2t\p_t+u\p_u$, $\p_x$, $2x\p_x+u\p_u$ \\
\hline
\end{tabular}
\\[2ex]}
Here $\delta=\pm1\bmod G^\sim$ and $\ve\in\{0,1\}\bmod G^\sim$. In Case~15 $\ve=0\bmod G^\sim$ if $p=-1/2$.\\
*) $q\ne0$, $(p,q)\ne(-1,-1),(-2,-3)$ in Case~14. 
\end{table}

In each case of Table~\ref{tab:IbragimovClassExtensions}
we present only basis elements of the corresponding Lie invariance algebra 
that belong to the complement of the basis $\{\p_t,\p_u,t\p_u\}$ of~$\mathfrak g^\cap$.  
The spans of~$\mathfrak g^\cap$ and the vector fields given in cases 1--6 and 9 of Table~\ref{tab:IbragimovClassExtensions}
are the maximal Lie invariance algebra of the corresponding equations 
for the general values of the associated parameter-functions~$\tilde f$ and~$\tilde g$, 
but for certain values of these parameter-functions additional extensions are possible.  

In the course of collecting cases of Lie symmetry extensions into Table~\ref{tab:IbragimovClassExtensions}, 
they are properly arranged. 
In particular, Cases~1 and~2 of Corollary~\ref{cor:OnAppropriateSubalgebras1} are merged with Cases~2.10 and~3.1 
into Cases~16 and~20 of Table~\ref{tab:IbragimovClassExtensions}, respectively. 
As the value $p=-1$ is singular for the basis of Case~2.10, 
the bases of Case~2.10 and Case~2 of Corollary~\ref{cor:OnAppropriateSubalgebras1} are changed 
in order to be agreed. 
Case~2.6 with $\ve_2=0$ is not included in Case~12 of Table~\ref{tab:IbragimovClassExtensions} 
since it is united with Case~2.7 into Case~13 of this table. 

Within the algebraic approach used for group classification of the class~\eqref{eq:IbragimovClass} 
the construction of Lie invariance algebras precedes the construction of associated invariant equations. 
This is why the simplification of the form of bases of Lie symmetry extensions, in a certain sense, 
dominates in Table~\ref{tab:IbragimovClassExtensions}.
The form of invariant equations can be slightly simplified 
if simultaneous minor complication of bases of the corresponding Lie invariance algebras are permitted. 
In particular, multipliers equal to two can be removed from arbitrary elements 
by equivalence transformations or re-denoting the parameter~$p$. 

Note that the unique inequivalent case of Lie symmetry extension for which the corresponding Lie invariance algebra 
is of maximal dimension possible for equations the class~\eqref{eq:IbragimovClass} and equal to seven, Case~22, 
is not associated with a subalgebra of the equivalence algebra~$\mathfrak g^\sim$.

Now we discuss nontrivial constraints for constant parameters which are imposed by the maximality condition for the corresponding extensions. 

The equation $u_{tt}=e^{2x}|u_x|^{2p}(\delta u_{xx}+\nu u_x)$ corresponding to Case~7 for general values of parameters is linear if $p=0$. 
If $p\ne-1$, it is reduced by the transformation $\tilde t=|p+1|^{-p-1}t$, $\tilde x=e^{-x/(p+1)}$, $\tilde u=u$ to the equation
$\tilde u_{\tilde t\tilde t}=|\tilde u_{\tilde x}|^{2p}(\delta\tilde u_{\tilde x\tilde x}+\tilde\nu\tilde x^{-1}\tilde u_{\tilde x})$ with $\tilde\nu=\delta-\nu(p+1)$, 
which coincides with the equation of Case~19 (resp.\ 20, resp.\ 22) if $p\ne-2$ and $\tilde\nu\ne0$ (resp.\ $p=-2$ and $\tilde\nu\ne0$, resp.\ $p=-2$ and $\tilde\nu=0$).

The equation $u_{tt}=e^{2u_x}(\delta x^2 u_{xx}+\nu xu_x)$ corresponding to Case~8 is similar with respect to the transformation
$\tilde t=t$, $\tilde x=x$, $\tilde u=u+x\ln|x|-x$ to the equation
$\tilde u_{\tilde t\tilde t}=e^{2\tilde u_{\tilde x}}(\delta\tilde u_{\tilde x\tilde x}+(\nu-\delta)\tilde x^{-1})$
which coincides with the equation of Case~21 if $\nu=\delta$. 

Consider the subclass of the class~\eqref{eq:IbragimovClass} associated with the additional constraints $f_x=g_x=0$, 
i.e., the class of equations of the general form 
\begin{equation}\label{eq:GandariasTorrisiValentiClass}
u_{tt}=f(u_x)u_{xx}+g(u_x),
\end{equation}
where $(f_{u_x},g_{u_xu_x})\ne(0,0)$. \looseness=-1
Lie symmetries of these equations were comprehensively described in~\cite{gand04Ay} using no equivalence relation. 
The selection of Cases~$5|_{\ve=0}$, 9, 10, 11, 13, 14, 16, 18, 20, 21 and 22 from Table~\ref{tab:IbragimovClassExtensions}, 
which are related to equations from the subclass~\eqref{eq:GandariasTorrisiValentiClass}, represents 
the exhaustive list of Lie symmetry extensions in this subclass up to general point equivalence, 
where the algebra $\mathfrak g^\cap_1=\langle\p_t,\p_u,t\p_u,\p_x\rangle$ given in Case~$5|_{\ve=0}$ is the corresponding kernel algebra. 
Theorems~\ref{thm:IbragimovClassNormalizedSubclass} and~\ref{thm:IbragimovClassSemiNormalizedSubclass} imply
that the equivalence group~$G^\sim_2$ of the subclass~\eqref{eq:GandariasTorrisiValentiClass} is a subgroup of the equivalence group~$G^\sim$, 
namely, it consists of transformations of the form~\eqref{eq:EquivalenceGroupIbragimovClass} with $\varphi_{xx}=\psi_{xx}=0$.
As all necessary shifts, scalings and sign changes of the derivatives $u_{tt}$, $u_{xx}$ and $u_x$ are induced by transformations from the equivalence group, 
the majority of constants parameterizing elements of the classification list from~\cite{gand04Ay} 
can be set to appropriate values (0, 1, $\pm1$ or others). 
In other words, these constants are inessential from the point of view of symmetry analysis. 
As a result, the classification list from~\cite{gand04Ay} is reduced to the above selection of cases from Table~\ref{tab:IbragimovClassExtensions}, 
excluding a single case given in~\cite{gand04Ay} as Case~X. 
After simplifications by shifts, scalings and sign changes of derivatives induced by transformations from~$G^\sim_2$, 
the value of arbitrary elements for this case takes the form $f=\delta u_x^{-2}$ and $g=\delta u_x^{-1}$, where $\delta=\pm 1$.
A~value of the coefficient of~$u_x^{-1}$ in the expression for~$g$ is not essential. 
We set it equal to~$\delta$ for convenience. 
The related equation $u_{tt}=\delta u_x^{-2}u_{xx}+\delta u_x^{-1}$ 
has the six-dimensional maximal Lie invariance algebra $\langle\p_t,\p_u,t\p_u,\p_x,e^x\p_x,t\p_t+u\p_u\rangle$ and 
is reduced by the transformation $\mathcal T$: $\tilde t=t$, $\tilde x=e^x$, $\tilde u=u$ 
to the equation $\tilde u_{\tilde t\tilde t}=\delta\tilde u_{\tilde x}^{-2}\tilde u_{\tilde x\tilde x}$, 
which corresponds to Case~$20|_{p=-1}$ from Table~\ref{tab:IbragimovClassExtensions}. 
This is why we put $(p,q)\ne(-1,-1)$ as a parameter constraint for Case~14 from Table~\ref{tab:IbragimovClassExtensions}. 
The presence of one more inequivalent case in the course of the classification 
of the subclass~\eqref{eq:GandariasTorrisiValentiClass} up to $G^\sim_2$-equivalence instead of general point equivalence
is explained by the fact that the transformation $\mathcal T$ does not belong to~$G^\sim_2$. 
To complete the discussion of singular values of parameters for Case~14 from Table~\ref{tab:IbragimovClassExtensions}, 
we only note that the values $p=-2$ and $q=-3$ directly give Case~18 
and for the value $q=0$ the corresponding equation is reduced by the transformation $\tilde u=u-t^2/2$ to the equation associated with Case~20. 
The same transformation reduces Case~13 with~$q$ set to zero formally to Case~21.

Any equation from the class~\eqref{eq:IbragimovClass} is a potential equation for the equation of the form 
\begin{equation}\label{eq:NonlinTelegraphEq}
v_{tt}=(f(x,v)v_x+g(x,v))_x
\end{equation}
with the same value of the arbitrary elements~$f$ and~$g$, where the argument~$u_x$ is replaced by~$v$.
Indeed, Eq.~\eqref{eq:NonlinTelegraphEq} possesses two inequivalent characteristics 
of conservation laws, $\lambda^1=1$ and $\lambda^2=t$. 
The potential systems constructed with the simplest conserved vectors associated with these characteristics is 
\begin{gather}\label{eq:NonlinTelegraphEqPotSystem1}
w^1_x=v_t,\quad w^1_t=f(x,v)v_x+g(x,v),
\\\label{eq:NonlinTelegraphEqPotSystem2}
w^2_x=tv_t-v,\quad w^2_t=tf(x,v)v_x+tg(x,v).
\end{gather}
We denote $tw^1-w^2$ by~$u$. 
In terms of the dependent variables~$v$, $w^1$ and~$u$,
the joint potential system~\eqref{eq:NonlinTelegraphEqPotSystem1}, \eqref{eq:NonlinTelegraphEqPotSystem2} 
takes the form $u_x=v$, $u_t=w^1$, $w^1_t=f(x,v)v_x+g(x,v)$ 
which is a potential system for system~\eqref{eq:NonlinTelegraphEqPotSystem1}, i.e., 
it is formally a second-level potential system of Eq.~\eqref{eq:NonlinTelegraphEq}. 
Hence $u$ is a second-level potential for this equation. 
Excluding $v$ and $w^1$ from the last system, we obtain Eq.~\eqref{eq:IbragimovClass}. 
In order to derive Eq.~\eqref{eq:NonlinTelegraphEq} from Eq.~\eqref{eq:IbragimovClass}, 
we should totally differentiate Eq.~\eqref{eq:IbragimovClass} with respect to~$x$ and replace~$u_x$ by~$v$.
As the coefficients of any Lie symmetry operator~$Q=\tau\p_t+\xi\p_x+\eta\p_u$ of Eq.~\eqref{eq:IbragimovClass}
satisfy the determining equations $\tau_u=\xi_u=\eta_{uu}=\eta_{xu}=0$, 
the coefficient of~$\p_v$ in the prolongation of this operator to~$v$ according to the equality $v=u_x$ 
is equal to $\eta_x+(\eta_u-\xi_x)u_x$  and hence does not depend on~$u$.
Therefore, Lie symmetries of Eq.~\eqref{eq:IbragimovClass} induce 
no purely potential symmetries of Eq.~\eqref{eq:NonlinTelegraphEq}. 

We checked cases from Table~\ref{tab:IbragimovClassExtensions} using the package \texttt{DESOLV}~\cite{carm00Ay,vu07Ay} 
for symbolic calculations of Lie symmetries, whenever it was possible.

\section{Conclusion}\label{sec:ConclusionIbragimovClass}

Results of this paper and those existing in the literature on symmetry analysis of differential equations 
allow us to comparatively analyze different approaches to group classification of differential equations 
(partial preliminary group classification, 
complete preliminary group classification and
complete group classification) 
within the framework of the algebraic method. 
Given a class~$\mathcal L|_{\mathcal S}$ of (systems of) differential equations with the equivalence group~$G^\sim$ and the equivalence algebra~$\mathfrak g^\sim$,
the application of each of the above approaches involves, in some way, classification of certain subalgebras of~$\mathfrak g^\sim$. 
The essential point is what subalgebras of~$\mathfrak g^\sim$ should be classified and what equivalence relation should be used in the course of the classification.  

In the course of \emph{partial preliminary group classification}, a proper subalgebra~$\mathfrak s$ of~$\mathfrak g^\sim$ is fixed 
and then only subalgebras of~$\mathfrak s$ are classified. 
This approach may be relevant only if the subalgebra~$\mathfrak s$ is noticeable from the physical or another point of view. 
Hence the choice of such subalgebra should be strongly justified which, unfortunately, is often ignored in the existing literature on the subject. 
Differences in the consideration of the subalgebra~$\mathfrak s$ instead of the whole algebra~$\mathfrak g^\sim$ are especially significant in the case
when $\mathfrak g^\sim$ is an infinite-dimensional algebra whereas $\mathfrak s$ is a finite-dimensional subalgebra.
A seeming advantage of replacing~$\mathfrak g^\sim$ by~$\mathfrak s$ is 
that in general finite-dimensional algebras are much simpler objects than infinite-dimensional ones. 
At the same time, partial preliminary group classification has a few essential weaknesses most of which are related to the following fact:  
As the fixed subalgebra~$\mathfrak s$ of~$\mathfrak g^\sim$ is usually not invariant under the adjoint action of the equivalence group~$G^\sim$, 
this group does not generate a well-defined equivalence relation on subalgebras of~$\mathfrak s$. 
This is a reason why subalgebras of~$\mathfrak s$ are classified up to the weaker internal equivalence on~$\mathfrak s$, 
which is induced by the adjoint action of the continuous transformation group associated with~$\mathfrak s$, instead of $G^\sim$-equivalence. 

The exhaustive classification of subalgebras up to the internal equivalence is a cumbersome algebraic problem, 
possessing no algorithmic solution even for finite-dimensional algebras.
In order to simplify it, only one-dimensional subalgebras are usually classified 
which crucially increases incompleteness of results obtained in the framework of partial preliminary group classification. 
Although the number of classification cases remains quite large, 
many of them are inessential from the $G^\sim$-equivalence point of view, not to mention the general point equivalence. 
The presence of equivalent cases unnecessarily complicates both the solution of the group classification problem 
and further applications of classification results, e.g., for the construction of exact solutions of systems from the class~$\mathcal L|_{\mathcal S}$. 

\emph{Complete preliminary group classification} of the class~$\mathcal L|_{\mathcal S}$ is based on the classification of subalgebras 
of the entire equivalence algebra~$\mathfrak g^\sim$ up to $G^\sim$-equivalence. 
As both the objects, $\mathfrak g^\sim$ and~$G^\sim$, are directly related to the class~$\mathcal L|_{\mathcal S}$ and well consistent to each other, 
this approach looks as quite natural. 
For weakly normalized classes of differential equations, it gives an exhaustive classification list. 
Moreover, complete preliminary group classification always is a necessary step for complete group classification within the framework of the algebraic method.  
It is obvious that complete preliminary group classification gives a list which is closer to exhausting all possible Lie symmetry extensions 
than any list obtained via partial preliminary group classification. 
At the same time, due to the usage of $G^\sim$-equivalence which is stronger than the internal equivalence on a subalgebra of~$\mathfrak g^\sim$, 
the former list can contain even a less number of cases than the latter one. 
For example, 33 cases of one-dimensional extensions of the kernel algebra were constructed for the class~\eqref{eq:IbragimovClass} in~\cite{ibra91Ay}
in the course of partial preliminary group classification involving a ten-dimensional subalgebra of the equivalence algebra of this class. 
All these cases are $G^\sim$-equivalent to particular subcases of Cases~1--5 from Table~\ref{tab:IbragimovClassExtensions} of the present paper. 

The approach of complete preliminary group classification can be optimized via selecting of appropriate subalgebras of~$\mathfrak g^\sim$, 
i.e.\ subalgebras whose projections to the space of system variables are maximal Lie invariance algebras of systems from the class~$\mathcal L|_{\mathcal S}$. 
The simplest common property of appropriate subalgebras is that they contain the kernel algebra. 
Other criteria for selecting of appropriate subalgebras including bounds for dimensions of extensions or additional extensions 
are derived via examination of the determining equations for Lie symmetries of systems from the class~$\mathcal L|_{\mathcal S}$. 
In a certain sense, this means combining the algebraic method of group classification 
with the direct method based on the study of compatibility and the integration of the determining equations up to $G^\sim$-equivalence.
The usage of the optimized technique often allows one to reduce the classification problem to classification of certain low-dimensional subalgebras 
of the equivalence algebra, even if the equivalence algebra is infinite-dimensional and there exist infinite-dimensional extensions of the kernel. 
Related calculations are not too cumbersome. 
Thus, minimal computations which are necessary for complete preliminary group classification of class~\eqref{eq:IbragimovClass} are exhausted 
by the first parts of Sections~\ref{sec:EquivalenceAlgebraIbragimovClass} and~\ref{sec:DetEqsForLieSymsIbragrimovClass} 
and entire Sections~\ref{sec:PreliminaryStudyOfAdmTrans}, \ref{sec:EquivGroup}, 
\ref{sec:ClassificationSubalgebrasIbragimovClass} and~\ref{sec:GroupClassificationIbragimovClass}.
These computations result in the absolute majority of inequivalent cases of Lie symmetry extensions for the class~\eqref{eq:IbragimovClass}, 
which are presented in Table~\ref{tab:IbragimovClassExtensions} (the exceptions are only Cases~6, 18, 19 and 22). 

There exist two ways to apply the algebraic method to \emph{complete group classification}. 
The first way is to reduce complete group classification to preliminary one. 
The reduction can be realized, e.g., via proving that the class~$\mathcal L|_{\mathcal S}$ is weakly normalized 
or partitioning this class into weakly normalized subclasses and other subclasses which can be easily classified by the direct method. 
Although the partition into subclasses usually involves cumbersome and sophisticated computations, it is an effective tool of group analysis
since it accurately adapts the classification procedure to the structure of the class~$\mathcal L|_{\mathcal S}$.  
This is the way that has been used in the present paper. 
The class~\eqref{eq:IbragimovClass} is partitioned into two subclasses possessing the same equivalence group as the whole class~\eqref{eq:IbragimovClass}. 
One of the subclasses is normalized, the other is semi-normalized 
and mapped by equivalence transformations onto its subclass~\eqref{eq:IbragimovClassSpecialLieSymExts1Subclass} of simple structure. 
Group classification of the subclass~\eqref{eq:IbragimovClassSpecialLieSymExts1Subclass} has been obtained in the course of the partition 
which results in only four special cases of Lie symmetry extension (Cases~6, 18, 19 and 22 from Table~\ref{tab:IbragimovClassExtensions}), 
which are not related to subalgebras of~$\mathfrak g^\sim$.
The second way is to directly classify $G^\sim$-inequivalent appropriate algebras contained 
in the span $\mathfrak g^{\langle\rangle}=\langle\mathfrak g_\theta|\theta\in\mathcal S\rangle$ of maximal Lie invariance algebras, $\mathfrak g_\theta$, 
of all systems from the class~$\mathcal L|_{\mathcal S}$. 
This way properly works only if the class~$\mathcal L|_{\mathcal S}$ possesses certain properties, 
e.g., if the maximal Lie invariance algebra~$\mathfrak g_\theta$ is of low dimension for any $\theta\in\mathcal S$ \cite{khab09Ay} 
or if the class~$\mathcal L|_{\mathcal S}$ is at least weakly normalized or partitioned into weakly normalized subclasses 
but this property is not explicitly checked \cite{basa01Ay,lahn06Ay,lahn07Ay,zhda99Ay}.
An explanation for the above observation is that $G^\sim$-equivalence is not appropriate in the course of  
classification of subalgebras contained in~$\mathfrak g^{\langle\rangle}$ 
if $\mathfrak g^{\langle\rangle}$ is strongly inconsistent with the equivalence algebra~$\mathfrak g^\sim$ 
(e.g., much wider than the projection~$\mathrm P\mathfrak g^\sim$ of~$\mathfrak g^\sim$ to the space of system variables). 

Due to the above partition of the class~\eqref{eq:IbragimovClass} we have obtained essentially stronger results 
than the solution of the usual group classification problem by Lie--Ovsiannikov for this class. 
The partition exhaustively describes the set of admissible transformations in the class~\eqref{eq:IbragimovClass}. 
Moreover, the fact that the whole class~\eqref{eq:IbragimovClass} is semi-normalized guarantees that 
there are no additional point equivalence transformations 
between cases of Lie symmetry extensions presented in Table~\ref{tab:IbragimovClassExtensions}, 
i.e., the same table gives the complete group classification of the class~\eqref{eq:IbragimovClass} 
with respect to general point equivalence. 

The extension and clarification of the group classification toolbox is by no means a pure mathematical problem. 
Methods from symmetry analysis including group classification have the potential to provide, in particular, 
a unifying framework to construct invariant local closure or parameterization schemes 
for averaged nonlinear differential equations~\cite{ober97Ay,popo10Cy,raza07By}. 
As finding appropriate closure ansatzes for averaged differential equations is 
at the basis of any numerical model of (geophysical) fluid dynamical systems, 
it is immediately clear that group classification can play 
a crucial role in the construction of different computational codes for such systems.  
The classes of differential equations arising in the course of the parameterization problem 
are usually much wider and have more complicated structure 
than the classes studied in conventional group classification. 
It generally cannot be expected to completely solve the group classification problems 
for such classes using existing methods. 
Hence the development of new tools for group classification of differential equations 
simultaneously with the improvement of well-known approaches still remains an attractive and challenging problem.  
Especially for the above complex classification problems, 
the whole framework of the algebraic method as described and extended in the present paper 
including the proposed algebraic method of finding equivalence groups seems to be most appealing.

\section*{Acknowledgements}

This research was supported by the Austrian Science Fund (FWF), projects P20632 (EDSCB and ROP) and P21335 (AB).


{\footnotesize
\begin{thebibliography}{10}
\providecommand{\url}[1]{\texttt{#1}}
\providecommand{\urlprefix}{URL }
\expandafter\ifx\csname urlstyle\endcsname\relax
  \providecommand{\doi}[1]{doi:\discretionary{}{}{}#1}\else
  \providecommand{\doi}{doi:\discretionary{}{}{}\begingroup
  \urlstyle{rm}\Url}\fi
\providecommand{\eprint}[2][]{\url{#2}}

\bibitem{akha91Ay}
Akhatov I.S., Gazizov R.K. and Ibragimov N.K., {Nonlocal symmetries. Heuristic
  approach}, \emph{J. Math. Sci.} \textbf{55} (1991), 1401--1450.

\bibitem{ande96Ay}
Anderson R.L., Baikov V.A., Gazizov R.K., Hereman W., Ibragimov N.H., Mahomed
  F.M., Meleshko S.V., Nucci M.C., Olver P.J., Sheftel' M.B., Turbiner A.V. and
  Vorob'ev E.M., \emph{{CRC handbook of Lie group analysis of differential
  equations. Vol. 3. New trends in theoretical developments and computational
  methods}}, CRC Press, Boca Raton, 1996.

\bibitem{basa01Ay}
Basarab-Horwath P., Lahno V. and Zhdanov R., {The structure of Lie algebras and
  the classification problem for partial differential equations}, \emph{Acta
  Appl. Math.} \textbf{69} (2001), 43--94.

\bibitem{bihl11By}
Bihlo A. and Popovych R.O., Lie symmetry analysis and exact solutions of the
  quasi-geostrophic two-layer problem, \emph{J. Math. Phys.} \textbf{52}
  (2011), 033103 (24 pages).

\bibitem{bihl11Cy}
Bihlo A. and Popovych R.O., Point symmetry group of the barotropic vorticity
  equation, in \emph{Proceedings of 5th Workshop {"}Group Analysis of
  Differential Equations \& Integrable Systems{"} (June 6-10, 2010, Protaras,
  Cyprus)}, 2011 pp. 15--27.

\bibitem{blum89Ay}
Bluman G. and Kumei S., \emph{{Symmetries and differential equations}},
  Springer, New York, 1989.

\bibitem{blum10Ay}
Bluman G.W., Cheviakov A.F. and Anco S.C., \emph{Application of symmetry
  methods to partial differential equations}, Springer, New York, 2010.

\bibitem{boro04Ay}
Borovskikh A.V., Group classification of the eikonal equations for a
  three-dimensional nonhomogeneous medium, \emph{Mat. Sb.} \textbf{195} (2004),
  23--64, in Russian; translation in {\it Sb. Math.}, 195, (2004), no. 3--4,
  479--520.

\bibitem{boro06Ay}
Borovskikh A.V., The two-dimensional eikonal equation, \emph{Siberian Math. J.}
  \textbf{47} (2006), 813--834.

\bibitem{carm00Ay}
Carminati J. and Vu K., Symbolic computation and differential equations: {L}ie
  symmetries, \emph{J. Symb. Comput.} \textbf{29} (2000), 95--116.

\bibitem{card11Ay}
Dos Santos Cardoso-Bihlo E.M., Bihlo A. and Popovych R.O., {Enhanced
  preliminary group classification of a class of generalized diffusion
  equations}, \emph{Commun. Nonlinear Sci. Numer. Simulat.} \textbf{16} (2011),
  3622--3638.

\bibitem{gand04Ay}
Gandarias M.L., Torrisi M. and Valenti A., Symmetry classification and optimal
  systems of a non-linear wave equation, \emph{Internat. J. Non-Linear Mech.}
  \textbf{39} (2004), 389--398.

\bibitem{hari93Ay}
Harin A.O., On a countabe-dimensional subalgebra of the equivalence algebra for
  equations $v_{tt}= f (x, v_x) v_{xx}+ g (x, v_x)$, \emph{J. Math. Phys.}
  \textbf{34} (1993), 3676--3682.

\bibitem{head93Ay}
Head A.K., {LIE}, a {PC} program for {L}ie analysis of differential equations,
  \emph{Comput. Phys. Comm.} \textbf{77} (1993), 241--248, (See also
  http://www.cmst.csiro.au/LIE/LIE.htm).

\bibitem{huan07Ay}
Huang D. and Ivanova N.M., Group analysis and exact solutions of a class of
  variable coefficient nonlinear telegraph equations, \emph{J. Math. Phys.}
  \textbf{48} (2007), 073507, 23 pp.

\bibitem{hydo00By}
Hydon P.E., How to construct the discrete symmetries of partial differential
  equations, \emph{Eur. J. Appl. Math.} \textbf{11} (2000), 515--527.

\bibitem{ibra00Ay}
Ibragimov N.H. and Khabirov S.V., Contact transformation group classification
  of nonlinear wave equations, \emph{Nonlin. Dyn.} \textbf{22} (2000), 61--71.

\bibitem{ibra91Ay}
Ibragimov N.H., Torrisi M. and Valenti A., {Preliminary group classification of
  equations $v_{tt}=f(x,v_x)v_{xx}+g(x,v_x)$}, \emph{J. Math. Phys.}
  \textbf{32} (1991), 2988--2995.

\bibitem{ibra04Ay}
Ibragimov N.H., Torrisi M. and Valenti A., Differential invariants of nonlinear
  equations $v_{tt}= f (x, v_x) v_{xx}+ g (x, v_x)$, \emph{Commun. Nonlinear
  Sci. Numer. Simul.} \textbf{9} (2004), 69--80.

\bibitem{ivan10Ay}
Ivanova N.M., Popovych R.O. and Sophocleous C., {Group analysis of variable
  coefficient diffusion-convection equations. I. Enhanced group
  classification}, \emph{Lobachevskii J. Math.} \textbf{31} (2010), 100--122.

\bibitem{jeff82Ay}
Jeffrey A., Acceleration wave propagation in hyperelastic rods of variable
  cross-section, \emph{Wave Motion} \textbf{4} (1982), 173--180.

\bibitem{khab09Ay}
Khabirov S.V., A property of the determining equations for an algebra in the
  group classification problem for wave equations, \emph{Sibirsk. Mat. Zh.}
  \textbf{50} (2009), 647--668, in Russian; translation in {\it Sib. Math. J.},
  50:515--532, 2009.

\bibitem{king91Ay}
Kingston J.G. and Sophocleous C., {On point transformations of a generalised
  Burgers equation}, \emph{Phys. Lett. A} \textbf{155} (1991), 15--19.

\bibitem{king98Ay}
Kingston J.G. and Sophocleous C., {On form-preserving point transformations of
  partial differential equations}, \emph{J. Phys. A} \textbf{31} (1998),
  1597--1619.

\bibitem{king01Ay}
Kingston J.G. and Sophocleous C., {Symmetries and form-preserving
  transformations of one-dimensional wave equations with dissipation},
  \emph{Int. J. Non-Lin. Mech.} \textbf{36} (2001), 987--997.

\bibitem{lagn02Ay}
Lagno V.I. and Samoilenko A.M., {Group Classification of nonlinear evolution
  equations: I. Invariance under semisimple local transformation groups},
  \emph{Differ. Equ.} \textbf{38} (2002), 384--391.

\bibitem{lahn06Ay}
Lahno V., Zhdanov R. and Magda O., {Group classification and exact solutions of
  nonlinear wave equations}, \emph{Acta Appl. Math.} \textbf{91} (2006),
  253--313.

\bibitem{lahn07Ay}
Lahno V.I. and Spichak S.V., {Group classification of quasilinear elliptic-type
  equations. I. Invariance with respect to Lie algebras with nontrivial Levi
  decomposition}, \emph{Ukrainian Math. J.} \textbf{59} (2007), 1719--1736.

\bibitem{lie81Ay}
Lie S., {\"Uber die Integration durch bestimmte Integrale von einer Klasse
  linearer partieller Differentialgleichungen}, \emph{Arch. for Math.}
  \textbf{6} (1881), 328--368, (Translation by N.H. Ibragimov: S. Lie, On
  Integration of a Class of Linear Partial Differential Equations by Means of
  Definite Integrals, {\it CRC Handbook of Lie Group Analysis of Differential
  Equations}, 2:473--508, 1994).

\bibitem{lie91Ay}
Lie S., \emph{{Vorlesungen \"uber Differentialgleichungen mit bekannten
  infinitesimalen Transformationen}}, B.G. Teubner, Leipzig, 1891.

\bibitem{lisl92Ay}
Lisle I.G., \emph{{Equivalence transformations for classes of differential
  equations}}, Ph.D. thesis, University of British Columbia, 1992.

\bibitem{mele94Ay}
Meleshko S.V., Group classification of equations of two-dimensional gas
  motions, \emph{Prikl. Mat. Mekh.} \textbf{58} (1994), 56--62, in Russian;
  translation in {\it J.~Appl. Math. Mech.}, 58:629--635.

\bibitem{ober97Ay}
Oberlack M., \emph{Invariant modeling in large-eddy simulation of turbulence},
  \emph{in: Annual research briefs}, Stanford University, 1997.

\bibitem{olve86Ay}
Olver P.J., \emph{{Application of Lie groups to differential equations}},
  Springer, New York, 2000.

\bibitem{oron86Ay}
Oron A. and Rosenau P., Some symmetries of the nonlinear heat and wave
  equations, \emph{Phys. Lett. A} \textbf{118} (1986), 172--176.

\bibitem{ovsi82Ay}
Ovsiannikov L.V., \emph{Group analysis of differential equations}, Acad. Press,
  New York, 1982.

\bibitem{ovsj62Ay}
Ovsjannikov L.V., \emph{Gruppovye svoistva differentsialnykh uravnenii.},
  Izdat. Sibirsk. Otdel. Akad. Nauk SSSR, Novosibirsk, 1962.

\bibitem{ovsi75Ay}
Ovsjannikov L.V. and Ibragimov N.H., {Group analysis of the differential
  equations of mechanics}, in \emph{General mechanics}, vol.~2, Moscow, pp.
  5--52, 1975. In Russian.

\bibitem{popo06Ay}
Popovych R.O., {Classification of admissible transformations of differential
  equations}, in \emph{Collection of Works of Institute of Mathematics},
  vol.~3, Kyiv, pp. 239--254, 2006.

\bibitem{popo10Cy}
Popovych R.O. and Bihlo A., {Symmetry preserving parameterization schemes},
  arXiv: 1010.3010v2, 36 pp., 2010.

\bibitem{popo03Ay}
Popovych R.O., Boyko V.M., Nesterenko M.O. and Lutfullin M.W., {Realizations of
  real low-dimensional Lie algebras}, \emph{J. Phys. A} \textbf{36} (2003),
  7337--7360, see arXiv:math-ph/0301029v7 for an extended and revised version.

\bibitem{popo04Ay}
Popovych R.O. and Ivanova N.M., New results on group classification of
  nonlinear diffusion--convection equations, \emph{J. Phys. A} \textbf{37}
  (2004), 7547--7565.

\bibitem{popo04By}
Popovych R.O., Ivanova N.M. and Eshraghi H., {Group classification of
  (1+1)-dimensional Schr\"odinger equations with potentials and power
  nonlinearities}, \emph{J. Math. Phys.} \textbf{45} (2004), 3049--3057.

\bibitem{popo10Ay}
Popovych R.O., Kunzinger M. and Eshraghi H., {Admissible transformations and
  normalized classes of nonlinear Schr\"odinger equations}, \emph{Acta Appl.
  Math.} \textbf{109} (2010), 315--359.

\bibitem{popo08Ay}
Popovych R.O., Kunzinger M. and Ivanova N.M., Conservation laws and potential
  symmetries of linear parabolic equations, \emph{Acta Appl. Math.}
  \textbf{100} (2008), 113--185.

\bibitem{popo10By}
Popovych R.O. and Vaneeva O.O., {More common errors in finding exact solutions
  of nonlinear differential equations: Part I}, \emph{Commun. Nonlinear Sci.
  Numer. Simul.} \textbf{15} (2010), 3887--3899, arXiv:0911.1848v2.

\bibitem{raza07By}
Razafindralandy D., Hamdouni A. and Oberlack M., {Analysis and development of
  subgrid turbulence models preserving the symmetry properties of the
  Navier--Stokes equations}, \emph{Eur. J. Mech. B/Fluids} \textbf{26} (2007),
  531--550.

\bibitem{song09Ay}
Song L. and Zhang H., {Preliminary group classification for the nonlinear wave
  equation $u_{tt} = f(x,u)u_{xx}+g(x,u)$}, \emph{Nonlinear Anal.} \textbf{70}
  (2009), 3512--3521.

\bibitem{vane07Ay}
Vaneeva O.O., Johnpillai A.G., Popovych R.O. and Sophocleous C., Enhanced group
  analysis and conservation laws of variable coefficient reaction-diffusion
  equations with power nonlinearities, \emph{J. Math. Anal. Appl.} \textbf{330}
  (2007), 1363--1386.

\bibitem{vane09Ay}
Vaneeva O.O., Popovych R.O. and Sophocleous C., {Enhanced group analysis and
  exact solutions of variable coefficient semilinear diffusion equations with a
  power source}, \emph{Acta Appl. Math.} \textbf{106} (2009), 1--46.

\bibitem{vu07Ay}
Vu K.T., Butcher J. and Carminati J., {Similarity solutions of partial
  differential equations using DESOLV}, \emph{Comput. Phys. Comm.} \textbf{176}
  (2007), 682--693.

\bibitem{wint92Ay}
Winternitz J.P. and Gazeau J.P., {Allowed transformations and symmetry classes
  of variable coefficient Korteweg-de Vries equations}, \emph{Phys. Lett. A}
  \textbf{167} (1992), 246--250.

\bibitem{witt04Ay}
Wittkopf A., \emph{Algorithms and implementations for differential
  elimination}, Ph.D. thesis, Simon Fraser University Burnaby, BC, Canada,
  2004.

\bibitem{zhda99Ay}
Zhdanov R.Z. and Lahno V.I., {Group classification of heat conductivity
  equations with a nonlinear source}, \emph{J. Phys. A} \textbf{32} (1999),
  7405--7418.

\end{thebibliography}

}
\end{document}